\documentclass[a4paper,11pt]{article}
\usepackage[margin=1in]{geometry}

\usepackage{amsmath}
\usepackage{amssymb}

\usepackage{amsthm}

\usepackage{amsfonts}
\usepackage{stmaryrd} 
\usepackage{bm}
\usepackage{mathtools}
\usepackage{xspace}
\usepackage{graphicx}

\usepackage[hidelinks]{hyperref}

\newcommand{\numberstates}[1]{\mathsf{nSt}(#1)}
\newcommand{\numbertransitions}[1]{\mathsf{nTr}(#1)}

\newcommand{\subsetlayers}{\mathcal{S}}
\newcommand{\subsetlayersstar}{\subsetlayers^{\circledast}}
\newcommand{\subsetlayersdefiningset}[1]{\subsetlayers^{\circ #1}}

\newcommand{\automatonnumberstates}{n}

\newcommand{\emptystring}{\lambda}
\newcommand{\orderalphabet}{<}
\newcommand{\orderstrings}{\prec}

\newcommand{\porderfrontierName}[1]{\chi_{#1}}
\newcommand{\porderfrontier}[2]{\porderfrontierName{#1}(#2)}

\newcommand{\permlayer}[3]{\ensuremath{\langle{#2#1#3}\rangle}}

\newcommand{\nstates}{\mathsf{nSt}}
\newcommand{\ntransitions}{\mathsf{nTr}}
\newcommand{\initialflag}{\iota}
\newcommand{\finalflag}{\phi}
\newcommand{\leftfrontier}{\ell}
\newcommand{\rightfrontier}{r}

\newcommand{\transduction}{\mathfrak{t}}
\newcommand{\atransduction}{\transduction}
\newcommand{\btransduction}{\transduction'}

\newcommand{\composition}{\circ}

\newcommand{\mergingleft}{\mu}
\newcommand{\mergingright}{\nu}
\newcommand{\collapsing}[2]{\zeta[#1,#2]}
\newcommand{\oddcollapsingname}[2]{\bm{\zeta}[#1,#2]}
\newcommand{\oddcollapsing}[3]{\oddcollapsingname{#1}{#2}(#3)}

\newcommand{\reachabilityleftname}{\vartheta}
\newcommand{\reachabilityrightname}{\eta}
\newcommand{\reachabilityleft}[2][]{\reachabilityleftname_{#1}(#2)}
\newcommand{\reachabilityright}[2][]{\reachabilityrightname_{#1}(#2)}

\newcommand{\removingname}[2]{\xi[#1,#2]}
\newcommand{\removing}[3]{\removingname{#1}{#2}(#3)}
\newcommand{\oddremovingname}[2]{\bm{\xi}[#1,#2]}
\newcommand{\oddremoving}[3]{\oddremovingname{#1}{#2}(#3)}

\newcommand{\multimaptransductionname}{\mathfrak{mm}}
\newcommand{\multimaptransduction}[1]{\multimaptransductionname[#1]}

\providecommand{\keywords}[1]
{\noindent{\small	
\textbf{\textit{Keywords:}} #1
}}

\newtheorem{theorem}{Theorem}

\newtheorem{lemma}[theorem]{Lemma}
\newtheorem{definition}[theorem]{Definition}
\newtheorem{corollary}[theorem]{Corollary}
\newtheorem{proposition}[theorem]{\textbf{Proposition}}
\newtheorem{observation}[theorem]{Observation}

\newcommand{\powerset}{\mathcal{P}}
\newcommand{\symmetricgroup}[1]{\mathbb{S}_{#1}}

\newcommand{\canonizationfunction}{\mathcal{C}}

\newtheoremstyle{named}{}{}{\itshape}{}{\bfseries}{.}{.5em}{Restatement of #1 \thmnote{#3}}
\theoremstyle{named}

\newcommand{\alphabet}{\ensuremath{\Sigma}}
\newcommand{\aalphabet}{\alphabet_1}
\newcommand{\balphabet}{\alphabet_2}
\newcommand{\calphabet}{\alphabet_3}

\newcommand{\alayer}{\ensuremath{B}}
\newcommand{\blayer}{\ensuremath{\alayer^{\prime}}}
\newcommand{\layer}{\alayer}
\newcommand{\odd}{\ensuremath{D}}
\newcommand{\aodd}{\ensuremath{\odd}}
\newcommand{\bodd}{\ensuremath{\odd^{\prime}}}

\newcommand{\layertransitions}{\ensuremath{T}}
\newcommand{\layerrightfrontier}{\ensuremath{r}}
\newcommand{\layerleftfrontier}{\ensuremath{\ell}}
\newcommand{\layerinitialstates}{\ensuremath{I}}
\newcommand{\layerfinalstates}{\ensuremath{F}}
\newcommand{\layerinitialflag}{\ensuremath{\iota}}
\newcommand{\layerfinalflag}{\ensuremath{\phi}}

\newcommand{\lang}{\ensuremath{L}}
\newcommand{\oddlang}[1]{\ensuremath{\mathcal{L}(#1)}}
\newcommand{\oddlanghigher}[2]{\ensuremath{\mathcal{L}_{#2}(#1)}}
\newcommand{\automatonlang}[1]{\ensuremath{{\mathcal{L}}(#1)}}
\newcommand{\automatonlanghigher}[2]{\ensuremath{{\mathcal{L}}_{#2}(#1)}}
\newcommand{\finiteautomaton}{\ensuremath{\mathcal{F}}}
\newcommand{\automatonstates}{\ensuremath{Q}}
\newcommand{\automatontransitions}{\ensuremath{T}}
\newcommand{\automatoninitialstates}{\ensuremath{I}}
\newcommand{\automatonfinalstates}{\ensuremath{F}}
\newcommand{\N}{\ensuremath{\mathbb{N}}}
\newcommand{\pN}{\ensuremath{\N_+}}

\newcommand{\relation}{\ensuremath{R}}
\newcommand{\arelation}{\relation}
\newcommand{\brelation}{\ensuremath{\relation^{\prime}}}
\newcommand{\layeralphabet}[2]{\ensuremath{\mathcal{B}(#1, #2)}}

\newcommand{\cdlayeralphabet}[2]{\ensuremath{\widehat{\mathcal{B}}(#1, #2)}}
\newcommand{\isomorphismsequence}{\overline{\pi}}

\newcommand{\width}{\ensuremath{w}}
\newcommand{\arity}{\ensuremath{\mathfrak{a}}}
\newcommand{\amap}{\ensuremath{g}}

\newcommand{\permutation}{\ensuremath{\pi}}

\DeclarePairedDelimiter{\abs}{\lvert}{\rvert}
\DeclarePairedDelimiter{\tuple}{(}{)}
\DeclarePairedDelimiter{\set}{\lbrace}{\rbrace}
\DeclarePairedDelimiter{\sequence}{\langle}{\rangle}
\DeclarePairedDelimiter{\bbset}{[}{]}
\DeclarePairedDelimiter{\dbset}{\llbracket}{\rrbracket}

\newcommand{\bset}[1]{\ensuremath{\bbset{#1}}}
\newcommand{\setst}{\xspace\colon\xspace}
\newcommand{\BigOh}{\ensuremath{\mathcal{O}}}
\newcommand{\astring}{\ensuremath{s}}
\newcommand{\bstring}{\ensuremath{u}}
\newcommand{\cstring}{\ensuremath{v}}
\newcommand{\cartesianproduct}{\times}
\newcommand{\tensorproduct}{\otimes}

\newcommand{\fwidth}{\ensuremath{\mathsf{\width}}}
\newcommand{\astate}{\ensuremath{q}}
\newcommand{\bstate}{\ensuremath{\ensuremath{q}^{\prime}}}
\newcommand{\cstate}{\ensuremath{p}}
\newcommand{\layerleftstate}{\ensuremath{\mathfrak{p}}}
\newcommand{\layerrightstate}{\ensuremath{\mathfrak{q}}}
\newcommand{\xset}{X}
\newcommand{\anumber}{\ensuremath{c}}
\newcommand{\oddlength}{\ensuremath{k}}
\newcommand{\length}{\ensuremath{k}}
\newcommand{\stringlength}{\ensuremath{k}}
\newcommand{\sset}{\ensuremath{X}}

\newcommand{\yset}{\ensuremath{Y}}

\newcommand{\asymbol}{\ensuremath{\sigma}}
\newcommand{\bsymbol}{\ensuremath{\tau}}

\newcommand{\pbsymbol}{\ensuremath{\bsymbol^{\prime}}}

\newcommand{\ie}{{i.e.}\xspace}

\newcommand{\true}{1}
\newcommand{\false}{0}

\newcommand{\defeq}{\doteq}

\newcommand{\pwset}[1]{\ensuremath{\mathcal{P}(#1)}}
\newcommand{\pwbijection}{\ensuremath{\Omega}}
\newcommand{\layerreachablestates}[3]{\ensuremath{\mathrm{\bf N}(#1,#2, #3)}}

\newcommand{\lpwmappingname}{\ensuremath{\mathsf{pw}}}
\newcommand{\lpwmapping}[1]{\ensuremath{\lpwmappingname(#1)}}

\newcommand{\compTransduction}[1]{\mathfrak{cp}[#1]} 
\newcommand{\lex}{\mathsf{lex}}

\newcommand{\mergingalphabet}[2]{\mathcal{M}(#1,#2)}

\newcommand{\reachabilityalphabet}[2]{\mathcal{R}(#1,#2)}

\newcommand{\agrowthconstant}{\alpha}
\newcommand{\bgrowthconstant}{\beta}
\newcommand{\cgrowthconstant}{\gamma}

\newcommand{\transductionlang}{\mathcal{L}}

\newcommand{\odddefiningset}[3]{\ensuremath{\layeralphabet{#1}{#2}^{\protect\circ{#3}}}}
\newcommand{\cdodddefiningset}[3]{\ensuremath{\cdlayeralphabet{#1}{#2}^{\protect\circ{#3}}}}
\newcommand{\odddefiningsetstar}[2]{\ensuremath{\layeralphabet{#1}{#2}^{\protect\circledast}}}
\newcommand{\cdodddefiningsetstar}[2]{\ensuremath{\cdlayeralphabet{#1}{#2}^{\protect\circledast}}}
\newcommand{\isomorphism}{\pi}

\usepackage{todonotes}

\newcommand{\secondlanguage}{\mathcal{S}}
\newcommand{\vertexset}{\ensuremath{V}}
\newcommand{\edgeset}{\ensuremath{E}}

\newcommand{\domain}{\mathsf{Dom}}
\newcommand{\image}{\mathsf{Im}}

\newcommand{\restr}[1]{\ensuremath{|_{#1}}}

\newcommand{\canonizationtransductionname}{\mathfrak{can}}
\newcommand{\canonizationtransduction}[2]{\canonizationtransductionname[#1,#2]}

\newcommand{\cdcanonizationtransductionname}{\widehat{\mathfrak{can}}}
\newcommand{\cdcanonizationtransduction}[2]{\cdcanonizationtransductionname[#1,#2]}

\newcommand{\determinizationtransductionname}{\mathfrak{det}}
\newcommand{\reachabilitytransductionname}{\mathfrak{rea}}
\newcommand{\reachabilitytransduction}[2]{\reachabilitytransductionname[#1,#2]}
\newcommand{\determinizationtransduction}[2]{\determinizationtransductionname[#1,#2]}
\newcommand{\mergingtransductionname}{\mathfrak{mer}}
\newcommand{\mergingtransduction}[2]{\mergingtransductionname[#1,#2]}
\newcommand{\normalizationtransductionname}{\mathfrak{nor}}
\newcommand{\normalizationtransduction}[2]{\normalizationtransductionname[#1,#2]}
\newcommand{\lengthfunction}[1]{\mathsf{len}(#1)}

\newcommand{\normalizingrelation}[2]{\mathrm{NR}[#1,#2]}
\newcommand{\normalizingcompatibility}[2]{\mathrm{NC}[#1,#2]}

\newcommand{\mergingrelationname}{\mathrm{MR}}
\newcommand{\mergingrelation}[2]{\mergingrelationname[#1,#2]}
\newcommand{\mergingcompatiblename}{\mathrm{MC}}
\newcommand{\mergingcompatible}[2]{\mergingcompatiblename[#1,#2]}

\newcommand{\reachabilityrelationname}{\mathrm{RR}}
\newcommand{\reachabilityrelation}[2]{\reachabilityrelationname[#1,#2]}
\newcommand{\reachabilitycompatibilityname}{\mathrm{RC}}
\newcommand{\reachabilitycompatibility}[2]{\reachabilitycompatibilityname[#1,#2]}

\newcommand{\duplicate}[1]{\mathfrak{d}(#1)}

\newcommand{\functionfromODD}{f}

\newcommand{\asymptoticfunction}{h}

\newcommand{\newhypercubegraph}{H}

\newcommand{\hypercubelanguage}{\hat{H}}
\newcommand{\hypercubeautomaton}{\mathcal{H}}

\newcommand{\evenautomaton}{\ensuremath{\finiteautomaton}}
\newcommand{\evenlanguage}{\ensuremath{\mathsf{Even}}}

\newcommand{\alldetlanguages}[2]{\mathsf{Det}(#1,#2)}
\newcommand{\boundedcomplement}[2]{\overline{#1}^{#2}}

\newcommand{\firstcap}[2]{\ensuremath{\mathsf{intersec}_{1}(#1,#2)}}
\newcommand{\firstcup}[2]{\ensuremath{\mathsf{union}_{1}(#1,#2)}}

\newcommand{\firstcomp}[1]{\ensuremath{\mathsf{compl}_{1}(#1)}}

\newcommand{\secondcap}[2]{\ensuremath{\mathsf{intersec}_{2}(#1,#2)}}
\newcommand{\secondcup}[2]{\ensuremath{\mathsf{union}_{2}(#1,#2)}}
\newcommand{\seconddiff}[2]{\ensuremath{\mathsf{diff}_{2}(#1,#2)}}
\newcommand{\secondcomp}[2]{\ensuremath{\mathsf{compl}_{2}(#1,#2)}}

\newcommand{\lstateoddlang}[3]{\ensuremath{\mathcal{L}(#1,#2,#3)}}
\newcommand{\slayerrightstate}[3]{\layerrightstate_{[#1,#2,#3]}}
\newcommand{\layerequivclass}[3]{[#1,#2,#3]}

\newcommand{\layerssymmericgroup}[2]{\mathcal{S}[\alphabet,\width]}
\newcommand{\relabelingmap}[2]{\eta[#1,#2]}

\begin{document}

\title{Second-Order Finite Automata\footnote{An extended abstract of this work corresponding to an invited talk at CSR 2020 appeared at \cite{MeloOliveira2020CSR}.}}
%
%
\author{
	Alexsander Andrade de Melo$^\dagger$ \hspace{1cm} Mateus de Oliveira Oliveira$^\ddagger$\\
\\
$^\dagger$Federal University of Rio de Janeiro, Rio de Janeiro, Brazil \\
\texttt{aamelo@cos.ufrj.br}
\\
$^\ddagger$University of Bergen, Bergen, Norway \\ 
\texttt{mateus.oliveira@uib.no}\\
}

\maketitle              



\begin{abstract} Traditionally, finite automata theory has been used as a
framework for the representation of possibly infinite sets of strings. In this
work, we introduce the notion of second-order finite automata, a formalism that
combines finite automata with ordered decision diagrams, with the aim of
representing possibly infinite {\em sets of sets} of strings. Our main result
states that second-order finite automata can be canonized with 
respect to the second-order languages they represent. Using this canonization 
result, we show
that sets of sets of strings represented by second-order finite automata are
closed under the usual Boolean operations, such as union, intersection,
difference and even under a suitable notion of complementation. Additionally,
emptiness of intersection and inclusion are decidable. 

We provide two algorithmic applications for second-order automata. First, we
show that several width/size minimization
problems for deterministic and nondeterministic ODDs are solvable in
fixed-parameter tractable time when parameterized by the width of the input
ODD. In particular, our results imply FPT algorithms for corresponding
width/size minimization problems for ordered binary decision diagrams (OBDDs) with a
fixed variable ordering. Previously, only algorithms that take exponential time
in the size of the input OBDD were known for width minimization,
even for OBDDs of constant width. Second, we show that for each $\oddlength$ and $\width$ one can count the
number of distinct functions computable by ODDs of width at most $\width$ and
length $\oddlength$ in time $\asymptoticfunction(|\alphabet|,\width)\cdot
\oddlength^{O(1)}$, for a suitable $\asymptoticfunction:\N\times \N\rightarrow \N$. 
This improves exponentially on the time necessary to explicitly enumerate all such functions, 
which is exponential in both the width parameter $\width$ and in the length
$\oddlength$ of the ODDs. \\

\end{abstract}


\section{Introduction}
\label{section:Introduction}

In its most traditional setting, automata theory has been used as a framework for the representation and manipulation of (possibly infinite) sets of strings. This framework has been generalized in many ways to allow the representation of sets of more elaborate combinatorial objects, such as trees~\cite{courcelle1989recognizable}, partial orders \cite{thomas1997automata}, 
graphs \cite{bozapalidis2008graph}, pictures \cite{giammarresi1992recognizable}, etc. 
Such notions of automata have encountered innumerous applications in fields such as formal verification \cite{godefroid1990using,bouajjani2006abstract}, finite model theory \cite{ebbinghaus1995finite}, concurrency theory \cite{priese1983automata}, parameterized complexity \cite{courcelle2010verifying,Courcelle1990MSO}, etc. Still, these generalized notions of automata share in common the fact that they are designed to represent (possibly infinite) sets of isolated objects. 

In this work, we combine traditional finite automata with ordered decision diagrams (ODDs) of bounded width to introduce a formalism that can be used to represent and manipulate {\em sets of sets} of strings,
or alternatively speaking, classes of languages. We call this combined formalism {\em second-order finite automata}.
We will show that the width of an ODD is a useful parameter when studying classes of languages from a complexity-theoretic point of view. Additionally, we will
use second-order finite automata to show that several computational problems involving ordered decision diagrams are fixed-parameter tractable 
when parameterized by width.

Given a finite alphabet $\alphabet$ and a number $\width\in \pN$, a $(\alphabet,\width)$-ODD is a sequence $\aodd = \alayer_1\alayer_2\dots \alayer_{\oddlength}$ of $(\alphabet,\width)$-layers. 
Each such a layer $\alayer_i$ has a set of left-states (a subset of $\{1,\dots,\width\}$), a set of right-states (also a subset of $\{1,\dots,\width\}$), and a set of transitions, labeled with letters in $\alphabet$, connecting left states to right states.
We require that for each $i\in \{1,\dots,\oddlength-1\}$, the set of right-states of the layer $\alayer_i$ is equal to the set of left states of the layer $\alayer_{i+1}$. 
The language of an ODD $\aodd$ is the set of strings labelling paths from 
its set of initial states (a subset of the left states of $\alayer_1$) to its final states (a subset of the right states of $\alayer_{\oddlength}$). Since the number of distinct $(\alphabet,\width)$-layers is finite, the set $\layeralphabet{\alphabet}{\width}$ of all $(\alphabet,\width)$-layers can itself be regarded as an alphabet. 
A finite automaton $\finiteautomaton$ over the  alphabet $\layeralphabet{\alphabet}{\width}$ is said to be a second-order finite automaton 
if each string $\aodd = \alayer_1\dots \alayer_{\oddlength}$ in the language $\automatonlang{\finiteautomaton}$ accepted by
$\finiteautomaton$ is a valid ODD. In this case, the second language of $\finiteautomaton$
is defined as the class $\automatonlanghigher{\finiteautomaton}{2} = \{\oddlang{\aodd}\;:\; \aodd\in \automatonlang{\finiteautomaton}\}$
of languages accepted by ODDs in $\automatonlang{\finiteautomaton}$. 
We say that a class of languages $\mathcal{X}$ is {\em regular-decisional} if there is some 
second-order finite automaton $\finiteautomaton$ such that 
$\automatonlanghigher{\finiteautomaton}{2} = \mathcal{X}$. 

\paragraph{Canonical Forms for Second Order Finite Automata.}
Our main result (Theorem \ref{theorem:CanonizationSecondOrder}) states that second-order finite automata can be effectively
canonized with respect to their second languages. More specifically, there is an algorithm that maps each second-order finite automaton
$\finiteautomaton$ to a second-order finite automaton $\canonizationfunction_2(\finiteautomaton)$,
called the second canonical form of $\finiteautomaton$, 
in such a way that the following three properties are satisfied. First, $\canonizationfunction_2(\finiteautomaton)$ and $\finiteautomaton$ have the 
same second language. That is to say, $\automatonlanghigher{\canonizationfunction_2(\finiteautomaton)}{2}= \automatonlanghigher{\finiteautomaton}{2}$. 
Second, any two second-order finite automata $\finiteautomaton$ and $\finiteautomaton'$ with 
identical second languages are mapped to the same canonical form. 
More formally, $\automatonlanghigher{\finiteautomaton}{2} = \automatonlanghigher{\finiteautomaton'}{2}\Rightarrow  \canonizationfunction_2(\finiteautomaton)  = \canonizationfunction_2(\finiteautomaton')$.
Third, $\automatonlang{\canonizationfunction_2(\finiteautomaton)} = \{\canonizationfunction(\aodd)\;:\; \aodd\in \automatonlang{\finiteautomaton}\}$.
Here, $\canonizationfunction(\aodd)$ is the unique deterministic, complete, normalized\footnote{By {\em normalized} we mean that the states of the ODD are numbered according to their lexicographical order. In this way
$\canonizationfunction(\aodd)$ is {\em syntactically} unique and not only unique up to isomorphism.} ODD with minimum number of states with the same language as $\aodd$. 
Intuitively, the language of $\canonizationfunction_2(\finiteautomaton)$ consists precisely of the set of canonical forms of ODDs in the language of $\finiteautomaton$. 
For this reason, we say that Theorem \ref{theorem:CanonizationSecondOrder} is a {\em canonical form of canonical forms theorem}. 
From a complexity-theoretic point of view, $\canonizationfunction_2(\finiteautomaton)$ can be constructed in time
$2^{\numberstates{\finiteautomaton}\cdot 2^{O(|\alphabet|\cdot \width \cdot 2^{\width})}}$, where $\numberstates{\finiteautomaton}$ is the number of states of $\finiteautomaton$. 
Additionally this construction can be sped up to time $2^{\numberstates{\finiteautomaton}\cdot 2^{O(|\alphabet|\cdot \width \cdot \log \width)}}$ if all ODDs
in $\automatonlang{\finiteautomaton}$ are deterministic and complete (Observation \ref{observation:BetterConstruction}). 

We note that canonizing a second-order finite automaton $\finiteautomaton$ with respect to its
second language $\automatonlanghigher{\finiteautomaton}{2}$ is {\em not} equivalent to canonizing 
$\finiteautomaton$ with respect to its language $\automatonlang{\finiteautomaton}$. 
For instance, let $\aodd$ and $\aodd'$ be distinct ODDs such that $\oddlang{\aodd} = \oddlang{\aodd'}$. 
Let $\finiteautomaton$ and $\finiteautomaton'$ be second-order finite automata with 
$\automatonlang{\finiteautomaton} = \{\aodd\}$ and $\automatonlang{\finiteautomaton'} = \{\aodd'\}$. 
Then the languages of $\finiteautomaton$ and $\finiteautomaton'$ are distinct
($\automatonlang{\finiteautomaton} \neq \automatonlang{\finiteautomaton'}$)  even though their 
second languages are equal ($\automatonlanghigher{\finiteautomaton}{2} = \{\oddlang{\aodd}\} = \{\oddlang{\aodd'}\} =  \automatonlanghigher{\finiteautomaton'}{2}$). 

At a high level, what our canonization algorithm does is to eliminate ambiguity in the language
of a given second-order finite automaton. More specifically, any two ODDs $\aodd$ and $\aodd'$ with
$\oddlang{\aodd}=\oddlang{\aodd'}$ in the language of a second-order finite automaton $\finiteautomaton$ 
correspond to a single ODD $\canonizationfunction(\aodd)=\canonizationfunction(\aodd')$ in the language of
$\canonizationfunction_2(\finiteautomaton)$. This implies almost immediately that the 
collection of regular-decisional classes of languages is closed under 
union, intersection, set difference, and even under a suitable notion of complementation. 
Furthermore, emptiness of intersection and inclusion for the second languages of second-order
finite automata are decidable (Theorem \ref{theorem:ClosureProperties}). 
It is interesting to note that non-emptiness of intersection 
for the second languages of second-order finite automata 
can be tested in fixed-parameter tractable time, where the parameter is the maximum width of an ODD 
accepted by one of the input automata (Observation \ref{observation:FPTIntersection}). 
Finally, closure under several operations that are specific to classes of 
languages, such as {\em pointwise union}, {\em pointwise intersection} and {\em pointwise negation}, 
among others can also be obtained as a direct corollary (Corollary \ref{corollary:PointwiseOperations})
of a technical lemma from \cite{deOliveiraOliveira2020symbolic}. 

\paragraph{Main Technical Tool.} 
Let $\odddefiningsetstar{\alphabet}{\width}$ be the set of all $(\alphabet,\width)$-ODDs and $\cdodddefiningsetstar{\alphabet}{\width}$ be the set of all deterministic, complete $(\alphabet,\width)$-ODDs. The main technical tool of this work
(Theorem \ref{theorem:CanonizationTransductionTheorem}) states that the transduction 
$\canonizationtransduction{\alphabet}{\width} = 
\{(\aodd,\canonizationfunction(\aodd))\;:\;\aodd\in \odddefiningsetstar{\alphabet}{\width}\}$
is $2^{O(|\alphabet|\cdot \width \cdot 2^{\width})}$-regular. In other words, there is an NFA
with $2^{O(|\alphabet|\cdot \width \cdot 2^{\width})}$ states accepting the language 
$\{\aodd\otimes \canonizationfunction(\aodd)\;:\; \aodd\in\odddefiningsetstar{\alphabet}{\width}\}$.
Additionally, the transduction $\cdcanonizationtransduction{\alphabet}{\width} = \{(\aodd,\canonizationfunction(\aodd))\;:\; \aodd\in \cdodddefiningsetstar{\alphabet}{\width}\}$, whose domain is restricted to deterministic, complete ODDs, is 
$2^{O(|\alphabet|\cdot \width \cdot \log \width)}$-regular. 

Most results of our work follow as a consequence of Theorem \ref{theorem:CanonizationTransductionTheorem}. 
If we do not take complexity theoretic issues into account, then some of 
our {\em decidability} results also follow by employing other notions of canonizing relations
(see Section \ref{section:Conclusion} for further discussion on this topic). Nevertheless, the transductions 
$\canonizationtransduction{\alphabet}{\width}$ and $\cdcanonizationtransduction{\alphabet}{\width}$ enjoy
special properties that make them attractive from a complexity theoretic point of view. In particular, as we
will see next, these transductions have applications in the fixed-parameter tractability theory of computational 
problems related to ordered decision diagrams ODDs. It is worth noting that ODDs 
comprise the well studied notion of ordered binary decision diagrams (OBDDs) with fixed variable ordering as a special case.
And indeed, the width parameter has relevance in several contexts, such as learning theory \cite{ergun1995learning}, 
the theory of pseudo-random generators \cite{forbes2018pseudorandom}, the theory of symbolic 
algorithms \cite{deOliveiraOliveira2020symbolic}, and structural graph theory \cite{deOliveiraOliveira2019width}.
Additionally, Theorem \ref{theorem:CanonizationTransductionTheorem} implies
that the set $\{\canonizationfunction(\aodd)\;:\;\aodd\in \odddefiningsetstar{\alphabet}{\width}\}$ of all 
minimized, deterministic, complete ODDs accepting the language of some ODD in $\odddefiningsetstar{\alphabet}{\width}$
is regular (Corollary \ref{corollary:RegularCanonicalForms}), and therefore, can be accepted by some deterministic finite 
automaton $\finiteautomaton$. This result may be of independent interest 
since the fact that the canonical form $\canonizationfunction(\aodd)$ has minimum number of states among all 
deterministic, complete ODDs with the same language as $\aodd$ is a relevant 
complexity theoretic information about the language $\oddlang{\aodd}$. One 
interesting consequence of this result is that there is a bijection $b$ from the 
set of accepting paths of $\finiteautomaton$ and the class of languages accepted by 
ODDs in $\odddefiningsetstar{\alphabet}{\width}$. Additionally, the ODD 
corresponding to each such a path $\mathfrak{p}$ has minimum 
number of states among all deterministic, complete ODDs accepting the language $b(\mathfrak{p})$. 

\paragraph{Algorithmic Applications.}
Although ODDs of constant width constitute a simple computational model, they
can already be used to represent many interesting functions. It is worth noting
that for each width $w\geq 3$, the class of functions that can be represented
by ODDs of constant width is at least as difficult to learn in the PAC-learning
model as the problem of learning DNFs \cite{ergun1995learning}. Additionally,
the study of ODDs of constant width is still very active in the theory of
pseudo-random generators \cite{forbes2018pseudorandom}. Our main results can be
used to show that several width/size minimization problems for nondeterministic and
deterministic ODDs can be solved in fixed parameter tractable time when
parameterized by width. For instance, we show that given an ODD $\aodd$ of length $\oddlength$
and width $\width$ over an alphabet $\alphabet$, one can compute in time
$2^{O(|\alphabet|\cdot \width \cdot 2^{\width})}\cdot \oddlength^{O(1)}$ an ODD $\aodd'$ of minimum width 
such that $\oddlang{\aodd'}=\oddlang{\aodd}$. A more efficient algorithm, running in time 
$2^{O(|\alphabet|\cdot \width \cdot \log \width)}\cdot \oddlength^{O(1)}$ can be obtained 
if the input ODD is deterministic (Theorem \ref{theorem:Minimization}).
Our algorithm is in fact more general and can be used to minimize other complexity measures, such as number of
states and number of transitions among all ODDs belonging to the language of a given second-order finite automaton $\finiteautomaton$ (Theorem \ref{theorem:MinimizationODDAutomaton}).

Our algorithm for width minimization of ODDs parameterized by width naturally can be used to minimize the width 
of ordered binary decision diagrams (OBDDs), since OBDDs {\em with a fixed variable ordering} correspond to 
ODDs over a binary alphabet. Width minimization problems for OBDDs have been considered before in the literature
\cite{bollig2014width,Bollig16}, but previously known algorithms are exponential on the size of the OBDD even for OBDDs of 
constant width, and even in the case of when one is not allowed to vary the order of the input variables. Our 
FPT result shows that width minimization for OBDDs of constant width with a fixed variable ordering can be achieved in polynomial time. 

As a second application of our main results, we show that the problem of counting the number of distinct functions computable 
by ODDs of a given width $\width$ and a given length $\oddlength$ can be solved in time 
$2^{2^{O(|\alphabet|\cdot \width \cdot 2^{\width})}}\cdot \oddlength^{O(1)}$. This running time 
can be improved to $2^{2^{O(|\alphabet|\cdot \width \cdot \log \width)}}\cdot \oddlength^{O(1)}$ 
if we are interested in counting the number of functions computable by deterministic, complete ODDs of width $\width$ and 
length $\oddlength$ (Corollary \ref{corollary:CountingFunctions}). 
We note that this restricted case is relevant because ordered binary decision diagrams (OBDDs) 
defined in the literature are usually deterministic and complete. Our results imply that counting the number of functions computable by 
OBDDs of width $\width$ with a fixed variable ordering can be solved in time polynomial in the number
of variables. This improves exponentially on the approach of explicit enumeration without
repetitions, which takes time exponential in $\oddlength$. This result is obtained as a consequence of a more general 
theorem analyzing the complexity of the problem counting functions represented by ODDs of a given length in the language 
of a given second-order finite automaton $\finiteautomaton$ (Theorem \ref{theorem:CountingFunctionsAutomaton}). 

The reminder of this paper is organized as follows. Next, in Section \ref{section:Preliminaries},
we define some basic concepts and state well-known results concerning finite automata and ordered
decision diagrams. Subsequently, in Section \ref{section:SecondOrderFiniteAutomata}, we formally
define the notion of second-order finite automata and state our main
results (Theorem \ref{theorem:CanonizationTransductionTheorem} and Theorem \ref{theorem:CanonizationSecondOrder}).
In Section \ref{section:ClosureProperties}, we state several closure properties for second-order finite automata.
In Section \ref{section:Applications}, we discuss several algorithmic applications of our main results.
In Section \ref{section:CanonizationTransductionTheorem} we prove Theorem \ref{theorem:CanonizationTransductionTheorem}. 
Finally, in Section \ref{section:Conclusion} we draw some concluding remarks and establish 
connections with related work.

\section{Preliminaries} 
\label{section:Preliminaries}

\subsection{Basics}

We denote by $\N \defeq \set{0, 1, \ldots}$ the set of natural numbers (including zero), and by $\pN \defeq \N\setminus \set{0}$ the set of positive natural numbers. 
For each $\anumber \in \pN$, we let $\bset{\anumber} \defeq \set{1, 2, \ldots, \anumber}$ and $\dbset{\anumber} \defeq \set{0, 1, \ldots, \anumber-1}$. 
For each finite set $\sset$, we let $\pwset{\sset} \defeq \set{\sset' \setst \sset' \subseteq \sset}$ denote the \emph{power set} of $\sset$. 
For each two sets $\sset$ and $\yset$, each function $f\colon \sset \rightarrow \yset$ and each subset $\sset'\subseteq \sset$, we let $f\restr{\sset'}$ denote the \emph{restriction} of $f$ to $\sset'$, \ie the function $f\restr{\sset'} \colon \sset' \rightarrow \yset$ such that $f\restr{\sset'}(x) = f(x)$ for each $x \in \sset'$.

\subsubsection*{Alphabets and Strings.}
An {\em alphabet} is any finite, non-empty set $\alphabet$. 
A \emph{string} over an alphabet $\alphabet$ is any finite sequence of symbols from $\alphabet$. 
The {\em empty string}, denoted by $\emptystring$, is the unique string of length zero. 
We denote by $\alphabet^*$ the set of all strings over $\alphabet$, including the empty string $\emptystring$, and by $\alphabet^+  \defeq \alphabet^*\setminus \set{\emptystring}$ the set of all non-empty strings over $\alphabet$. A \emph{language} over $\alphabet$ is any subset 
$\lang$ of $\alphabet^{*}$. 
In particular, for each $\stringlength \in \N$, we let $\alphabet^{\stringlength}$ be the language of all strings of length $\stringlength$ over $\alphabet$.
We say that an alphabet $\alphabet$ is ordered if it is endowed with a total order $\orderalphabet_{\alphabet}:\alphabet\times \alphabet$. 
Such an order $\orderalphabet_{\alphabet}$ is extended naturally to a lexicographical order $\orderstrings_{\alphabet}\subseteq \alphabet^*\times \alphabet^*$ on the set $\alphabet^*$. 
Unless stated otherwise, we assume that each alphabet considered in this paper is endowed with a fixed total order.

\subsubsection*{Finite Automata.}
A \emph{finite automaton} (FA) over an alphabet $\alphabet$ is a tuple 
$\finiteautomaton = \tuple{\alphabet,\automatonstates,\automatoninitialstates,\automatonfinalstates,\automatontransitions}$, 
where $\automatonstates$ is a finite set of \emph{states}, $\automatoninitialstates\subseteq \automatonstates$ is a set of \emph{initial states}, $\automatonfinalstates\subseteq \automatonstates$ is a set of \emph{final states} and ${\automatontransitions\subseteq \automatonstates\times \alphabet\times \automatonstates}$ is a set of \emph{transitions}.
The \emph{size} of $\finiteautomaton$ is defined as $\abs{\finiteautomaton} \defeq \abs{\automatonstates} + \abs{\automatontransitions}\cdot\log{\abs{\alphabet}}$.
We denote the number of states of $\finiteautomaton$ by $\nstates(\finiteautomaton)\defeq \abs{\automatonstates}$, and the number of transitions of $\finiteautomaton$ by $\ntransitions(\finiteautomaton) \defeq \abs{\automatontransitions}$. 

Let $\astring\in \alphabet^*$, and $\astate,\bstate\in \automatonstates$. We say that $\astring$ reaches $\bstate$ from 
$\astate$ if either $\astring  = \emptystring$ and $\astate=\bstate$, or if $\astring = \asymbol_1\dots\asymbol_{\stringlength}$ for some 
$k\in \pN$ and there is a sequence $$\sequence{\tuple{\astate_{0},\asymbol_{1},\astate_{1}},\tuple{\astate_{1},\asymbol_{2},\astate_{2}},
\ldots,\tuple{\astate_{\stringlength-1},\asymbol_\stringlength,\astate_\stringlength}}\text{,}$$
of transitions such that $\astate_0 = \astate$, $\astate_{\stringlength} = \bstate$ and $\tuple{\astate_{i},\asymbol_{i+1},\astate_{i+1}} \in \automatontransitions$ for each $i\in\bset{\stringlength-1}$. 
We say that $\finiteautomaton$ \emph{accepts} $\astring$ if there exist states $\astate\in \automatoninitialstates$ and $\bstate\in \automatonfinalstates$ such that 
$\astring$ reaches $\bstate$ from $\astate$. 
The \emph{language} of $\finiteautomaton$ is defined as the set $$\automatonlang{\finiteautomaton} \defeq \set*{\astring \in \alphabet^{*} \setst \astring \text{ is accepted by } \finiteautomaton}$$ of all finite 
strings over $\alphabet$ accepted by $\finiteautomaton$. 
For $\agrowthconstant\in \N$, we say that a language $\lang\subseteq \alphabet^*$ is 
{\em $\agrowthconstant$-regular} if there exists a
finite automaton with at most $\agrowthconstant$ states such that $\automatonlang{\finiteautomaton} = \lang$.

We say that $\finiteautomaton$ is \emph{deterministic} if 
$\finiteautomaton$ contains exactly one initial state, \ie $\abs{\automatoninitialstates} = 1$, and 
for each  $\astate \in \automatonstates$ and each $\asymbol\in\alphabet$, there exists at most one state 
$\bstate \in \automatonstates$ such that $\tuple{\astate,\asymbol,\bstate}$ is a transition in $\automatontransitions$. 
We say that $\finiteautomaton$ is \emph{complete} if it has at least one initial state,
and for each $\astate \in \automatonstates$ and each $\asymbol\in\alphabet$,
there exists at least one state $\bstate \in \automatonstates$ such that $\tuple{\astate,\asymbol,\bstate}$ is a transition in $\automatontransitions$.
We say that $\finiteautomaton$ is \emph{reachable} if for each state $\astate\in \automatonstates$, there is a sequence of 
transitions from some initial state of $\finiteautomaton$ to $\astate$. If $\finiteautomaton$ is a reachable finite automaton, 
then for each state $\astate\in \automatonstates$, we let $\lex(\astate)$ denote the lexicographically first string 
that reaches $\astate$ from some initial state, according to the order $\orderstrings_{\alphabet}$. 
We say that $\finiteautomaton$ is \emph{normalized} if $\automatonstates = \dbset{\automatonnumberstates}$
for some $\automatonnumberstates\in \pN$, and $\astate<\bstate$ if and only if $\lex(\astate) \orderstrings_{\alphabet} \lex(\bstate)$ for each $\astate,\bstate \in \automatonstates$. 

In what follows, we may write $\automatonstates(\finiteautomaton)$, $\automatontransitions(\finiteautomaton)$, $\automatoninitialstates(\finiteautomaton)$ and $\automatonfinalstates(\finiteautomaton)$ to refer to the sets $\automatonstates$, $\automatontransitions$, $\automatoninitialstates$ and $\automatonfinalstates$, respectively. 

The following theorem, stating the existence of canonical forms for finite automata, is one of the most fundamental results in automata theory. 

\begin{theorem}\label{theorem:MinimizationAutomaton}
For each finite automaton $\finiteautomaton$, there exists a unique finite automaton $\canonizationfunction(\finiteautomaton)$ with minimum number of states such that $\canonizationfunction(\finiteautomaton)$ is deterministic, complete, normalized, and satisfies $\automatonlang{\canonizationfunction(\finiteautomaton)} = \automatonlang{\finiteautomaton}$. 
\end{theorem}

We note that given a (possibly non-deterministic) finite automaton $\finiteautomaton$, the canonical form $\canonizationfunction(\finiteautomaton)$ of $\finiteautomaton$ can be obtained by the following process. 
First, one applies Rabin's power-set construction to $\finiteautomaton$ in order to obtain a {\em deterministic, complete} finite automaton $\finiteautomaton'$ that accepts the same language as $\finiteautomaton$. Subsequently, by using Hopcroft's algorithm~\cite{hopcroft1971n} for instance, one minimizes $\finiteautomaton'$ in order to obtain a deterministic finite automaton $\finiteautomaton''$ that accepts the same language as $\finiteautomaton$ and has the minimum number of states. 
At this point, the finite automaton $\finiteautomaton''$ is unique {\em up to renaming of states}. 
Thus, as a last step, one obtains the canonical form $\canonizationfunction(\finiteautomaton)$ by renaming the states of $\finiteautomaton''$ in such a way that the normalization property is satisfied. 
Note that the automaton $\canonizationfunction(\finiteautomaton)$ is finally syntactically unique. 
In particular, for each two finite automata $\finiteautomaton$ and $\finiteautomaton'$, $\automatonlang{\finiteautomaton} = \automatonlang{\finiteautomaton'}$ if and only if $\canonizationfunction(\finiteautomaton) = \canonizationfunction(\finiteautomaton')$.

\subsection{Ordered Decision Diagrams}

\subsubsection*{Layers.} Let $\alphabet$ be an alphabet and $\width\in\pN$. 
A \emph{$\tuple{\alphabet,\width}$-layer} is a tuple $\alayer \defeq \tuple{\layerleftfrontier,\layerrightfrontier, \layertransitions, \layerinitialstates, \layerfinalstates, \layerinitialflag, \layerfinalflag}$, where $\layerleftfrontier \subseteq \dbset{\width}$ is a set of \emph{left states}, $\layerrightfrontier \subseteq \dbset{\width}$  is a set of \emph{right states}, $\layertransitions \subseteq \layerleftfrontier\cartesianproduct\alphabet \cartesianproduct \layerrightfrontier$ is a set of \emph{transitions}, $\layerinitialstates \subseteq \layerleftfrontier$ is a set of \emph{initial states}, $\layerfinalstates \subseteq \layerrightfrontier$ is a set of \emph{final states} and $\layerinitialflag, \layerfinalflag \in \set{\false, \true}$ are Boolean flags satisfying the two following conditions: 
\begin{enumerate}
\item if $\layerinitialflag = \false$, then $\layerinitialstates =\emptyset$; 
\item if $\layerfinalflag = \false$, then $\layerfinalstates = \emptyset$.
\end{enumerate}

In what follows, we may write $\layerleftfrontier(\alayer)$, $\layerrightfrontier(\alayer)$, $\layertransitions(\alayer)$, $\layerinitialstates(\alayer)$, $\layerfinalstates(\alayer)$, $\layerinitialflag(\alayer)$ and $\layerfinalflag(\alayer)$ to refer to the sets $\layerleftfrontier$, $\layerrightfrontier$, $\layertransitions$,
$\layerinitialstates$ and $\layerfinalstates$ and to the Boolean flags $\layerinitialflag$ and $\layerfinalflag$, respectively.

We let $\layeralphabet{\alphabet}{\width}$ denote the set of all $\tuple{\alphabet,\width}$-layers.  
Note that $\layeralphabet{\alphabet}{\width}$ is non-empty and has at most $2^{\BigOh(\abs{\alphabet} \cdot \width^2)}$ elements. 
Therefore, $\layeralphabet{\alphabet}{\width}$ may be regarded as an alphabet.

\subsubsection*{Ordered Decision Diagrams.}
Let $\alphabet$ be an alphabet and $\width,\oddlength\in\pN$.  
A \emph{$(\alphabet,\width)$-ordered decision diagram} (or simply, $(\alphabet,\width)$-\emph{ODD}) 
of \emph{length} $\oddlength$ is a string $\aodd \defeq \alayer_1 \cdots \alayer_\oddlength \in \layeralphabet{\alphabet}{\width}^{\oddlength}$ of length $\oddlength$ over the alphabet $\layeralphabet{\alphabet}{\width}$ satisfying the following
conditions:

\begin{enumerate} 
\item\label{condition:OneODD} for each $i\in\bset{\oddlength-1}$, $\layerleftfrontier(\alayer_{i+1}) = \layerrightfrontier(\alayer_i)$; 
\item\label{condition:TwoODD} $\layerinitialflag(\alayer_{1}) = \true$ and, for each $i \in \set{2,\ldots,\oddlength}$, $\layerinitialflag(\alayer_{i}) = \false$; 
\item\label{condition:ThreeODD} $\layerfinalflag(\alayer_{\oddlength}) = \true$ and, for each $i \in \bset{\oddlength-1}$, $\layerfinalflag(\alayer_{i}) = \false$.  
\end{enumerate}

Intuitively, Condition~\ref{condition:OneODD} expresses that for each $i\in\bset{\oddlength-1}$, the set of right states of $\alayer_i$ can be identified with the set of left states of $\alayer_{i+1}$. 
Condition \ref{condition:TwoODD} guarantees that only the first layer of an ODD is allowed to have initial states. 
Analogously, Condition \ref{condition:ThreeODD} guarantees that only the last layer of an ODD is allowed to have final states. 

Let $\aodd = \alayer_{1} \cdots \alayer_{\oddlength}$ be a $(\alphabet,\width)$-ODD of length $\oddlength$, for some $\oddlength\in\pN$. 
We let $\lengthfunction{\aodd}\defeq\oddlength$ denote the length of $\aodd$, 
$\numberstates{\aodd}\defeq\abs{\layerleftfrontier(\alayer_{1})} + \sum_{i\in \bset{\oddlength}} \abs{\layerrightfrontier(\alayer_{i})}$ denote the \emph{number of states} of $\aodd$, $\numbertransitions{\aodd} \defeq \abs{\layertransitions(\alayer_{1})} + \sum_{i\in \bset{\oddlength}} \abs{\layertransitions(\alayer_{i})}$ denote the \emph{number of transitions} of $\aodd$,
$$\fwidth(\aodd) \defeq \max\set*{\abs{\layerleftfrontier(\alayer_{1})}, \ldots, \abs{\layerleftfrontier(\alayer_{\oddlength})}, \abs{\layerrightfrontier(\alayer_{\oddlength})}}$$ denote the \emph{width} of $\aodd$. 
We remark that $\fwidth(\aodd) \leq \width$. 

For each subset $\subsetlayers\subseteq\layeralphabet{\alphabet}{\width}$ and each positive integer $\oddlength \in \pN$, we denote by $\subsetlayersdefiningset{\oddlength}$ the set of all
$(\alphabet,\width)$-ODDs of length $\oddlength$ whose layers belong to the set $\subsetlayers$. 
Additionally, for each subset $\subsetlayers\subseteq\layeralphabet{\alphabet}{\width}$, we denote by $\subsetlayersstar \defeq \bigcup_{\oddlength\in \pN} \subsetlayersdefiningset{\oddlength}$ the set of all $(\alphabet,\width)$-ODDs whose layers belong to the set $\subsetlayers$. 
In particular, we denote by $\odddefiningset{\alphabet}{\width}{\oddlength}$ the set of all $(\alphabet,\width)$-ODDs of length $\oddlength$, and we denote by $\odddefiningsetstar{\alphabet}{\width}$ the set of all $(\alphabet,\width)$-ODDs.

\subsubsection*{Length Typed Subsets of $\alphabet^{k}$.} 
Let $\alphabet$ be an alphabet and $\stringlength\in \pN$.
In this work, it is convenient to assume that subsets of $\alphabet^{\stringlength}$ are typed with their length. 
This can be achieved by viewing each subset $\lang \subseteq \alphabet^{\stringlength}$ as a pair of the form $(\stringlength,\lang)$. 
We let $\powerset_{\stringlength}(\alphabet) = \set{\tuple{\stringlength,\lang} \setst \lang\subseteq \alphabet^{\stringlength}}$
be the set of all length typed subsets of $\alphabet^{\stringlength}$. Given length typed sets 
$(\stringlength,\lang_1)$ and $(\stringlength,\lang_2)$, we define
$(\stringlength,\lang_1)\cup (\stringlength,\lang_2) \defeq (\stringlength,\lang_1\cup \lang_2)$, 
$(\stringlength,\lang_1)\cap (\stringlength,\lang_2) \defeq (\stringlength,\lang_1\cap \lang_2)$, 
$(\stringlength,\lang_1)\backslash (\stringlength,\lang_2) \defeq (\stringlength,\lang_1\backslash \lang_2)$,
$(\stringlength,\lang_1)\otimes (\stringlength,\lang_2) \defeq (\stringlength,\lang_1\otimes \lang_2)$, 
and for maps $g:\alphabet\rightarrow \alphabet'$ and $h:\alphabet'\rightarrow \alphabet$, 
we let $g(\stringlength,\lang) \defeq (\stringlength,g(\lang))$ and $h^{-1}(\stringlength,\lang) \defeq (\stringlength,h^{-1}(\lang))$.

\subsubsection*{Language Accepted by an ODD.} 
Let $\alphabet$ be an alphabet, $\width,\oddlength\in\pN$, $\odd = \layer_1 \cdots \layer_{\oddlength}$ be an ODD in $\odddefiningset{\alphabet}{\width}{\oddlength}$ and $\astring = \asymbol_{1}
\cdots \asymbol_{\oddlength}$ be a string in $\alphabet^{\oddlength}$. 
A \emph{valid sequence} for $\astring$ in $\odd$ is a sequence of transitions $$\sequence{\tuple{\layerleftstate_1,\asymbol_{1},\layerrightstate_1}, \ldots, \tuple{\layerleftstate_{\oddlength},\asymbol_{\oddlength},\layerrightstate_{\oddlength}}}$$ such that $\layerleftstate_{i+1} = \layerrightstate_{i}$ for each $i\in\bset{\oddlength-1}$, and $\tuple{\layerleftstate_{i},\asymbol_{i},\layerrightstate_{i}}\in\layertransitions(\alayer_{i})$ for each $i\in\bset{\oddlength}$. 
Such a valid sequence is called \emph{accepting} for $\astring$ if, additionally, $\layerleftstate_{1}$ is an initial state in $\layerinitialstates(\alayer_1)$ and $\layerrightstate_{\oddlength}$ is a final state in $\layerfinalstates(\alayer_{\oddlength})$.  
We say that $\odd$ \emph{accepts} $\astring$ if there exists an accepting sequence for $\astring$ in $\aodd$. 
The \emph{language} of $\aodd$ is defined as the (length-typed) set $$\oddlang{\aodd} \defeq \tuple{\oddlength,\set*{\astring \in \alphabet^{\oddlength} \setst \astring \text{ is accepted by } \aodd}}$$ of all strings accepted by $\aodd$. Note that every string accepted by $\aodd$ has length $\oddlength$. 

In Figure~\ref{fig:evenODD}, we depict an ODD $\odd\in\odddefiningset{\set{0,1}}{2}{5}$ whose language is the length-typed set $\oddlang{\aodd}=\tuple{5,\set{\astring=\asymbol_{1}\cdots\asymbol_{5}\in\set{0,1}^{5} \setst \asymbol_{1}+\cdots+\asymbol_{5} \equiv 0 \pmod 2}}$ of all binary strings of length $5$ with an even number of occurrences of the symbol `$1$'. 
For instance, 
$$\sequence{\tuple{0,0,0},\tuple{0,1,1},\tuple{1,0,0},\tuple{0,1,0},\tuple{0,0,0}}$$ is an accepting sequence in $\odd$ for the string $01010$, which has two occurrences of the symbol `$1$'. 

\begin{figure}[ht]\centering
	\includegraphics[scale=1.4]{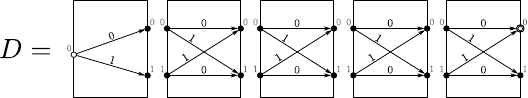}
	\caption{Example of ODD $\odd\in\odddefiningset{\set{0,1}}{2}{5}$ whose language consists of all binary strings of length $5$ with an even number of occurrences of the symbol `$1$'.}\label{fig:evenODD}
\end{figure}

\subsubsection*{Deterministic and Complete ODDs.}
Let $\alphabet$ be an alphabet and $\width\in\pN$. 
A $\tuple{\alphabet,\width}$-layer $\alayer$ is called \emph{deterministic} if the following conditions are satisfied:
\begin{enumerate}
  \item if $\layerinitialflag(\layer) = 1$, then $\layerinitialstates(\layer) = \layerleftfrontier(\layer)$ and $\abs{\layerleftfrontier(\layer)} = 1$;
  \item for each $\layerleftstate \in \layerleftfrontier(\layer)$ and each $\asymbol \in \alphabet$, there exists at most one right state $\layerrightstate \in \layerrightfrontier(\layer)$ such that $\tuple{\layerleftstate, \asymbol, \layerrightstate} \in \layertransitions(\layer)$.
\end{enumerate}
A $\tuple{\alphabet,\width}$-layer $\alayer$ is called \emph{complete} if the following conditions are satisfied: 
\begin{enumerate}
  \item if $\layerinitialflag(\layer) = 1$, then $\layerinitialstates(\layer) \neq \emptyset$;
  \item for each $\layerleftstate \in \layerleftfrontier(\layer)$ and each $\asymbol \in \alphabet$, there exists at least one right state $\layerrightstate \in \layerrightfrontier(\layer)$ such that $\tuple{\layerleftstate, \asymbol, \layerrightstate} \in \layertransitions(\layer)$.  
\end{enumerate}

We let $\cdlayeralphabet{\alphabet}{\width}$ be the subset of $\layeralphabet{\alphabet}{\width}$
comprising all deterministic, complete $\tuple{\alphabet,\width}$-layers.

\begin{observation}
\label{observation:NumberLayers}
Let $\alphabet$ be an alphabet, and $\width \in \pN$.
\begin{enumerate}
	\item The alphabet $\cdlayeralphabet{\alphabet}{\width}$ has $2^{O(|\alphabet|\width \log \width)}$ layers. 
	\item The alphabet $\layeralphabet{\alphabet}{\width}$ has $2^{O(|\alphabet|\width^2)}$ layers. 
\end{enumerate}
\end{observation}
\begin{proof}
$ $
\begin{enumerate}
\item Let $\alphabet$ be an alphabet, and $x,y\in \{0,1,\dots,\width\}$. We note that there are 
at most $$d(\alphabet,x,y)=\binom{\width}{x}\binom{\width}{y}(x+1)(1+2^{y})y^{|\alphabet|x} =  w^{O(|\alphabet|w)} = 2^{O(|\alphabet|\width \log \width)}$$
deterministic complete layers with $x$ left states, $y$ right states and transitions labeled by symbols in $\alphabet$. Indeed, there 
are $\binom{\width}{x}$ ways of choosing $x$ left states, out of the set $\{1,\dots,\width\}$, $\binom{\width}{y}$ ways
of choosing $y$ right states out of the set $\{1,\dots,\width\}$, $(\width+1)$ ways of choosing the initial set of states 
$\layerinitialstates(\alayer)$ together with the initial flag $\initialflag(\alayer)$ 
(because $\layerinitialstates(\alayer)=\emptyset$ if $\initialflag(\alayer)=0$ and $|\layerinitialstates(\alayer)|=1$ if $\initialflag(\alayer)=1$, due to determinism), 
$1+2^{y}$ ways of choosing the subset of final states $\layerfinalstates(\alayer)$ together with the final flag $\finalflag(\alayer)$
(because $\layerfinalstates(\alayer)=0$ if $\finalflag(\alayer)=0$ and $\layerfinalstates(\alayer)$ is an arbitrary subset of $\rightfrontier(\alayer)$ if $\finalflag(\alayer)=1$), 
and $y^{|\alphabet|x}$ ways of choosing the transition relation $\layertransitions(\alayer)$ (because there are $x$ left states, and for each such state $\astate$ and 
each symbol $a\in \alphabet$ there are $y$ ways of choosing the unique transition with label $a$ leaving $\astate$). Therefore, we have that 
$|\cdlayeralphabet{\alphabet}{\width}|\leq \sum_{x,y=0}^{\width} d(\alphabet,x,y) = (w+1)^2\cdot 2^{O(|\alphabet|\width\log w)} = 2^{O(|\alphabet|\width \log\width)}$.
\item By a similar analysis we can conclude that for each alphabet $\alphabet$, and each $x,y\in \{0,1,\dots,\width\}$ there are at most  
at most $$n(\alphabet,x,y)=\binom{\width}{x}\binom{\width}{y}(1+2^{x})(1+2^{y})2^{|\alphabet|xy} = 2^{O(|\alphabet|\width^2)}$$
(possibly nondeterministic) layers with $x$ left state, $y$ right states, and transitions labeled with symbols from $\alphabet$. The essential
differences are that in the nondeterministic case, there are $(1+2^{x})$ ways of choosing the set of initial states together with the initial flag
(because $\layerinitialstates(\alayer)=\emptyset$ if $\initialflag(\alayer)=0$, and $\layerinitialstates(\alayer)$ may be an arbitrary subset of
$\layerleftfrontier(\alayer)$ if $\initialflag(\alayer)=1$), and that there are $2^{|\alphabet|xy}$ ways of choosing the transition relation
$\layertransitions(\alayer)$ (because there are $x$ left states, and for each such a state $\astate$ and each symbol $a\in \alphabet$ there are $2^{y}$ 
ways of choosing the set of transitions with label $a$ leaving $\astate$). Therefore, we have that 
$|\layeralphabet{\alphabet}{\width}|\leq \sum_{x,y=0}^{\width} n(\alphabet,x,y) = (w+1)^2\cdot 2^{O(|\alphabet|\width^2)} = 2^{O(|\alphabet|\width^2)}$.
\end{enumerate}
\end{proof}

Let $\oddlength\in\pN$ and $\odd = \alayer_1\cdots\alayer_{\oddlength} \in \odddefiningset{\alphabet}{\width}{\oddlength}$. 
We say that $\odd$ is \emph{deterministic} (\emph{complete}, resp.) if for each $i \in \bset{\oddlength}$, $\alayer_i$ is a deterministic (complete, resp.) layer. 
We remark that if $\odd$ is deterministic, then there exists at most one valid sequence in $\odd$ for each string in $\alphabet^{\oddlength}$.
On the other hand, if $\odd$ is complete, then there exists at least one valid sequence in $\odd$ for each string in $\alphabet^{\oddlength}$.

For each $\oddlength \in \pN$, we denote by $\cdodddefiningset{\alphabet}{\width}{\oddlength}$ the subset of
$\odddefiningset{\alphabet}{\width}{\oddlength}$ comprising all deterministic, complete $(\alphabet,\width)$-ODDs of length $\oddlength$. 
We denote by $\cdodddefiningsetstar{\alphabet}{\width}$ the subset of $\odddefiningsetstar{\alphabet}{\width}$ comprising all deterministic, 
complete $(\alphabet,\width)$-ODDs.

\subsubsection*{Isomorphism of ODDs.} 
Let $\alphabet$ be an alphabet, $\width,\oddlength\in\pN$, and let $\aodd = \alayer_1\cdots\alayer_{\oddlength}$ and $\bodd = \blayer_1\cdots\blayer_{\oddlength}$ be two ODDs in $\odddefiningset{\alphabet}{\width}{\oddlength}$. 
An \emph{isomorphism from $\aodd$ to $\bodd$} is a sequence $\isomorphismsequence \defeq \sequence{\isomorphism_{0},\ldots,\isomorphism_{\oddlength}}$ 
of functions that satisfy the following conditions: 
\begin{enumerate}
	\item $\isomorphism_{0}\colon\layerleftfrontier(\alayer_{0})\rightarrow \layerleftfrontier(\blayer_{0})$ is a bijection from $\layerleftfrontier(\alayer_{0})$ to $\layerleftfrontier(\blayer_{0})$; 
	\item $\isomorphism_{0}\restr{\layerinitialstates(\alayer_{0})}$ is a bijection from $\layerinitialstates(\alayer_{0})$ to $\layerinitialstates(\blayer_{0})$; 
	\item for each $i \in \bset{\oddlength}$, $\isomorphism_{i}\colon\layerrightfrontier(\alayer_{i})\rightarrow \layerrightfrontier(\blayer_{i})$ is a bijection from $\layerrightfrontier(\alayer_{i})$ to $\layerrightfrontier(\blayer_{i})$;
	\item $\isomorphism_{\oddlength}\restr{\layerfinalstates(\alayer_{\oddlength})}$ is a bijection from $\layerfinalstates(\alayer_{\oddlength})$ to $\layerfinalstates(\blayer_{\oddlength})$;
	\item for each $i \in \bset{\oddlength}$, each left state $\layerleftstate\in\layerleftfrontier(\alayer_{i})$, each symbol $\asymbol\in\alphabet$ and each right state $\layerrightstate\in\layerrightfrontier(\alayer_{i})$, $\tuple{\layerleftstate,\asymbol,\layerrightstate}\in \layertransitions(\alayer_{i})$ if and only if $\tuple{\isomorphism_{i-1}(\layerleftstate),\asymbol,\isomorphism_{i}(\layerrightstate)}\in \layertransitions(\blayer_{i})$.
\end{enumerate}

We remark that if $\isomorphismsequence = \sequence{\isomorphism_{0},\ldots,\isomorphism_{\oddlength}}$ is an isomorphism from $\aodd$ to $\bodd$, then the sequence $\isomorphismsequence^{-1} \defeq \sequence{\isomorphism_{0}^{-1},\ldots,\isomorphism_{\oddlength}^{-1}}$ is an isomorphism from $\bodd$ to $\aodd$, where $\isomorphism_{i}^{-1}$ denotes the inverse function of $\isomorphism_{i}$ for each $i \in \dbset{\oddlength+1}$. 
We say that $\aodd$ and $\bodd$ are \emph{isomorphic} if there exists an isomorphism $\isomorphismsequence$ between $\aodd$ and $\bodd$. 
The following proposition is immediate.

\begin{proposition}\label{proposition:IsomorphicEquivalenceLanguages}
Let $\alphabet$ be an alphabet, $\width \in \pN$, and let $\aodd$ and $\bodd$ be two $(\alphabet,\width)$-ODDs. 
If $\aodd$ and $\bodd$ are isomorphic, then $\oddlang{\aodd} = \oddlang{\bodd}$. 
\end{proposition}

\subsubsection*{Normalized ODDs.}
Let $\alphabet$ be an alphabet, $\width\in\pN$, and let $\layer$ be a $(\alphabet,\width)$-layer.  
We say that $\alayer$ is {\em reachable} if for each right state $\layerrightstate\in\layerrightfrontier(\alayer)$, there exist a symbol $\asymbol \in \alphabet$ and
a left state $\layerleftstate\in\layerleftfrontier(\alayer)$ such that $\tuple{\layerleftstate\,\asymbol,\layerrightstate}$ is a transition in $\layertransitions(\alayer)$. 
If $\alayer$ is reachable, then we let $\porderfrontierName{\alayer} \colon \layerrightfrontier(\alayer) \rightarrow \layerleftfrontier(\alayer) \cartesianproduct \alphabet$ be the function such that for each right state $\layerrightstate \in \layerrightfrontier(\layer)$, $$\porderfrontier{\layer}{\layerrightstate} \defeq \min\set{\tuple{\layerleftstate,\asymbol} \setst \tuple{\layerleftstate,\asymbol,\layerrightstate}\in
\layertransitions(\layer)}\text{,}$$ where the minimum is taken lexicographically, \ie, for each two left states $\layerleftstate, \layerleftstate' \in \layerleftfrontier(\layer)$ and each two symbols $\asymbol, \bsymbol \in \alphabet$, we have that $\tuple{\layerleftstate,\asymbol}<\tuple{\layerleftstate',\bsymbol}$ if and only if $\layerleftstate<\layerleftstate'$, or $\layerleftstate=\layerleftstate'$ and $\asymbol\orderalphabet_{\alphabet}\bsymbol$.
(Recall we are assuming that the alphabet $\alphabet$ is endowed with a fixed total order $\orderalphabet_{\alphabet}\subseteq \alphabet\times \alphabet$.) 
We say that $\layer$ is {\em well-ordered} if it is a reachable, deterministic layer such that for each two right states $\layerrightstate,\layerrightstate'\in \layerrightfrontier(\layer)$, we have that $\layerrightstate<\layerrightstate'$ if and only if $\porderfrontier{\layer}{\layerrightstate}<\porderfrontier{\layer}{\layerrightstate'}$. 
We say that $\alayer$ is {\em contiguous} if $\leftfrontier(\layer) = \dbset{\width_1}$ and $\rightfrontier(\layer) = \dbset{\width_2}$ for some $\width_1,\width_2\in\bset{\width}$. 
Then, we say that $\layer$ is {\em normalized} if it is both well-ordered and contiguous. 

Let $\oddlength\in\pN$ and $\odd = \alayer_1\cdots\alayer_{\oddlength}$ be an ODD in $\odddefiningset{\alphabet}{\width}{\oddlength}$. 
We say that $\aodd$ is \emph{reachable}/\emph{well-ordered}/{\em conti\-guous}/{\em nor\-ma\-li\-zed} if for each $i \in \bset{\oddlength}$, the layer $\alayer_i$ is reachable/well-ordered/contiguous/nor\-ma\-li\-zed. 
Note that $\odd$ is normalized if and only if it is both well-ordered and contiguous.

\subsubsection*{Minimized ODDs.}

Let $\alphabet$ be an alphabet, $\width,\oddlength\in\pN$, and let
$\aodd=\alayer_{1}\cdots\alayer_{\oddlength}$ be a deterministic, complete ODD
in $\cdodddefiningset{\alphabet}{\width}{\oddlength}$. We say that $\aodd$ is
\emph{minimized} if for each $\width'\in\pN$ and each
$\bodd=\blayer_{1}\cdots\blayer_{\oddlength}\in\cdodddefiningset{\alphabet}{\width'}{\oddlength}$,
with $\oddlang{\aodd} = \oddlang{\bodd}$, we have that $\numberstates{\aodd}
\leq \numberstates{\bodd}$.  In other words, $\aodd$ is minimized if no
deterministic, complete ODD with the same language as $\aodd$ has less states
than $\aodd$. The following theorem is the analog of Theorem~\ref{theorem:MinimizationAutomaton}
in the realm of the theory of ordered decision diagrams.

\begin{theorem}
\label{theorem:MinimizationODDs}
Let $\alphabet$ be an alphabet, $\width,\oddlength\in\pN$, and let $\aodd$ be an ODD in $\odddefiningset{\alphabet}{\width}{\oddlength}$. 
There exists a unique minimized ODD~ $\canonizationfunction(\aodd) \in \cdodddefiningset{\alphabet}{2^{\width}}{\oddlength}$ 
such that $\canonizationfunction(\aodd)$ is deterministic, complete, normalized and satisfies
$\oddlang{\canonizationfunction(\aodd)} = \oddlang{\aodd}$. Additionally, if $\aodd\in \cdodddefiningset{\alphabet}{\width}{\oddlength}$ then 
$\canonizationfunction(\aodd)\in \cdodddefiningset{\alphabet}{\width}{\oddlength}$.  
\end{theorem}

We call the ODD~ $\canonizationfunction(\aodd)$ of Theorem \ref{theorem:MinimizationODDs} 
the {\em canonical form } of $\aodd$. We note that $\canonizationfunction(\aodd)$ is unique not 
only up to isomorphism, but also unique up to equality. In particular, this implies that for 
each alphabet $\alphabet$, each $\width,\width',\oddlength\in \N$, and each two ODDs 
$\aodd\in \odddefiningset{\alphabet}{\width}{\oddlength}$ and 
$\bodd\in \odddefiningset{\alphabet}{\width'}{\oddlength}$ with 
$\oddlang{\aodd} = \oddlang{\bodd}$, we have that $\canonizationfunction(\aodd) = \canonizationfunction(\bodd)$.
The construction of $\canonizationfunction(\aodd)$ follows a similar process to the construction of 
canonical forms of OBDDs with a fixed variable, or equivalently, read-once oblivious
branching programs \cite{wegener2000branching}.

\subsection{Regular Transductions}
\label{subsection:Transductions}
Let $\aalphabet$ and $\balphabet$ be two alphabets. In this work, a 
\emph{$(\aalphabet,\balphabet)$-transduction} is a binary relation $\transduction\subseteq \aalphabet^{+}\cartesianproduct\balphabet^{+}$ where 
$|\astring|=|\bstring|$ for each $(\astring,\bstring)\in \transduction$. 
We let $$\image(\transduction) \defeq \set{\bstring\in \alphabet_2^{+} \setst \exists\, \astring\in \alphabet_{1}^{+}, \tuple{\astring,\bstring}\in \transduction}$$ be the \emph{image} of $\transduction$, and we let $$\domain(\transduction) \defeq \set{\astring\in \alphabet_1^{+} \setst \exists\, \bstring\in \alphabet_{2}^{+}, \tuple{\astring,\bstring}\in \transduction}$$ be the \emph{domain} of $\transduction$. 
We say that a $(\aalphabet,\balphabet)$-transduction $\transduction$ is \emph{functional}
if, for each string $\astring\in \aalphabet^{+}$, there exists at most one string
$\bstring\in \balphabet^{+}$ such that $\tuple{\astring,\bstring}\in \transduction$.

Let $\aalphabet$, $\balphabet$ and $\calphabet$ be three (not-necessarily distinct) alphabets. 
If $\atransduction$ is a $(\aalphabet,\balphabet)$-transduction and $\btransduction$ is a $(\balphabet,\calphabet)$-transduction, then the \emph{composition} of $\atransduction$ with $\btransduction$ is defined as the $(\aalphabet,\calphabet)$-transduction $$\atransduction\composition\btransduction \defeq \set{\tuple{\astring,\cstring}\in\aalphabet^{+}\cartesianproduct\calphabet^{+}\setst \exists\, \bstring \in \balphabet^{+}, \tuple{\astring,\bstring} \in \atransduction \text{ and } \tuple{\bstring,\cstring}\in\btransduction}\text{.}$$ 

For each language $\lang\subseteq \aalphabet^{+}$, we let $$\duplicate{\lang} \defeq \set{\tuple{\astring,\astring} \setst \astring\in \lang}$$ be the \emph{$(\aalphabet,\aalphabet)$-transduction derived from $\lang$}. 
Then, for each language $\lang\subseteq \aalphabet^{+}$ and each $(\aalphabet,\balphabet)$-transduction $\atransduction$, we let $$\atransduction(\lang) \defeq \image(\duplicate{\lang}\circ \atransduction) = \set{\bstring\in \balphabet^{+} \setst \exists\, \astring\in \lang,\, \tuple{\astring,\bstring}\in \transduction}$$ be the image of $\lang$ under $\atransduction$.

\subsubsection*{Tensor Product.}
Let $\alphabet_{1},\ldots,\alphabet_{\arity}$ be $\arity$ alphabets and $\stringlength \in \pN$. 
For each $i\in\bset{\arity}$, let $\astring_{i}=\asymbol_{i,1}\cdots \asymbol_{i,\stringlength}$ be a string of length $\stringlength$ over the alphabet $\alphabet_{i}$. 
The \emph{tensor product} of $\astring_1,\ldots,\astring_{\arity}$ is defined as the string $$\astring_1\tensorproduct\cdots\tensorproduct\astring_{\arity} \defeq (\asymbol_{1,1},\ldots,\asymbol_{\arity,1}) \cdots (\asymbol_{\stringlength,1},\ldots,\asymbol_{\stringlength,\arity})\text{}$$ of length $\stringlength$ over the alphabet $\alphabet_{1}\cartesianproduct\cdots\cartesianproduct\alphabet_{\arity}$.  
For each $i\in\bset{\arity}$, let $\lang_{i}\subseteq\alphabet_{i}^{+}$ be a language over $\alphabet_{i}$.
The \emph{tensor product} of $\lang_1, \ldots, \lang_{\arity}$ is defined as the language 
\begin{align*}
    \lang_1\tensorproduct \cdots \tensorproduct \lang_{\arity} \defeq \set{\astring_{1}\tensorproduct \cdots \tensorproduct \astring_{\arity} \setst \abs{\astring_1}=\cdots=\abs{\astring_\arity},\, \astring_{i}\in \lang_{i} \text{ for each } i\in\bset{\arity}}\text{.}
\end{align*}

\subsubsection*{Regular transductions.}
For $\alpha\in \pN$, we say that a $(\aalphabet,\balphabet)$-transduction $\atransduction$ is {\em $\alpha$-regular} if the language 
$$\transductionlang(\transduction) \defeq \set{\astring \otimes \bstring \setst \tuple{\astring,\bstring}\in \transduction} \subseteq (\alphabet_1\times \alphabet_2)^{+}$$ is $\alpha$-regular. The following proposition states some straightforward
quantitative properties of regular transductions.

\begin{proposition}\label{proposition:PropertiesTransductions}
Let $\aalphabet,\balphabet$ and $\calphabet$ be three alphabets, $\atransduction$ be an $\agrowthconstant$-regular $(\aalphabet,\balphabet)$-transduction, $\btransduction$ be a $\bgrowthconstant$-regular $(\balphabet,\calphabet)$-transduction, and let $\lang\subseteq \aalphabet^{+}$ be a $\cgrowthconstant$-regular language, for some $\agrowthconstant,\bgrowthconstant,\cgrowthconstant\in\pN$. The following statements hold.
\begin{enumerate}
    \item \label{item:automaton_image_domain} The languages $\image(\atransduction)$ and $\domain(\atransduction)$ are $\agrowthconstant$-regular. 
	\item \label{item:automaton_composition} The composition $\atransduction\circ \btransduction$ is $(\agrowthconstant\cdot \bgrowthconstant)$-regular.
	\item \label{item:automaton_duplicate} The transduction $\duplicate{\lang}$ is $\cgrowthconstant$-regular. 
	\item \label{item:automaton_lang_transduction} The language $\transduction(\lang)$ is $(\cgrowthconstant\cdot\agrowthconstant)$-regular. 
\end{enumerate}
\end{proposition}
\begin{proof}
Let $\finiteautomaton_{\atransduction}$ be a finite automaton with $\agrowthconstant$ states and language $\automatonlang{\finiteautomaton_{\atransduction}}=\transductionlang(\atransduction)$, $\finiteautomaton_{\btransduction}$ be a finite automaton with $\bgrowthconstant$ states and language $\automatonlang{\finiteautomaton_{\btransduction}}=\transductionlang(\btransduction)$, and let $\finiteautomaton_{\lang}$ be a finite automaton with $\cgrowthconstant$ states and language $\automatonlang{\finiteautomaton_{\lang}}=\lang$. 
Note that, such automata $\finiteautomaton_{\atransduction}$, $\finiteautomaton_{\btransduction}$ and $\finiteautomaton_{\lang}$ exist, since by hypothesis $\atransduction$ is $\agrowthconstant$-regular, $\btransduction$ is $\bgrowthconstant$-regular and $\lang$ is $\cgrowthconstant$-regular, respectively.

\begin{enumerate}
    \item We let $\finiteautomaton_{\image(\atransduction)}$ and $\finiteautomaton_{\domain(\atransduction)}$ be the finite automata over the alphabets $\balphabet$ and $\aalphabet$, respectively, defined exactly as $\finiteautomaton_{\atransduction}$ except for their transition sets, which is defined as follows: $$\automatontransitions(\finiteautomaton_{\image(\atransduction)})= \set{\tuple{\astate,\bsymbol,\bstate} \setst \exists\, \asymbol\in\aalphabet, \tuple{\astate,\tuple{\asymbol,\bsymbol},\bstate}\in \automatontransitions(\finiteautomaton_{\atransduction})}\text{, and}$$
    $$\automatontransitions(\finiteautomaton_{\domain(\atransduction)})= \set{\tuple{\astate,\asymbol,\bstate} \setst \exists\, \bsymbol\in\balphabet, \tuple{\astate,\tuple{\asymbol,\bsymbol},\bstate}\in \automatontransitions(\finiteautomaton_{\atransduction})}\text{.}$$ 
    Clearly, $\finiteautomaton_{\image(\atransduction)}$ and $\finiteautomaton_{\domain(\atransduction)}$ have at most $\agrowthconstant$ states each. 
    Moreover, $\finiteautomaton_{\image(\atransduction)}$ accepts a string $\bstring \in \balphabet^{+}$ if and only if there exists a string $\astring\in\aalphabet^{+}$ such that $\astring\tensorproduct\bstring\in\transductionlang(\atransduction)$. 
    Analogously, one can verify that $\finiteautomaton_{\domain(\atransduction)}$ accepts a string $\astring \in \aalphabet^{+}$ if and only if there exists a string $\bstring\in\balphabet^{+}$ such that $\astring\tensorproduct\bstring\in\transductionlang(\atransduction)$. 
    Therefore, the language of $\finiteautomaton_{\image(\atransduction)}$ is $\automatonlang{\finiteautomaton_{\image(\atransduction)}}=\image(\atransduction)$, and the language of $\finiteautomaton_{\domain(\atransduction)}$ is $\automatonlang{\finiteautomaton_{\domain(\atransduction)}}=\domain(\atransduction)$. \vspace{1.0ex}

    \item We let $\finiteautomaton_{\atransduction\composition\btransduction}$ be the finite automaton over the alphabet $\aalphabet\cartesianproduct\calphabet$, with state set $\automatonstates(\finiteautomaton_{\atransduction\composition\btransduction})=\automatonstates(\finiteautomaton_{\atransduction})\cartesianproduct\automatonstates(\finiteautomaton_{\btransduction})$, initial state set $\automatoninitialstates(\finiteautomaton_{\atransduction\composition\btransduction})=\automatoninitialstates(\finiteautomaton_{\atransduction})\cartesianproduct\automatoninitialstates(\finiteautomaton_{\btransduction})$, final state set $\automatonfinalstates(\finiteautomaton_{\atransduction\composition\btransduction})=\automatonfinalstates(\finiteautomaton_{\atransduction})\cartesianproduct\automatonfinalstates(\finiteautomaton_{\btransduction})$ and transition set 
    $$\begin{multlined}[t][0.90\textwidth]
        \automatontransitions(\finiteautomaton_{\atransduction\composition\btransduction})=\set{\tuple{\tuple{\cstate,\cstate'},\tuple{\asymbol,\pbsymbol},\tuple{\astate,\bstate}} \setst \exists\,\bsymbol\in\balphabet,\\ \tuple{\cstate,\tuple{\asymbol,\bsymbol},\astate}\in\automatontransitions(\finiteautomaton_{\atransduction}),\tuple{\cstate',\tuple{\bsymbol,\pbsymbol},\bstate}\in\automatontransitions(\finiteautomaton_{\btransduction})}\text{.}
    \end{multlined}$$
    
    We remark $\finiteautomaton_{\atransduction\composition\btransduction}$ is a finite automaton with at most ($\agrowthconstant\cdot\bgrowthconstant$) states. 
    Moreover, $\finiteautomaton_{\atransduction\composition\btransduction}$ accepts a string $\astring\tensorproduct\cstring \in (\aalphabet\cartesianproduct\calphabet)^{+}$ if and only if there exists $\bstring\in\balphabet^{+}$ such that $\astring\tensorproduct\bstring\in\transductionlang(\atransduction)$ and $\bstring\tensorproduct\cstring\in\transductionlang(\btransduction)$. 
    Therefore, the language of $\finiteautomaton_{\atransduction\composition\btransduction}$ is $\automatonlang{\finiteautomaton_{\atransduction\composition\btransduction}}=\transductionlang(\atransduction\composition\btransduction)$.  \vspace{1.0ex}

    \item We let $\finiteautomaton_{\duplicate{\lang}}$ be the finite automata over the alphabet $\aalphabet\cartesianproduct\aalphabet$ defined exactly as $\finiteautomaton_{\lang}$ except for its transition set, which is defined as follows: $$\automatontransitions(\duplicate{\lang})= \set{\tuple{\astate,\tuple{\asymbol,\asymbol},\bstate} \setst \tuple{\astate,\asymbol,\bstate}\in \automatontransitions(\finiteautomaton_{\lang})}\text{.}$$
    Clearly, $\finiteautomaton_{\duplicate{\lang}}$ has at most $\cgrowthconstant$ states. 
    Moreover, $\finiteautomaton_{\duplicate{\lang}}$ accepts a string $\astring\tensorproduct\bstring \in (\aalphabet\cartesianproduct\aalphabet)^{+}$ if and only if $\bstring=\astring$ and $\astring\in\transductionlang(\atransduction)$. 
    Therefore, the language of $\finiteautomaton_{\duplicate{\lang}}$ is $\automatonlang{\finiteautomaton_{\duplicate{\lang}}}=\transductionlang(\duplicate{\lang})$.  \vspace{1.0ex}

    \item We let $\finiteautomaton_{\transduction(\lang)}$ be the finite automaton over the alphabet $\aalphabet\cartesianproduct\calphabet$ such that $\finiteautomaton_{\transduction(\lang)} = \finiteautomaton_{\image(\duplicate{\lang}\composition\atransduction)}$. 
    Based on (\ref{item:automaton_composition})--(\ref{item:automaton_lang_transduction}), $\finiteautomaton_{\transduction(\lang)}$ is a finite automaton with at most $(\cgrowthconstant\cdot\agrowthconstant)$ states and with language $\automatonlang{\finiteautomaton_{\transduction(\lang)}}=\transduction(\lang)=\image(\duplicate{\lang}\composition\atransduction)$.\qedhere 
\end{enumerate}
\end{proof}

\section{Second-Order Finite Automata}
\label{section:SecondOrderFiniteAutomata}

In this section, we formally define the main object of study of this work, namely,
the notion of second-order finite automata.

\begin{definition}[Second-Order Finite Automata]\label{definition:SecondOrderAutomaton} 
Let $\alphabet$ be an alphabet and $\width \in \pN$. 
A finite automaton $\finiteautomaton$ over the alphabet $\layeralphabet{\alphabet}{\width}$ is called 
	a \emph{$(\alphabet,\width)$-second-order finite automaton} (SOFA) if 
$\automatonlang{\finiteautomaton} \subseteq \odddefiningsetstar{\alphabet}{\width}$. 
\end{definition} 

In other words, a $(\alphabet,\width)$-second-order finite automaton $\finiteautomaton$ is a finite automaton over the 
alphabet $\layeralphabet{\alphabet}{\width}$ such that each string $\aodd = \alayer_1\cdots\alayer_{\oddlength}$ in 
$\automatonlang{\finiteautomaton}$ is a $(\alphabet,\width)$-ODD, for some $\oddlength\in\pN$. 

From now on, for every $(\alphabet,\width)$-second-order finite automaton $\finiteautomaton$, we may refer to $\automatonlang{\finiteautomaton}$ as the {\em first language} of $\finiteautomaton$. 
Since each string $\aodd\in \automatonlang{\finiteautomaton}$ is a $(\alphabet,\width)$-ODD, we can also associate with $\finiteautomaton$ a 
{\em second language}, denoted by $\automatonlanghigher{\finiteautomaton}{2}$, which 
consists of the set of languages accepted by ODDs in $\automatonlang{\finiteautomaton}$.
More precisely, the {\em second language} of a $(\alphabet,\width)$-second-order finite automaton $\finiteautomaton$ is defined as the
set $$\automatonlanghigher{\finiteautomaton}{2} \defeq \set{\oddlang{\aodd} \setst \aodd\in \automatonlang{\finiteautomaton}}\text{.}$$ 
Note that $\automatonlanghigher{\finiteautomaton}{2}$ is a possibly infinite 
subset of $\bigcup_{\oddlength\in \pN} \powerset_{\oddlength}(\alphabet)$. We say that a subset 
$\mathcal{X} \subseteq \bigcup_{\oddlength\in \pN} \powerset_{\oddlength}(\alphabet)$ is {\em regular-decisional} if 
there is a second-order finite automaton $\finiteautomaton$ such that $\mathcal{X} = \automatonlanghigher{\finiteautomaton}{2}$.

\begin{lemma}\label{lemma:LayerSubsetRegular}
	Let $\alphabet$ be an alphabet and $\width\in\pN$.
	For each $\subsetlayers\subseteq\layeralphabet{\alphabet}{\width}$, there exists a $(\alphabet,\width)$-second-order finite automaton $\finiteautomaton_{\subsetlayers}$ with $(\abs{\subsetlayers}+1)$ states such that $\oddlang{\finiteautomaton_{\subsetlayers}} = \subsetlayersstar$. 
\end{lemma}
\begin{proof}
	Let $\finiteautomaton_{\subsetlayers}$ be the $(\alphabet,\width)$-second-order finite automaton over the alphabet $\subsetlayers$, with state set $\automatonstates(\finiteautomaton_{\subsetlayers}) = \set*{\astate} \cup \set{\astate_{\layer} \setst \layer \in \subsetlayers}$, initial state set $\automatoninitialstates(\finiteautomaton_{\subsetlayers}) = \set{\astate}$, final state set $\automatonfinalstates(\finiteautomaton_{\subsetlayers}) = \set{\astate_{\layer} \in \automatonstates(\finiteautomaton_{\subsetlayers}) \setst \layerfinalflag(\layer) = \true}\text{}$ and transition set
$
\automatontransitions(\finiteautomaton_{\subsetlayers})\; = \; \set{\tuple{\astate, \layer, \astate_{\layer}} \setst \layer \in \subsetlayers, \layerinitialflag(\layer) = \true} \; \cup \;  \set{\tuple{\astate_{\alayer}, \alayer, \astate_{\blayer}} \setst \alayer, \blayer \in \subsetlayers,\, \layerleftfrontier(\blayer) = \layerrightfrontier(\alayer), \layerfinalflag(\alayer) = \false, \layerinitialflag(\blayer) = \false}\text{.}
$
Since each transition is labeled with some element from $\subsetlayers$, it should be clear that 
	$\oddlang{\finiteautomaton_{\subsetlayers}} \subseteq \subsetlayersstar$. Now, let $\oddlength\in \N$ and  $\aodd = \alayer_1\alayer_2\dots \alayer_k$ be 
an ODD in $\subsetlayersstar$. Then it should be clear that the sequence of transitions 
	$\tuple{\astate, \alayer_1, \astate_{\alayer_1}}  \tuple{\astate_{\alayer_1}, \alayer_2, \astate_{\alayer_2}} \dots 
	\tuple{\astate_{\alayer_{{\oddlength}-1}}, \alayer_{\oddlength}, \astate_{\alayer_{\oddlength}}}$ is an accepting sequence
	in $\finiteautomaton_{\subsetlayers}$. This implies that $\oddlang{\finiteautomaton_{\subsetlayers}} \supseteq \subsetlayersstar$. 
\end{proof}

The following Corollary is an immediate consequence of Lemma \ref{lemma:LayerSubsetRegular} and Observation \ref{observation:NumberLayers}. 

\begin{corollary}
\label{corollary:AllODDs}
Let $\alphabet$ be an alphabet, and $\width \in \pN$. 
\begin{enumerate}
\item The $(\alphabet,\width)$-SOFA $\finiteautomaton_{\layeralphabet{\alphabet}{\width}}$ has $2^{O(|\alphabet|\cdot \width^2)}$ states 
and $\automatonlang{\finiteautomaton_{\layeralphabet{\alphabet}{\width}}} = \odddefiningsetstar{\alphabet}{\width}$. 
\item The $(\alphabet,\width)$-SOFA $\finiteautomaton_{\cdlayeralphabet{\alphabet}{\width}}$ has $2^{O(|\alphabet|\cdot \width\log \width)}$ states 
and $\automatonlang{\finiteautomaton_{\cdlayeralphabet{\alphabet}{\width}}} = \cdodddefiningsetstar{\alphabet}{\width}$.  
\end{enumerate}
\end{corollary}

\subsubsection*{\bf Example 1: The Even Language}
In Figure~\ref{fig:EvenAutomaton}, we depict a $(\set{0,1},2)$-second-order finite automaton $\evenautomaton$ whose second language consists of all (length-typed) sets $$\evenlanguage_{\oddlength} = \tuple{\oddlength,\set{\astring=\asymbol_{1}\cdots\asymbol_{\oddlength}\in\set{0,1}^{\oddlength}\setst \asymbol_{1}+\cdots+\asymbol_{\oddlength}\equiv 0 \pmod 2}}$$ of all binary strings of length $\oddlength$ with an even number of occurrences of the symbol `1', for each $\oddlength\in\pN$. 
Note that, for each $\oddlength\in\pN$, $\evenautomaton$ accepts a unique $(\set{0,1},2)$-ODD of length $\oddlength$, whose language is $\evenlanguage_{\oddlength}$. 
In particular, the language $\evenlanguage_{5}$ is represented by the ODD depicted in Figure~\ref{fig:evenODD}, which is accept by $\finiteautomaton$ upon following the sequence of states $\astate_{0},\astate_{1},\astate_{1},\astate_{1},\astate_{2}$. 

\begin{figure}[ht]\centering
    \includegraphics[scale=1.4]{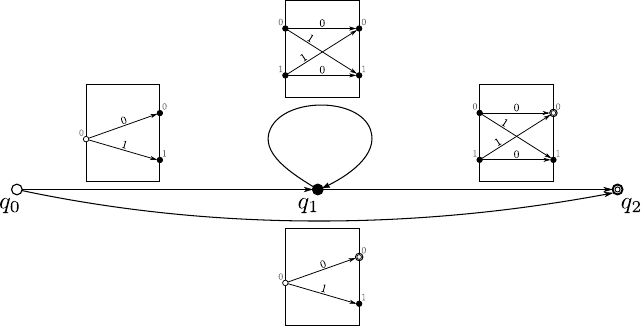}
    \caption{A $(\set{0,1},2)$-second-order finite automaton $\evenautomaton$ with second language $\automatonlanghigher{\evenautomaton}{2} = \set{\evenlanguage_{\oddlength}\setst \oddlength\in \pN}$. }\label{fig:EvenAutomaton}   
\end{figure}

\subsubsection*{\bf Example 2: The Hypercube Language.}
The \emph{hypercube of dimension $\oddlength$} can be defined as the graph $\newhypercubegraph_{\oddlength}$ with vertex set $\vertexset(\newhypercubegraph) = \set{0,1}^{\oddlength}$ and edge set $$\edgeset(\newhypercubegraph_{\oddlength}) = \set{\tuple{\astring,\astring'} \setst \astring,\astring'\in \set{0,1}^{\oddlength}, \exists!\, j \in \bset{\oddlength}\; \astring_{j}\neq \astring_{j}''}\text{.}$$

Intuitively, vertices of the hypercube $\newhypercubegraph_\oddlength$ are strings in $\set{0,1}^{\oddlength}$ and edges are pairs of strings from $\set{0,1}^{\oddlength}$ that differ in exactly one position. 
From a formal language standpoint, the edge set of the graph $\newhypercubegraph_{\oddlength}$ can be encoded by the language $$\hypercubelanguage_{\oddlength} = \tuple{\oddlength,\set{(\astring_{1},\astring_{1}')\cdots (\astring_{\oddlength},\astring_{\oddlength}') \setst \tuple{\astring,\astring'}\in \edgeset(\newhypercubegraph_{\oddlength})}}\text{,}$$ 
Note that, $\hypercubelanguage_{\oddlength}$ is a language over the alphabet $\set{0,1}^{\times 2}$. 

In Figure \ref{fig:edgeHypercubeAutomaton}, we depict a $(\set{0,1}^{\cartesianproduct 2},2)$-second-order finite automaton $\hypercubeautomaton$ whose second language is $\automatonlanghigher{\hypercubeautomaton}{2} = \set{\hypercubelanguage_{\oddlength} \setst \oddlength\in \pN}$. 
Similarly to the second-order finite automaton illustrated in the previous example, for each $\oddlength\in\pN$, $\hypercubeautomaton$ accepts a unique $(\set{0,1}^{\cartesianproduct 2},2)$-ODD $\odd_{\oddlength}$ of length $\oddlength$, whose language is $\hypercubelanguage_{\oddlength}$. 
In particular, the language $\hypercubelanguage_{5}$ is represented by the ODD $\odd_{5}$ depicted in Figure \ref{fig:edgeHypercubeTuple}, which is accept by $\hypercubeautomaton$ upon following the sequence of states $\astate_{0},\astate_{1},\astate_{1},\astate_{1},\astate_{2}$.  

\begin{figure}[ht]\centering
    \includegraphics[scale=1.4]{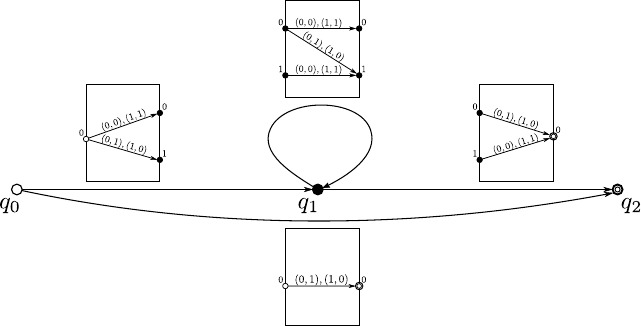}
    \caption{A $(\set{0,1}^{2},2)$-second-order finite automaton $\hypercubeautomaton$ with second language  $\automatonlanghigher{\hypercubeautomaton}{2} = \set{\hypercubelanguage_\oddlength \setst \oddlength\in \pN}$.}\label{fig:edgeHypercubeAutomaton}
\end{figure}
\begin{figure}[ht]\centering
    \includegraphics[scale=1.4]{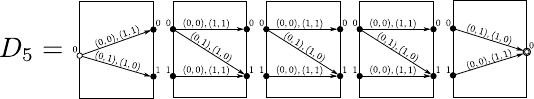}
    \caption{The $(\set{0,1}^{\cartesianproduct 2},2)$-ODD $\odd_{5}$, with language $\oddlang{\aodd_5} = \tuple{5,\hypercubelanguage_5}$, accepted by $\hypercubeautomaton$ upon following the sequence of states $\astate_{0},\astate_{1},\astate_{1},\astate_{1},\astate_{2}$. }\label{fig:edgeHypercubeTuple}   
\end{figure}

\subsubsection*{\bf Main Results.}
The main result of this work (Theorem \ref{theorem:CanonizationSecondOrder}) states that second order finite automata can be canonized with respect to their {\em second} languages. 
In other words, there is an algorithm that sends each SOFA $\finiteautomaton$ to a SOFA $\canonizationfunction_2(\finiteautomaton)$ with 
$\oddlanghigher{\finiteautomaton}{2} = \oddlanghigher{\canonizationfunction_2(\finiteautomaton)}{2}$ in such a way that $\canonizationfunction_2(\finiteautomaton)=\canonizationfunction_2(\finiteautomaton')$
for any SOFA $\finiteautomaton'$ with the same second language as $\finiteautomaton$. Indeed, $\canonizationfunction_2(\finiteautomaton)$
satisfies the following interesting property: $\automatonlang{\canonizationfunction_2(\finiteautomaton)} = \{\canonizationfunction(\aodd)\;:\; \aodd \in \automatonlang{\finiteautomaton}\}$.  
Here, for each ODD $\aodd$, $\canonizationfunction(\aodd)$ denotes the unique deterministic, complete, normalized and minimized ODD with the same language as $\aodd$, as specified in Theorem \ref{theorem:MinimizationODDs}. 
In other words, the first language of $\canonizationfunction_2(\finiteautomaton)$ is precisely the set of canonical forms of ODDs in the first language of $\finiteautomaton$. 

We note that even though $\finiteautomaton$ and $\canonizationfunction_2(\finiteautomaton)$ have 
the same {\em second} language, \ie $\oddlanghigher{\canonizationfunction_2(\finiteautomaton)}{2} = \oddlanghigher{\finiteautomaton}{2}$, 
the {\em first} languages of $\finiteautomaton$ and $\canonizationfunction_2(\finiteautomaton)$ may differ. In other words, it may be the case that 
$\automatonlang{\canonizationfunction_2(\finiteautomaton)} \neq \automatonlang{\finiteautomaton}$. 
As a simple example for this observation, let $\aodd$ be an ODD in $\odddefiningsetstar{\alphabet}{\width}$ for some alphabet $\alphabet$ and $\width\in \pN$. 
Let $\finiteautomaton_{\aodd}$ be the second order finite automaton such that 
$\automatonlang{\finiteautomaton_{\aodd}} = \{\aodd\}$. Then the language 
$\automatonlang{\canonizationfunction_2(\finiteautomaton_{\aodd})} =\{\canonizationfunction(\aodd)\}$ is distinct from $\automatonlang{\finiteautomaton_{\aodd}}$ 
whenever $\canonizationfunction(\aodd) \neq \aodd$. Therefore, canonization of a finite automaton $\finiteautomaton$ with respect to its second language $\automatonlanghigher{\finiteautomaton}{2}$
cannot be achieved by simply canonizing $\finiteautomaton$ with respect to its first language $\automatonlang{\finiteautomaton}$ according to 
Theorem \ref{theorem:MinimizationAutomaton}.

The proof of our main result is a direct consequence of the following theorem, stating that the traditional minimization and canonization algorithm for ODDs can be simulated 
in terms of functional regular transductions. 

\begin{theorem}[Canonization as Transduction Theorem]
\label{theorem:CanonizationTransductionTheorem}
Let $\alphabet$ be an alphabet and let $\width\in \pN$. 
\begin{enumerate}
\item The functional transduction $\canonizationtransduction{\alphabet}{\width} = \set{\tuple{\aodd,\canonizationfunction(\aodd)} \setst \aodd\in \odddefiningsetstar{\alphabet}{\width}}$
	is $2^{O(|\alphabet|\cdot \width \cdot 2^{\width})}$-regular. 
\item The functional transduction $\cdcanonizationtransduction{\alphabet}{\width} = \set{\tuple{\aodd,\canonizationfunction(\aodd)} \setst \aodd\in \cdodddefiningsetstar{\alphabet}{\width}}$
	is $2^{O(|\alphabet|\cdot \width\cdot \log \width)}$-regular.
\end{enumerate}
\end{theorem} 

Intuitively, the transduction $\canonizationtransduction{\alphabet}{\width}$ is obtained as a composition of regular transductions that simulate the application of the usual steps in the 
canonization of a single ODD: determinization, elimination of unreachable states, merging of equivalent states and normalization. The transduction $\cdcanonizationtransduction{\alphabet}{\width}$
is obtained by a similar process, except that one may skip the application of the determinization transduction, yielding in this way, 
a more efficient construction. Due to its 
technical nature, the proof of Theorem \ref{theorem:CanonizationTransductionTheorem} will be postponed to
Section \ref{section:CanonizationTransductionTheorem}. 
Next, we show how Theorem \ref{theorem:CanonizationTransductionTheorem} can be used to provide a canonization procedure for second order finite automata. Later, in Section \ref{section:Applications},
we will provide some algorithmic applications of this theorem in the realm of the theory of ODDs of bounded width.

\begin{theorem}[Canonical Form of Canonical Forms Theorem]
\label{theorem:CanonizationSecondOrder}
Let $\alphabet$ be an alphabet (endowed with a total order $\orderalphabet_{\alphabet}\; \subset \alphabet\times \alphabet$), 
$\width\in\pN$, and let $\finiteautomaton$ be a $(\alphabet,\width)$-SOFA. 
One can construct in time $2^{\nstates(\finiteautomaton)\cdot 2^{\BigOh(\abs{\alphabet}\cdot\width\cdot 2^{\width})}}$
a deterministic, complete, normalized $(\alphabet,2^{\width})$-SOFA $\canonizationfunction_2(\finiteautomaton)$ satisfying the following properties. 
\begin{enumerate}
	\item \label{FAcanonizationOne} $\automatonlang{\canonizationfunction_2(\finiteautomaton)} = \set{\canonizationfunction(\aodd) \setst \aodd \in \automatonlang{\finiteautomaton}}$;  
	\item \label{FAcanonizationTwo} $\oddlanghigher{\canonizationfunction_2(\finiteautomaton)}{2}= \automatonlanghigher{\finiteautomaton}{2}$; 
	\item \label{FAcanonizationThree} For each $\width' \in \pN$ and each $(\alphabet,\width')$-SOFA $\finiteautomaton'$, if 
$\automatonlanghigher{\finiteautomaton'}{2} = \automatonlanghigher{\finiteautomaton}{2}$, then $\canonizationfunction_2(\finiteautomaton')= \canonizationfunction_2(\finiteautomaton)$. 
\end{enumerate}
\end{theorem}
\begin{proof}
Let $\finiteautomaton$ be a $(\alphabet,\width)$-SOFA and $\canonizationtransduction{\alphabet}{\width}$ be the $(\layeralphabet{\alphabet}{\width},\cdlayeralphabet{\alphabet}{\width})$-transduction
specified in Theorem \ref{theorem:CanonizationTransductionTheorem}. Then, the image of $\automatonlang{\finiteautomaton}$ under the transduction $\canonizationtransduction{\alphabet}{\width}$ 
is the language $\canonizationtransduction{\alphabet}{\width}(\automatonlang{\finiteautomaton}) = \{\canonizationfunction(\aodd)\;:\; \aodd \in \automatonlang{\finiteautomaton}\}$. 
Here, for each ODD $\aodd\in \odddefiningsetstar{\alphabet}{\width}$, $\canonizationfunction(\aodd)\in \cdodddefiningsetstar{\alphabet}{2^{\width}}$ denotes the unique ODD with minimum 
number of states such that $\canonizationfunction(\aodd)$ is deterministic, complete, normalized and satisfies $\oddlang{\canonizationfunction(\aodd)} = \oddlang{\aodd}$, as specified in 
Theorem \ref{theorem:MinimizationODDs}. 
Since $\canonizationtransduction{\alphabet}{\width}$ is $2^{O(|\alphabet|\cdot \width \cdot 2^{\width})}$-regular, it follows from 
Proposition \ref{proposition:PropertiesTransductions}.(\ref{item:automaton_lang_transduction}), that one can construct a 
$(\alphabet,2^{\width})$-SOFA $\finiteautomaton^{\dagger}$ with $\nstates(\finiteautomaton)\cdot 2^{\BigOh(|\alphabet|\cdot \width\cdot2^{\width})}$ states such that 
$\automatonlang{\finiteautomaton^{\dagger}} = \canonizationtransduction{\alphabet}{\width}(\automatonlang{\finiteautomaton})$. 
Now, let $\canonizationfunction(\finiteautomaton^{\dagger})$ be the unique finite automaton with minimum number of states such that 
$\canonizationfunction(\finiteautomaton^{\dagger})$ is deterministic, complete, normalized and satisfies 
$\automatonlang{\canonizationfunction(\finiteautomaton^{\dagger})} = \automatonlang{\finiteautomaton^{\dagger}}$, as specified in Theorem \ref{theorem:MinimizationAutomaton}.
Then $\canonizationfunction(\finiteautomaton^{\dagger})$ can be constructed in time  $2^{\nstates\cdot 2^{\BigOh(|\alphabet|\cdot \width\cdot2^{\width})}}$ 
by the applying the standard power-set construction to $\finiteautomaton^{\dagger}$, followed by a DFA minimization algorithm, such as Hopcroft's algorithm.
Now, by defining $\canonizationfunction_2(\finiteautomaton)$ as $\canonizationfunction(\finiteautomaton^{\dagger})$, we have 
that $\automatonlang{\canonizationfunction_2(\finiteautomaton)} = \set{\canonizationfunction(\aodd) \setst \aodd \in \automatonlang{\finiteautomaton}}$, and therefore, 
Condition \ref{FAcanonizationOne} is satisfied. This immediately implies that 
$\oddlanghigher{\canonizationfunction_2(\finiteautomaton)}{2}= \automatonlanghigher{\finiteautomaton}{2}$, since each ODD $\aodd\in \automatonlang{\finiteautomaton}$
has the same language as its canonical form $\canonizationfunction(\aodd)$ in $\automatonlang{\canonizationfunction_2(\finiteautomaton)}$. Therefore, Condition \ref{FAcanonizationTwo} is also satisfied.
Finally,  $\canonizationfunction_2(\finiteautomaton) = \canonizationfunction_2(\finiteautomaton')$ for any $(\alphabet,\width')$-SOFA $\finiteautomaton'$ satisfying 
$\oddlanghigher{\finiteautomaton'}{2}= \oddlanghigher{\finiteautomaton}{2}$, since for any two ODDs $\aodd\in \automatonlang{\finiteautomaton}$ and $\aodd'\in \automatonlang{\finiteautomaton'}$, 
$\oddlang{\aodd} = \oddlang{\aodd'}$ if and only if $\canonizationfunction(\aodd) = \canonizationfunction(\aodd')$. Therefore, Condition \ref{FAcanonizationThree} is also satisfied. 
\end{proof}

Let $\finiteautomaton$ be a $(\alphabet,\width)$-SOFA. We call the $(\alphabet,2^{\width})$-SOFA
$\canonizationfunction_2(\finiteautomaton)$ specified in Theorem~\ref{theorem:CanonizationSecondOrder} 
the \emph{second canonical form} of $\finiteautomaton$. We note that if all ODDs in the language $\finiteautomaton$ 
are deterministic and complete, then $\canonizationfunction_2(\finiteautomaton)$ is actually a $(\alphabet,\width)$-SOFA, 
and a faster canonization algorithm can be obtained, since in this case, the transduction $\canonizationtransduction{\alphabet}{\width}$
used in the proof of Theorem \ref{theorem:CanonizationSecondOrder} can be replaced by the transduction 
$\cdcanonizationtransduction{\alphabet}{\width}$, which is $2^{O(|\alphabet|\cdot \width\cdot \log \width)}$-regular. 

\begin{observation}
\label{observation:BetterConstruction}
	If $\finiteautomaton$ is a $(\alphabet,\width)$-SOFA such that $\automatonlang{\finiteautomaton}\subseteq \cdodddefiningsetstar{\alphabet}{\width}$, 
then $\canonizationfunction_2(\finiteautomaton)$ is also a $(\alphabet,\width)$-SOFA and can be constructed in time 
$2^{\nstates(\finiteautomaton)\cdot 2^{\BigOh(\abs{\alphabet}\cdot\width\log\width)}}$. 
\end{observation}

An immediate consequence of Theorem \ref{theorem:CanonizationTransductionTheorem} and of Proposition \ref{proposition:PropertiesTransductions}.(1) is that for each alphabet $\alphabet$, and each $\width\in \pN$, the set
of canonical forms of ODDs in $\odddefiningsetstar{\alphabet}{\width}$ is a regular set. The same holds for the set of canonical forms of ODDs in $\cdodddefiningsetstar{\alphabet}{\width}$. 

\begin{corollary}
\label{corollary:RegularCanonicalForms}
Let $\alphabet$ be an alphabet and $\width\in \pN$. 
\begin{enumerate}
\item The language $\image(\canonizationtransduction{\alphabet}{\width}) = \{\canonizationfunction(\aodd)\;:\; \aodd\in \odddefiningsetstar{\alphabet}{\width}\}$ is $2^{O(|\alphabet|\cdot \width\cdot 2^{\width})}$-regular. 
\item The language $\image(\cdcanonizationtransduction{\alphabet}{\width}) = \{\canonizationfunction(\aodd)\;:\; \aodd\in \cdodddefiningsetstar{\alphabet}{\width}\}$ is $2^{O(|\alphabet|\cdot \width\cdot \log \width)}$-regular.  
\end{enumerate}
\end{corollary}

\section{Closure Properties}
\label{section:ClosureProperties}

\subsection{Basic Closure Properties}

Theorem~\ref{theorem:CanonizationSecondOrder} implies that 
regular-decisional subsets of
$\bigcup_{\oddlength\in \pN} \powerset_{\oddlength}(\alphabet)$ are closed under Boolean operations such as union, intersection and even
a suitable notion of bounded width complementation. These closure properties are formally stated in Theorem~\ref{theorem:ClosureProperties} below. 
Let $\alphabet$ be an alphabet and $\width\in\pN$. 
We denote by $$\alldetlanguages{\alphabet}{\width} \defeq \{\oddlang{\aodd}\setst \aodd\in \cdodddefiningsetstar{\alphabet}{\width}\}$$ the set of
all sets of strings accepted by some deterministic, complete $(\alphabet,\width)$-ODD.
Moreover, given a subset $\secondlanguage \subseteq \bigcup_{\oddlength\in \pN} \powerset_{\oddlength}(\alphabet)$, we denote by $\boundedcomplement{\secondlanguage}{\width}\defeq \alldetlanguages{\alphabet}{\width} \backslash \secondlanguage$ the {\em width-$\width$ complement} of $\secondlanguage$.

\begin{theorem}\label{theorem:ClosureProperties}
Let $\alphabet$ be an alphabet, $\width\in \pN$, and let $\finiteautomaton$, $\finiteautomaton_1$ and $\finiteautomaton_2$ be $(\alphabet,\width)$-second-order finite automata. The following statements hold. 
\begin{enumerate}
	\item \label{item:SecondFAcap} There is a $(\alphabet,2^{\width})$-second-order finite automaton $\secondcap{\finiteautomaton_1}{\finiteautomaton_2}$ such that $$\automatonlanghigher{\secondcap{\finiteautomaton_1}{\finiteautomaton_2}}{2} = \automatonlanghigher{\finiteautomaton_1}{2} \cap \automatonlanghigher{\finiteautomaton_2}{2}\text{.}$$ 
	
	\item \label{item:SecondFAcup} There is a $(\alphabet,2^{\width})$-second-order finite automaton $\secondcup{\finiteautomaton_1}{\finiteautomaton_2}$ such that $$\automatonlanghigher{\secondcup{\finiteautomaton_1}{\finiteautomaton_2}}{2} = \automatonlanghigher{\finiteautomaton_1}{2} \cup \automatonlanghigher{\finiteautomaton_2}{2}\text{.}$$ 
	
	\item \label{item:SecondFAdiff} There is a $(\alphabet,2^{\width})$-second-order finite automaton $\seconddiff{\finiteautomaton_1}{\finiteautomaton_2}$ such that $$\automatonlanghigher{\seconddiff{\finiteautomaton_1}{\finiteautomaton_2}}{2} = \automatonlanghigher{\finiteautomaton_1}{2} \backslash \automatonlanghigher{\finiteautomaton_2}{2}\text{.}$$
	
	\item \label{item:SecondFAdet} There is a $(\alphabet,\width)$-second-order finite automaton $\finiteautomaton(\alphabet,\width)$ such that $$\automatonlanghigher{\finiteautomaton(\alphabet,\width)}{2} = \alldetlanguages{\alphabet}{\width}\text{.}$$
	
	\item \label{item:SecondFAcomp} For each $\width'\in \pN$, there is a $(\alphabet,2^{\max\set{\width,\width'}})$-second-order finite automaton $\secondcomp{\finiteautomaton}{\width'}$ such that 
		$$\automatonlanghigher{\secondcomp{\finiteautomaton}{\width'}}{2} = \boundedcomplement{\automatonlanghigher{\finiteautomaton}{2}}{\width'}\text{.}$$
	
	\item \label{item:SecondFAemptiness} It is decidable whether $\automatonlanghigher{\finiteautomaton_1}{2}\cap \automatonlanghigher{\finiteautomaton_2}{2} = \emptyset$.
	
	\item \label{item:SecondFAcontainement} It is decidable whether $\automatonlanghigher{\finiteautomaton_1}{2}\subseteq \automatonlanghigher{\finiteautomaton_2}{2}$. 
\end{enumerate}
\end{theorem}
\begin{proof}
	Let $\finiteautomaton_1'=\canonizationfunction_2(\finiteautomaton_1)$ and $\finiteautomaton_2'=\canonizationfunction_2(\finiteautomaton_2)$ be the second canonical forms specified in Theorem~\ref{theorem:CanonizationSecondOrder} of the automata $\finiteautomaton_1$ and $\finiteautomaton_2$, respectively. 
	It is well-known that regular languages are closed under intersection, union and complementation~\cite{Hopcroft2007}.
	Consequently, there exist finite automata $\firstcap{\finiteautomaton_1'}{\finiteautomaton_2'}$, $\firstcup{\finiteautomaton_1'}{\finiteautomaton_2'}$ and $\firstcomp{\finiteautomaton_2'}$ over the alphabet $\cdlayeralphabet{\alphabet}{2^{\width}}$, such that $$\automatonlang{\firstcap{\finiteautomaton_1'}{\finiteautomaton_2'}} = \automatonlang{\finiteautomaton_1'} \cap \automatonlang{\finiteautomaton_2'},$$ $\automatonlang{\firstcup{\finiteautomaton_1'}{\finiteautomaton_2'}} = \automatonlang{\finiteautomaton_1'} \cup \automatonlang{\finiteautomaton_2'}$ and $\automatonlang{\firstcomp{\finiteautomaton_2'}} = \cdlayeralphabet{\alphabet}{2^{\width}}^{*} \setminus \automatonlang{\finiteautomaton_2'}$. 
	
	Clearly, $$\automatonlang{\firstcup{\finiteautomaton_1'}{\finiteautomaton_2'}} = \set{\canonizationfunction(\aodd) \setst \aodd \in \automatonlang{\finiteautomaton_{1}}} \cup \set{\canonizationfunction(\bodd) \setst \bodd \in \automatonlang{\finiteautomaton_{2}}}\text{.}$$ 
	Thus, $\secondcup{\finiteautomaton_1}{\finiteautomaton_2} = \firstcup{\finiteautomaton_1'}{\finiteautomaton_2'}$ is a $(\alphabet,2^{\width})$-second-order finite 
	automaton with second language 
	$$\automatonlanghigher{\secondcup{\finiteautomaton_1}{\finiteautomaton_2}}{2} = \automatonlanghigher{\finiteautomaton_1}{2} \cup \automatonlanghigher{\finiteautomaton_2}{2}.$$
	
	Moreover, owing to the fact that any two ODDs with the same language have the same canonical form, one can verify that 
	$$\automatonlang{\firstcap{\finiteautomaton_1'}{\finiteautomaton_2'}} = \set{\canonizationfunction(\aodd) \setst \aodd \in \automatonlang{\finiteautomaton_{1}}, \exists\, \bodd \in \automatonlang{\finiteautomaton_{2}}, \oddlang{\aodd}=\oddlang{\bodd}}\text{.}$$ 
	Thus, $\secondcap{\finiteautomaton_1}{\finiteautomaton_2} = \firstcap{\finiteautomaton_1'}{\finiteautomaton_2'}$ is a $(\alphabet,2^{\width})$-second-order finite automata with second language 
	$$\automatonlanghigher{\secondcap{\finiteautomaton_1}{\finiteautomaton_2}}{2} = \automatonlanghigher{\finiteautomaton_1}{2} \cap \automatonlanghigher{\finiteautomaton_2}{2}.$$ 

	Furthermore, we have that $\seconddiff{\finiteautomaton_1}{\finiteautomaton_2} = \secondcap{\finiteautomaton_1}{\firstcomp{\finiteautomaton_{2}'}}$ is a $(\alphabet,2^{\width})$-second-order finite automata with first language 
	$$\begin{array}{lcl}\automatonlang{\seconddiff{\finiteautomaton_1}{\finiteautomaton_2}} & = & \automatonlang{\finiteautomaton_1'} \cap \automatonlang{\firstcomp{\finiteautomaton_{2}'}} = \automatonlang{\finiteautomaton_1'} \cap \big(\cdlayeralphabet{\alphabet}{2^{\width}}^{*} \setminus \automatonlang{\finiteautomaton_2'}\big)\\ &= & \automatonlang{\finiteautomaton_1'} \setminus \automatonlang{\finiteautomaton_2'}\text{.}
	\end{array}$$
	Thus, since ODDs with the same language have the same canonical form, the second language of $\seconddiff{\finiteautomaton_1}{\finiteautomaton_2}$ is 
	$$\automatonlang{\seconddiff{\finiteautomaton_1}{\finiteautomaton_2}} = \automatonlanghigher{\finiteautomaton_1}{2} \setminus \automatonlanghigher{\finiteautomaton_2}{2}.$$ 
	
	Based on Lemma~\ref{lemma:LayerSubsetRegular}, we let $\finiteautomaton(\alphabet,\width) = \finiteautomaton_{\subsetlayers}$ be the $(\alphabet,\width)$-second-order finite automaton over the alphabet $\subsetlayers$, where $\subsetlayers=\cdlayeralphabet{\alphabet}{\width}$. 
	One can readily verify that $\automatonlanghigher{\finiteautomaton(\alphabet,\width)}{2} =\alldetlanguages{\alphabet}{\width}$.

	Now, let $\finiteautomaton'=\canonizationfunction_2(\finiteautomaton)$ be the second canonical form specified in Theorem~\ref{theorem:CanonizationSecondOrder} of the automaton $\finiteautomaton$. 
	For each $\width'\in \pN$, we let $\secondcomp{\finiteautomaton}{\width'} = \seconddiff{\finiteautomaton(\alphabet,\width')}{\finiteautomaton'}$. 
	It is straightforward that $\secondcomp{\finiteautomaton}{\width'}$ is a $(\alphabet,2^{\max\set{\width,\width'}})$-second-order finite automaton with second language 
	$$\automatonlanghigher{\secondcomp{\finiteautomaton}{\width'}}{2} = \boundedcomplement{\automatonlanghigher{\finiteautomaton}{2}}{\width'}.$$ 

	Finally, we note that deciding whether $\automatonlanghigher{\finiteautomaton_1}{2}\cap\automatonlanghigher{\finiteautomaton_2}{2}=\emptyset$ is equivalent to deciding whether $\automatonlang{\finiteautomaton_1'}\cap\automatonlang{\finiteautomaton_2'}=\emptyset$. 
	Similarly, we have that deciding whether $\automatonlanghigher{\finiteautomaton_1}{2}\subseteq\automatonlanghigher{\finiteautomaton_2}{2}$ is equivalent to deciding whether $\automatonlang{\finiteautomaton_1'}\subseteq\automatonlang{\finiteautomaton_2'}$, which in turn is equivalent to deciding whether $$\automatonlang{\finiteautomaton_1'}\cap(\cdlayeralphabet{\alphabet}{2^{\width}}^{*}\setminus\automatonlang{\finiteautomaton_2'})=\automatonlang{\finiteautomaton_1'}\cap\firstcomp{\finiteautomaton_2'}=\emptyset\text{.}$$ 
	Therefore, since disjointness of regular languages is a decidable problem~\cite{Hopcroft2007}, we obtain that the problems of verifying whether $\automatonlanghigher{\finiteautomaton_1}{2}\cap\automatonlanghigher{\finiteautomaton_2}{2}=\emptyset$ and verifying whether $\automatonlanghigher{\finiteautomaton_1}{2}\subseteq\automatonlanghigher{\finiteautomaton_2}{2}$ are both decidable. 
\end{proof}

We note that all binary operations described in Theorem~\ref{theorem:ClosureProperties} are also defined when $\finiteautomaton_{1}$ is a $(\alphabet,\width_{1})$-second-order finite automaton and $\finiteautomaton_{2}$ is a $(\alphabet,\width_{2})$-second-order finite automaton, for distinct positive integers $\width_{1}$ and $\width_{2}$. Indeed, it suffices to view both finite automata as $(\alphabet,\max\set{\width_{1},\width_{2}})$-second-order finite automata. We also note that the SOFAs $\secondcap{\finiteautomaton_1}{\finiteautomaton_2}$, 
$\secondcup{\finiteautomaton_1}{\finiteautomaton_2}$ and $\seconddiff{\finiteautomaton_1}{\finiteautomaton_2}$ are actually 
$(\alphabet,\width)$-SOFAs if all ODDs in the languages $\automatonlang{\finiteautomaton_1}$ and $\automatonlang{\finiteautomaton_2}$
are deterministic and complete, since in this case one can use the more efficient construction given in
Observation \ref{observation:BetterConstruction}. Finally, it is worth remarking that non-emptiness of intersection of the second languages
of SOFAs is not only decidable, but can be achieved in fixed-parameter tractable time (Observation \ref{observation:FPTIntersection}). 

\begin{observation}
\label{observation:FPTIntersection}
Let $\alphabet$ be an alphabet, and $\width\in \pN$ and $\finiteautomaton_1$ and $\finiteautomaton_2$ be $(\alphabet,\width)$-SOFAs. 
\begin{enumerate}
\item One can determine whether $\automatonlanghigher{\finiteautomaton_1}{2}\cap \automatonlanghigher{\finiteautomaton_2}{2}\neq \emptyset$ in time
$2^{O(|\alphabet|\cdot \width \cdot 2^{\width})}\cdot \nstates(\finiteautomaton_1)\cdot \nstates(\finiteautomaton_2)$. 
\item If all ODDs in $\automatonlang{\finiteautomaton_1}$ and $\automatonlang{\finiteautomaton_2}$ are deterministic and complete, then one
one can can determine whether $\automatonlanghigher{\finiteautomaton_1}{2}\cap \automatonlanghigher{\finiteautomaton_2}{2}\neq \emptyset$
in time $2^{O(|\alphabet|\cdot \width \cdot \log \width)}\cdot \nstates(\finiteautomaton_1)\cdot \nstates(\finiteautomaton_2)$. 
\end{enumerate}
\end{observation}
\begin{proof}
Since $\canonizationtransduction{\alphabet}{\width}$ is $2^{O(|\alphabet|\cdot \width \cdot 2^{\width})}$-regular, for each $i\in \{1,2\}$, 
one can construct from $\finiteautomaton_i$ a finite automaton $\finiteautomaton_i'$  with $2^{O(|\alphabet|\cdot \width \cdot 2^{\width})}\cdot \nstates(\finiteautomaton_i)$ states 
such that $\automatonlang{\finiteautomaton_i'} = \canonizationtransduction{\alphabet}{\width}(\automatonlang{\finiteautomaton_i}) = \{\canonizationfunction(\aodd)\;:\; \aodd\in \automatonlang{\finiteautomaton_i}\}$. 
Therefore, testing whether $\automatonlanghigher{\finiteautomaton_1}{2} \cap \automatonlanghigher{\finiteautomaton_2}{2}\neq \emptyset$ is equivalent to testing whether 
$\automatonlang{\finiteautomaton_1'}\cap \automatonlang{\finiteautomaton_2'}\neq \emptyset$, which can be done in time
$2^{O(|\alphabet|\cdot \width \cdot 2^{\width})}\cdot \nstates(\finiteautomaton_1)\cdot \nstates(\finiteautomaton_2)$. If the languages of the automata
$\finiteautomaton_1$ and $\finiteautomaton_2$ only contain deterministic, complete ODDs, then one can apply a similar argument using the 
transduction $\cdcanonizationtransduction{\alphabet}{\width}$ instead of 
$\canonizationtransduction{\alphabet}{\width}$ to infer that non-emptiness of intersection for the languages $\automatonlanghigher{\finiteautomaton_1}{2}$ and $\automatonlanghigher{\finiteautomaton_2}{2}$ 
can be tested in time $2^{O(|\alphabet|\cdot \width \cdot \log \width)}\cdot \nstates(\finiteautomaton_1)\cdot \nstates(\finiteautomaton_2)$. 
\end{proof}

\subsection{Closure Properties Specific for Language Classes} 
\label{subsection:FurtherClosureProperties}

In this subsection, we show that regular-decisional classes of languages are also closed 
under operations that are specific to language classes. 
Let $\alphabet_1$ and $\alphabet_2$ be alphabets, and $g:\alphabet_1\rightarrow \alphabet_2$ be a map from 
$\alphabet_1$ to $\alphabet_2$. 
Given languages $\lang\subseteq \alphabet_1^+$ and $L'\subseteq \alphabet_2^{+}$, 
we let 
$$g(L) = \{u\;:\; \exists w\in L,\;|u|=|w|,\; u_i=g(w_i) \mbox{ for each } i\in [|u|]\}$$
and
$$g^{-1}(L') = \{u\;:\; \exists w\in L',\; |u|=|w|,\; u_i\in g^{-1}(w_i) \mbox{ for each } i\in [|u|]\}.$$

The following lemma from \cite{deOliveiraOliveira2020symbolic} 
states that several operations that are effective for regular languages 
may be realized on ODDs using maps that act layerwisely. Below, 
for ODDs $\aodd=\alayer_1\alayer_2\dots\alayer_{\oddlength}$ and $\aodd' = \alayer_1'\alayer_2'\dots\alayer_{\oddlength}'$, 
we let $\aodd\otimes \aodd' = (\alayer_1,\alayer_1')(\alayer_2,\alayer_2')\dots (\alayer_{\oddlength},\alayer_{\oddlength}')$.

\begin{lemma}[Simulation Lemma (see Lemma 2 of \cite{deOliveiraOliveira2020symbolic})]
\label{lemma:SimulationLemma}
Let $\alphabet_1$ and $\alphabet_2$ be alphabets, $\width_1,\width_2\in \pN$,
and $g:\alphabet_1\rightarrow \alphabet_2$ be a map from $\alphabet_1$ to $\alphabet_2$. 
There exist maps
\begin{enumerate} 
\item $f_{\cup}:\layeralphabet{\alphabet_1}{\width_1}\times \layeralphabet{\alphabet_2}{\width_2}
	\rightarrow \layeralphabet{\alphabet_1\cup \alphabet_2}{\width_1+\width_2}$,
\item $f_{\cap}:\layeralphabet{\alphabet_1}{\width_1}\times \layeralphabet{\alphabet_2}{\width_2} 
	\rightarrow \layeralphabet{\alphabet_1\cup \alphabet_2}{\width_1\cdot \width_2}$, 
\item $f_{\otimes}:\layeralphabet{\alphabet_1}{\width_1}\times \layeralphabet{\alphabet_2}{\width_2}
	\rightarrow \layeralphabet{\alphabet_1\times \alphabet_2}{\width_1\cdot \width_2}$,
\item $f_{g}:\layeralphabet{\alphabet_1}{\width_1}\rightarrow \layeralphabet{\alphabet_2}{\width_1}$,
\item $f_{g^{-1}}:\layeralphabet{\alphabet_2}{\width_2}\rightarrow \layeralphabet{\alphabet_1}{\width_2}$,
\item $f_{\neg}:\cdlayeralphabet{\alphabet_1}{\width_1}\rightarrow \cdlayeralphabet{\alphabet_1}{\width_1}$,
\end{enumerate}

such that for each $(\alphabet_1,\width_1)$-ODD
$\aodd=\alayer_1\alayer_2\dots \alayer_{\oddlength}$,
each $(\alphabet_2,\width_2)$-ODD
$\aodd' = \alayer_1'\alayer_2'\dots\alayer_{\oddlength}'$, 
and each deterministic, complete $(\alphabet_1,\width_1)$-ODD 
$\aodd''=\alayer_1''\alayer_2''\dots\alayer_{\oddlength}''$,
the following hold. 

\begin{enumerate}
\item $f_{\cup}(\aodd \otimes \aodd') \defeq f_{\cup}(\alayer_1,\alayer_1')f_{\cup}(\alayer_2,\alayer_2')\dots f_{\cup}(\alayer_{\oddlength},\alayer_{\oddlength}')$ is a $(\alphabet_1\cup \alphabet_2,\width_1 + \width_2)$-ODD such that
$$\oddlang{f_{\cup}(\aodd \otimes \aodd')} = \oddlang{\aodd}\cup \oddlang{\aodd'}.$$ 
\item $f_{\cap}(\aodd \otimes \aodd') \defeq f_{\cap}(\alayer_1,\alayer_1')f_{\cap}(\alayer_2,\alayer_2')\dots f_{\cap}(\alayer_{\oddlength},\alayer_{\oddlength}')$ is a $(\alphabet_1\cup \alphabet_2,\width_1\cdot \width_2)$-ODD such that
$$\oddlang{f_{\cap}(\aodd \otimes \aodd')} = \oddlang{\aodd} \cap \oddlang{\aodd'}.$$ 
\item $f_{\otimes}(\aodd \otimes \aodd') \defeq f_{\otimes}(\alayer_1,\alayer_1')f_{\otimes}(\alayer_2,\alayer_2')\dots f_{\otimes}(\alayer_{\oddlength},\alayer_{\oddlength}')$ is a $(\alphabet_1\times \alphabet_2,\width_1\cdot \width_2)$-ODD such that
$$\oddlang{f_{\otimes}(\aodd \otimes \aodd')} = \oddlang{\aodd}\otimes \oddlang{\aodd'}.$$
\item $f_{g}(\aodd) \defeq f_{g}(\alayer_1)f_{g}(\alayer_2)\dots f_{g}(\alayer_{\oddlength})$ is a 
$(\alphabet_2,\width_1)$-ODD such that
$$\oddlang{f_{g}(\aodd)} = g(\oddlang{\aodd}).$$
\item $f_{g^{-1}}(\aodd') \defeq f_{g^{-1}}(\alayer_1')f_{g^{-1}}(\alayer_2')\dots f_{g^{-1}}(\alayer_{\oddlength}')$ is a 
$(\alphabet_1,\width_2)$-ODD such that
$$\oddlang{f_{g^{-1}}(\aodd')} = g^{-1}(\oddlang{\aodd}).$$
\item $f_{\neg}(\aodd) \defeq f_{\neg}(\alayer_1'')f_{\neg}(\alayer_2'')\dots f_{\neg}(\alayer_{\oddlength}'')$ is 
a deterministic, complete $(\alphabet_1,\width)$-ODD such that
$$\oddlang{f_{\neg}(\aodd)} = \alphabet^{k} \backslash \oddlang{\aodd}.$$ 
\end{enumerate}
\end{lemma}

Lemma \ref{lemma:SimulationLemma} immediately implies implies that the collection of regular-decisional
classes of languages is effectively closed under several {\em pointwise} operations, as stated in the
next corollary.

\begin{corollary}
\label{corollary:PointwiseOperations}
Let $\alphabet_1$ and $\alphabet_2$ be alphabets, $\width_1,\width_2\in \pN$,
$g:\alphabet_1\rightarrow \alphabet_2$ be a map from $\alphabet_1$ to $\alphabet_2$, 
$\finiteautomaton$ be a $(\alphabet_1,\width_1)$-SOFA, and $\finiteautomaton'$ be a 
$(\alphabet_2,\width_2)$-SOFA. 
\begin{enumerate}
\item {\bf Pointwise union.} There is a SOFA 
$\finiteautomaton\dot{\cup}\finiteautomaton'$ such that 
$$\automatonlanghigher{\finiteautomaton\dot{\cup}\finiteautomaton'}{2} =
\{\oddlang{\aodd}\cup \oddlang{\aodd'}\;:\; \aodd \in \automatonlang{\finiteautomaton},\;\odd'\in \automatonlang{\finiteautomaton'},
\;\lengthfunction{\aodd}=\lengthfunction{\aodd'}\}.$$
\item {\bf Pointwise intersection.} There is a SOFA $\finiteautomaton\dot{\cap}\finiteautomaton'$, 
$$\automatonlanghigher{\finiteautomaton\dot{\cap}\finiteautomaton'}{2} = 
\{\oddlang{\aodd}\cap \oddlang{\aodd'} \;:\; \aodd\in \automatonlang{\finiteautomaton},\; \aodd'\in \automatonlang{\finiteautomaton'}, 
\lengthfunction{\aodd}=\lengthfunction{\aodd'}\}.$$
\item {\bf Pointwise tensor product.} There is a SOFA $\finiteautomaton\dot{\otimes}\finiteautomaton'$, 
$$\automatonlanghigher{\finiteautomaton\dot{\otimes}\finiteautomaton'}{2} =
\{\oddlang{\aodd}\otimes \oddlang{\aodd'}\;:\; \aodd \in \automatonlang{\finiteautomaton},\; 
\aodd' \in \automatonlang{\finiteautomaton'}, \lengthfunction{\aodd}=\lengthfunction{\aodd'}\}.$$
\item {\bf Pointwise map.} There is a SOFA $\dot{g}(\finiteautomaton)$ such that 
$$\automatonlanghigher{\dot{g}(\finiteautomaton)}{2} = \{g(\oddlang{\aodd})\;:\; \aodd\in \automatonlang{\finiteautomaton}\}.$$ 
\item {\bf Pointwise inverse map:} There is a SOFA $\dot{g}^{-1}(\finiteautomaton')$ such that 
$$\automatonlanghigher{\dot{g}^{-1}(\finiteautomaton')}{2} = \{g^{-1}(\oddlang{\aodd})\;:\; \aodd\in \automatonlang{\finiteautomaton'}\}.$$ 
\item {\bf Pointwise negation:} There is a SOFA $\dot{\neg} \finiteautomaton$ such that 
$$\automatonlanghigher{\dot{\neg}\finiteautomaton}{2} = \{ (\oddlength,\alphabet^{\oddlength})\backslash \oddlang{\aodd}\;:\;
\aodd\in \automatonlang{\finiteautomaton}, \lengthfunction{\aodd}=\oddlength\}.$$ 
\end{enumerate}
\end{corollary}
\begin{proof}
The proof follows directly from the fact that regular languages are closed under maps, together with
Lemma \ref{lemma:SimulationLemma}. 
The SOFAs $\dot{g}(\finiteautomaton)$, $\dot{g}^{-1}(\finiteautomaton)$, and $\dot{\neg}\finiteautomaton$
are obtained from $\finiteautomaton$ by replacing each transition $(\astate,\alayer,\astate')$ with 
the transitions $(\astate,g(\alayer),\astate')$, $(\astate,g^{-1}(\alayer),\astate')$, and 
$(\astate,\neg\alayer,\astate')$ respectively. For the binary operations, we first compute a
finite automaton $\finiteautomaton\otimes\finiteautomaton'$ over the alphabet 
$\layeralphabet{\alphabet_1}{\width_1}\times \layeralphabet{\alphabet_2}{\width_2}$ that accepts 
a string $\aodd\otimes \aodd'  = (\alayer_1,\alayer_1')(\alayer_2,\alayer_2')\dots (\alayer_{\oddlength},\alayer_{\oddlength}')$
if and only if $\aodd'=\alayer_1\alayer_2\dots\alayer_{\oddlength}$ is accepted by $\finiteautomaton$ and 
$\aodd' = \alayer_1'\alayer_2'\dots\alayer_{\oddlength}$ is accepted by $\finiteautomaton'$. Subsequently 
we define $\finiteautomaton\dot{\cup}\finiteautomaton'$, $\finiteautomaton\dot{\cap}\finiteautomaton'$ and 
$\finiteautomaton{\dot\otimes}\finiteautomaton'$ by replacing each transition
$(\astate,(\alayer,\alayer'),\astate')$ of $\finiteautomaton\otimes \finiteautomaton'$ with the 
transitions $(\astate,f_{\cup}(\alayer,\alayer'),\astate')$, $(\astate,f_{\cap}(\alayer,\alayer'),\astate')$, 
and $(\astate,f_{\otimes}(\alayer,\alayer'),\astate')$ respectively. 

We exemplify how Lemma \ref{lemma:SimulationLemma} can be used to complete the proof with the first item. The 
others follow an analogous argument. From the construction of $\finiteautomaton\dot{\cup}\finiteautomaton'$, 
we have that $\aodd\in \automatonlang{\finiteautomaton}$ and $\aodd'\in \automatonlang{\finiteautomaton'}$
are such that $\lengthfunction{\aodd}=\lengthfunction{\aodd'}$ if 
and only if $f_{\cup}(\aodd\otimes \aodd')$ belongs to $\automatonlang{\finiteautomaton\dot{\cup}\finiteautomaton'}$.
Since $\oddlang{f_{\cup}(\aodd\otimes \aodd')} = \oddlang{\aodd}\cup \oddlang{\aodd'}$, we have that 
$\automatonlanghigher{\finiteautomaton\dot{\cup}\finiteautomaton'}{2} = \{\oddlang{\aodd}\cup\oddlang{\aodd'}\;:\; \aodd\in 
\automatonlang{\finiteautomaton},\; \aodd'\in \automatonlang{\finiteautomaton'},\;\lengthfunction{\aodd}=\lengthfunction{\aodd'}\}$. 
\end{proof}

\section{Algorithmic Applications}\label{section:Applications}

In this section, we show that Theorems \ref{theorem:CanonizationTransductionTheorem} and Theorem \ref{theorem:CanonizationSecondOrder} can be used to 
provide novel algorithmic applications in the realm of the theory of ODDs of bounded width, and therefore also in the realm of the theory of ordered 
binary decision diagrams (OBDDs) of bounded width. In Subsection \ref{subsection:ApplicationMinimization} we will show that several 
minimization problems for deterministic and nondeterministic ODDs can be solved in fixed parameter tractable time when parameterized by width. Subsequently, in Subsection \ref{subsection:ApplicationCounting}
we will show that the problem of counting the number of distinct functions computable by some ODD of length $\oddlength$ and width $\width$ can be 
solved in time $h(|\alphabet|,\width)\cdot k^{O(1)}$ for a suitable $h:\N\times\N\rightarrow \N$. 

\subsection{Width and Size Minimization of Nondeterministic ODDs}
\label{subsection:ApplicationMinimization}

Models of computation comprised by ODDs of constant width have been studied in a variety of fields, such as 
symbolic computation, machine learning and property testing \cite{bollig2014width,Newman2002testing,goldreich2010testing,RonTsur2009testing}. 
In this section, we show that width minimization for ODDs is fixed-parameter tractable in the width parameter. Additionally,
the space of ODDs where the minimization will take place may be selected as the language $\automatonlang{\finiteautomaton}$ of a given 
second-order finite automaton $\finiteautomaton$. Furthermore, if such a minimum width ODD $\aodd'$ with $\oddlang{\aodd}=\oddlang{\aodd'}$ exists in $\automatonlang{\finiteautomaton}$, 
then one can furthermore impose that $\aodd'$ has minimum number of states or minimum number of transitions.

As important special cases, if we set $\finiteautomaton$ to be the finite automaton accepting the 
language $\odddefiningsetstar{\alphabet}{\width}$ the minimization occurs in the space of all (possibly nondeterministic) ODDs of width at most 
$\width$, while by setting $\finiteautomaton$ to be the finite automaton accepting the language 
$\cdodddefiningsetstar{\alphabet}{\width}$, the minimization takes place over the space of deterministic, 
complete ODDs of width at most $\width$.

\newcommand{\allodds}[2]{X(#1,#2)}

\begin{lemma}
\label{lemma:AllODDs}
Let $\alphabet$ be an alphabet, $\width\in \pN$, $\aodd$ be a $(\alphabet,\width)$-ODD, and $\finiteautomaton$ be a $(\alphabet,\width)$-SOFA.
One can construct in time $2^{O(|\alphabet|\cdot \width \cdot 2^{\width})}\cdot \nstates(\finiteautomaton)\cdot \oddlength$ 
a $(\alphabet,\width)$-SOFA $\allodds{\finiteautomaton}{\aodd}$ with 
$2^{O(|\alphabet|\cdot \width \cdot 2^{\width})}\cdot \nstates(\finiteautomaton)\cdot \oddlength$ states  such that
$\automatonlang{\allodds{\finiteautomaton}{\aodd}} =
\{\aodd'\in \automatonlang{\finiteautomaton}\;:\; \oddlang{\aodd} = \oddlang{\aodd'}\}$. 
\end{lemma}
\begin{proof}
Let $\finiteautomaton$ be a $(\alphabet,\width)$-SOFA, and $\aodd \in \odddefiningset{\alphabet}{\width}{\oddlength}$. 
Consider the $(\cdlayeralphabet{\alphabet}{2^{\width}},\layeralphabet{\alphabet}{\width})$-trans\-duc\-tion
$\atransduction_{\aodd} = \{(\canonizationfunction(\aodd),\aodd)\}$. Note that $\atransduction_{\aodd}$ is a singleton, and 
therefore, it is $(\oddlength+1)$-regular, since the language 
$\lang(\atransduction_{\aodd}) = \{\canonizationfunction(\aodd)\otimes \aodd\}$ is accepted by a finite automaton $\hat{\finiteautomaton}$
with $(\oddlength+1)$ states $\{\astate_0,\dots,\astate_{\oddlength}\}$. Here, $\astate_0$ is the unique initial state and 
$\astate_{\oddlength}$ is the unique final state. Indeed, let $\canonizationfunction(\aodd) = \alayer_1'\alayer_2'\dots\alayer_{\oddlength}'$. Note that 
this canonical form can be constructed in time $2^{O(\width)}\cdot |\alphabet|\cdot \oddlength$ by applying the standard minimization algorithm for
a single ODD (Theorem \ref{theorem:MinimizationODDs}). Then, for each $i\in \{0,\dots,\oddlength-1\}$, the automaton has a unique transition
leaving $\astate_i$, namely, the transition $(\astate_i, (\alayer_i',\alayer_i) ,\astate_{i+1})$.
It should be clear that $\canonizationfunction(\aodd)\otimes \aodd = (\alayer_1',\alayer_1)(\alayer_2',\alayer_2)\dots (\alayer_{\oddlength}',\alayer_{\oddlength}).$
is the only string accepted by $\hat{\finiteautomaton}$. 

Now consider the transduction $\canonizationtransduction{\alphabet}{\width}$. Since this transduction is 
$2^{O(|\alphabet|\cdot \width\cdot 2^{\width})}$-regular (Theorem \ref{theorem:CanonizationTransductionTheorem}),
it follows from Proposition \ref{proposition:PropertiesTransductions}.(\ref{item:automaton_composition}) that the transduction
$$\canonizationtransduction{\alphabet}{\width}\circ \atransduction_{\aodd} = \{(\aodd',\aodd)\;:\; 
\canonizationfunction(\aodd') = \canonizationfunction(\aodd) \} = \{(\aodd',\aodd)\;:\; \oddlang{\aodd'}=\oddlang{\aodd}\}$$
is $2^{O(|\alphabet|\cdot \width \cdot 2^{\width})}$-regular, and 
therefore, the language 
$\domain(\canonizationtransduction{\alphabet}{\width}\circ \atransduction_{\aodd}) = 
\{\aodd'\;:\; \oddlang{\aodd'} = \oddlang{\aodd}\}$ is 
$2^{O(|\alphabet|\cdot \width \cdot 2^{\width})}\cdot \oddlength$-regular. 
Additionally, an automaton $\finiteautomaton'$ accepting 
$\domain(\canonizationtransduction{\alphabet}{\width}\circ \atransduction_{\aodd})$ can be constructed 
in time $2^{O(|\alphabet|\cdot \width \cdot 2^{\width})}\cdot \oddlength$. 
This implies that one can construct in time
$2^{O(|\alphabet|\cdot \width \cdot 2^{\width})}\cdot \numberstates{\finiteautomaton}\cdot \oddlength$
a finite automaton $\allodds{\finiteautomaton}{\aodd}$ 
with $2^{O(|\alphabet|\cdot \width \cdot 2^{\width})}\cdot \numberstates{\finiteautomaton}\cdot \oddlength$ states
accepting the language 
$\domain(\canonizationtransduction{\alphabet}{\width}\circ \atransduction_{\aodd}) \cap \automatonlang{\finiteautomaton_{\aodd}}
= \{\aodd'\in \automatonlang{\finiteautomaton}\;:\; \oddlang{\aodd} = \oddlang{\aodd'}\}$.
\end{proof}

Let $\oplus:\dbset{a+1}\times \dbset{a+1} \rightarrow \dbset{a+1}$ be a binary operation for some $a\in \N$ and $\omega:\layeralphabet{\alphabet}{\width}\rightarrow \dbset{a+1}$ be a
weighting function. Then the weight of an ODD $\aodd\in \odddefiningset{\alphabet}{\width}{\oddlength}$ is defined as $\omega_{\oplus}(\aodd) = \bigoplus_{i=1}^{\oddlength}\omega(\alayer_i)$, where
the sum is performed from left to right.

\newcommand{\finiteautomatonJ}{\mathcal{J}}
\begin{proposition}
\label{proposition:Counting}
Let $\finiteautomatonJ$ be a $(\alphabet,\width)$-SOFA with non-empty language, $\omega:\layeralphabet{\alphabet}{\width}\rightarrow \dbset{a+1}$, and 
$\oplus:\dbset{a+1}\times \dbset{a+1}\rightarrow \dbset{a+1}$. Then one can construct in time $O(|\finiteautomatonJ|\cdot a\cdot \log a)$ 
an ODD $\aodd\in \automatonlang{\finiteautomatonJ}$ such that $\omega_{\oplus}(\aodd) = \min\{\omega_{\oplus}(\aodd')\;:\;\aodd'\in \automatonlang{\finiteautomatonJ}\}$. 
\end{proposition}
\begin{proof}
Let $\finiteautomatonJ'$ be the finite automaton with 
set of states $\automatonstates(\finiteautomatonJ') = \automatonstates(\finiteautomatonJ) \times \dbset{a+1}$, initial states 
$\automatoninitialstates(\finiteautomatonJ') = \automatoninitialstates(\finiteautomatonJ) \times \{0\}$, transition relation
$$\automatontransitions(\finiteautomatonJ') = \{[(\astate,u),\alayer,(\astate',u\oplus\omega(\alayer))]\;:\; 
(\astate,u)\in \automatonstates(\finiteautomatonJ'),\; (\astate,\alayer,\astate') \in \automatontransitions(\finiteautomatonJ)\},$$
and set of final states $\automatonfinalstates(\finiteautomatonJ') = \automatonfinalstates(\finiteautomatonJ) \times \{\alpha\}$,
where
$$\alpha = \min\{u\;:\; \exists \astate\in \automatonfinalstates(\finiteautomatonJ), (\astate,u) \mbox{ is reachable from an initial state of $\finiteautomatonJ'$}\}.$$
Then, we have that
$\sequence{[(0,\astate_1,u_1),\alayer_1,(1,\astate_2,u_2)],\ldots,
[(\oddlength-1,\astate_{\oddlength-1},u_{\oddlength-1}),\alayer_{\oddlength},(\oddlength,\astate_{\oddlength},u_{\oddlength})]}$ is an accepting 
sequence of transitions in $\finiteautomatonJ'$ if and only if $\aodd = \alayer_1\dots\alayer_{\oddlength}$ is an ODD in $\automatonlang{\finiteautomatonJ}$ 
of weight $\omega_{\oplus}(\aodd) = u_{\oddlength} = \alpha$, where $\alpha$ is the minimum weight of an ODD in $\automatonlang{\finiteautomatonJ'}$.

Clearly, one can construct the automaton $\finiteautomatonJ'$ in time $O(|\finiteautomatonJ|\cdot a\log a)$. Therefore, one 
can also obtain an ODD $\aodd'\in \automatonlang{\finiteautomatonJ'}$ in the same amount of time. 
\end{proof}

By combining Lemma \ref{lemma:AllODDs} with Proposition \ref{proposition:Counting} we obtain the following theorem. 

\begin{theorem}
\label{theorem:MinimizationODDAutomaton}
Let $\aodd$ be an ODD in $\odddefiningset{\alphabet}{\width}{\oddlength}$ and let $\finiteautomaton$ be a $(\alphabet,\width)$-SOFA.
One can determine in time $2^{O(|\alphabet|\cdot \width \cdot 2^{\width})}\cdot \nstates(\finiteautomaton)\cdot \oddlength$ whether
there is an ODD $\aodd'\in \automatonlang{\finiteautomaton}$ such that $\oddlang{\aodd'} = \oddlang{\aodd}$. Suppose such an ODD exists.
\begin{enumerate}
\item One can construct in time $2^{O(|\alphabet|\cdot \width \cdot 2^{\width})}\cdot \nstates(\finiteautomaton)\cdot \oddlength$
an ODD $\aodd'\in \automatonlang{\finiteautomaton}$ of minimum width such that $\oddlang{\aodd} = \oddlang{\aodd'}$. 
\item One can construct in time $2^{O(|\alphabet|\cdot \width \cdot 2^{\width})}\cdot \nstates(\finiteautomaton)\cdot \oddlength^2\cdot \log \oddlength$
an ODD $\aodd'\in \automatonlang{\finiteautomaton}$ with minimum number of states such that $\oddlang{\aodd} = \oddlang{\aodd'}$. 
\item One can construct in time $2^{O(|\alphabet|\cdot \width \cdot 2^{\width})}\cdot \nstates(\finiteautomaton)\cdot \oddlength^2\cdot \log \oddlength$
an ODD $\aodd'\in \automatonlang{\finiteautomaton}$ with minimum number of transitions such that $\oddlang{\aodd} = \oddlang{\aodd'}$. 
\end{enumerate}
\end{theorem}
\begin{proof}
$ $
\begin{enumerate}
\item In Proposition \ref{proposition:Counting}, set $\finiteautomatonJ = \allodds{\finiteautomaton}{\aodd}$,
$a=\width$, $\oplus(x,y) = \max\{x,y\}$, and for each $\alayer\in \layeralphabet{\alphabet}{\width}$ set $\omega(\alayer)=\fwidth(\alayer)$.
Then for each $\aodd'\in \odddefiningset{\alphabet}{\width}{\oddlength}$, $\omega_{\oplus}(\aodd) = \fwidth(\aodd)$. Therefore, 
by Proposition \ref{proposition:Counting}, one can construct in time $2^{O(|\alphabet|\cdot \width \cdot 2^{\width})}\cdot \nstates(\finiteautomaton)\cdot \oddlength$ 
an ODD $\aodd'\in \automatonlang{\finiteautomaton}$ of minimum width such that $\oddlang{\aodd'} = \oddlang{\aodd}$. 
\item In Proposition \ref{proposition:Counting}, set $\finiteautomatonJ = \allodds{\finiteautomaton}{\aodd}$,
$a=\width\cdot (\oddlength+1)$, $\oplus(x,y) = x+y$, and for each $\alayer\in \layeralphabet{\alphabet}{\width}$ set $\omega(\alayer)=|\leftfrontier(\layer)| + \finalflag(\layer)\cdot |\rightfrontier(\layer)|$.
Then for each $\aodd'\in \odddefiningset{\alphabet}{\width}{\oddlength}$, $\omega_{\oplus}(\aodd) = \numberstates{\aodd}$. Therefore, 
by Proposition \ref{proposition:Counting}, one can construct in time $2^{O(|\alphabet|\cdot \width \cdot 2^{\width})}\cdot \nstates(\finiteautomaton)\cdot \oddlength^2\cdot \log \oddlength$ 
an ODD $\aodd'\in \automatonlang{\finiteautomaton}$ with minimum number of states such that $\oddlang{\aodd'} = \oddlang{\aodd}$.
\item In Proposition \ref{proposition:Counting}, set $\finiteautomatonJ = \allodds{\finiteautomaton}{\aodd}$,
$a=|\alphabet|\cdot \width^2\cdot \oddlength$, $\oplus(x,y) = x+y$, and for each $\alayer\in \layeralphabet{\alphabet}{\width}$ set $\omega(\alayer)=|\layertransitions(\layer)|$.
Then for each $\aodd'\in \odddefiningset{\alphabet}{\width}{\oddlength}$, $\omega_{\oplus}(\aodd) = \numbertransitions{\aodd}$. Therefore, 
by Proposition \ref{proposition:Counting}, one can construct in time $2^{O(|\alphabet|\cdot \width \cdot 2^{\width})}\cdot \nstates(\finiteautomaton)\cdot \oddlength^2\cdot \log \oddlength$ 
an ODD $\aodd'\in \automatonlang{\finiteautomaton}$ with minimum number of transitions such that $\oddlang{\aodd'} = \oddlang{\aodd}$.
\end{enumerate}
\end{proof}

Let $\subsetlayers\subseteq \layeralphabet{\alphabet}{\width}$. Then by plugging $\finiteautomaton_{\subsetlayers}$ in Theorem \ref{theorem:MinimizationODDAutomaton},
the following theorem, which can be used to address several minimization problems for ODDs over the space of ODDs in $\subsetlayers^{\circledast}$.

\begin{theorem}
\label{theorem:Minimization}
Let $\aodd$ be an ODD in $\odddefiningset{\alphabet}{\width}{\oddlength}$ and let $\subsetlayers\subseteq \layeralphabet{\alphabet}{\width}$. One 
can determine in time $2^{O(|\alphabet|\cdot \width \cdot 2^{\width})}\cdot \oddlength$ whether there is an ODD 
$\aodd'\in \subsetlayers^{\circ \oddlength}$ such that $\oddlang{\aodd'} = \oddlang{\aodd}$. Suppose such an ODD exists.
\begin{enumerate}
\item One can construct in time $2^{O(|\alphabet|\cdot \width \cdot 2^{\width})}\cdot \oddlength$
an ODD $\aodd'\in \subsetlayers^{\circledast}$ of minimum width such that $\oddlang{\aodd} = \oddlang{\aodd'}$. 
\item One can construct in time $2^{O(|\alphabet|\cdot \width \cdot 2^{\width})}\cdot \oddlength^2\log \oddlength$
an ODD $\aodd'\in \subsetlayers^{\circledast}$ with minimum number of states such that $\oddlang{\aodd} = \oddlang{\aodd'}$. 
\item One can construct in time $2^{O(|\alphabet|\cdot \width \cdot 2^{\width})}\cdot \oddlength^2\log \oddlength$ 
an ODD $\aodd'\in \subsetlayers^{\circledast}$ with minimum number of transitions such that $\oddlang{\aodd} = \oddlang{\aodd'}$. 
\end{enumerate}
\end{theorem}

\subsection{Counting Functions Computable by ODDs of a Given Width.}
\label{subsection:ApplicationCounting} 

Let $\alphabet$ be an alphabet and $\width,\oddlength\in \pN$. 
Each ODD $\aodd\in \odddefiningset{\alphabet}{\width}{\oddlength}$ can be regarded 
as a representation of a function $\functionfromODD_{\aodd}:\alphabet^{\oddlength}\rightarrow \set{0,1}$. 
More precisely, for each $\astring\in \alphabet^{\oddlength}$, $\functionfromODD_{\aodd}(\astring) = 1$ if 
and only if $\astring\in \oddlang{\aodd}$. We say that $\functionfromODD_{\aodd}$ is the function
computed by $\aodd$.  

In this subsection, we analyze the problem of counting the number of functions of type 
$\alphabet^{\oddlength}\rightarrow \{0,1\}$ that can be computed by some ODD of width $\width$ over the alphabet $\alphabet$.
We note that to solve this problem it is not enough to count the number of ODDs in $\odddefiningset{\alphabet}{\width}{\oddlength}$. 
The caveat is that several ODDs in $\odddefiningset{\alphabet}{\width}{\oddlength}$ may represent the same function. 
Fortunately, we can solve the issue of multiple representatives for a given function by resorting to our
canonical form of canonical forms theorem (Theorem \ref{theorem:CanonizationSecondOrder}).

\newcommand{\finiteautomatonA}{\mathcal{A}}
\newcommand{\Gammaalphabet}{\Gamma}

It is well known that the problem of counting the number of strings of length $\oddlength$ accepted by a given deterministic 
finite automaton $\finiteautomatonA$ can be solved in time polynomial in $\oddlength$ and in the number of states of $\finiteautomatonA$.
Below we state a more precise upper bound. 

\begin{proposition}
\label{proposition:CountingAcceptedWords}
Let $\finiteautomatonA$ be a deterministic finite automaton over an alphabet $\Gammaalphabet$. 
Then, for each $\oddlength\in \N$, one can count in time 
$O(\numberstates{\finiteautomatonA}\cdot \oddlength^2\cdot |\Gammaalphabet|\cdot\log|\Gammaalphabet|)$
the number of words of length $\oddlength$ accepted by $\finiteautomatonA$. 
\end{proposition}
\begin{proof}
Let $\finiteautomatonA = (\Gammaalphabet,\automatonstates,\automatoninitialstates,\automatonfinalstates,\automatontransitions)$. Since 
$\finiteautomatonA$ is deterministic, $\automatoninitialstates = \{\astate_0\}$ for some state $\astate_0$.  Additionally, there is a 
bijection from the set words of length $\oddlength$ accepted by $\finiteautomatonA$ to the set accepting sequences 
of transitions connecting the initial state $\astate_0$ to some final state in $\automatonfinalstates$. 

We start by constructing a matrix $M:\dbset{\oddlength}\times \automatonstates\rightarrow \N$  such that for each $i\in \oddlength$ and each $\astate \in \automatonstates$, 
the entry $M(i,\astate)$ is equal to the number of valid sequences of transitions of length $\oddlength-i$ from $\astate$ to some final state in $\automatonfinalstates$. In particular, 
$M(0,\astate_0)$ is the number of valid sequences of transitions of length $\oddlength$ from $\astate_0$ to some final state in $\automatonfinalstates$. The matrix $M$ 
is constructed by induction on $\oddlength-i$. In the base case, $i=\oddlength$. In this case, we set $M(\oddlength,\astate) = 1$ if $\astate\in \automatonfinalstates$,
and set $M(\oddlength,\astate)=0$ otherwise. Now, let $i\in \dbset{\oddlength-1}$ and assume that the value $M(i+1,\astate)$ has been determined 
for every $\astate\in \automatonstates$. Then, for each $\astate\in \automatonstates$, we let $M(i,\astate) = \sum_{(\astate,\asymbol,\astate')\in \automatontransitions} M(i+1,\astate')$. 
In other words, $M(i,\astate)$ is defined as the sum of all $M(i,\astate')$ for which $(\astate,\asymbol,\astate')$ is a transition in $\finiteautomatonA$ for some $\asymbol\in \Gammaalphabet$. 

Since, there are at most $|\Gammaalphabet|^{\oddlength}$ words of length $\oddlength$, we have that each entry 
of $M$ can be represented using $\oddlength\cdot \log |\Gammaalphabet|$ bits. Additionally, the computation of each entry involves the summation of 
$|\Gammaalphabet|$ entries, which in overall can be performed in time $O(\oddlength\cdot |\Gammaalphabet|\cdot \log \Gammaalphabet)$. Since the matrix has 
$(\oddlength+1) \cdot \numberstates{\finiteautomatonA}$ entries, the whole matrix can be constructed in time $O(\numberstates{\finiteautomatonA}\cdot \oddlength^2\cdot |\Gammaalphabet|\cdot \log |\Gammaalphabet|)$. 
\end{proof}

\begin{theorem}
\label{theorem:CountingFunctionsAutomaton}
Let $\finiteautomaton$ be a $(\alphabet,\width)$-second order finite automaton. For each $\oddlength\in \N$, 
one can count in time $2^{\nstates(\finiteautomaton)\cdot 2^{\BigOh(\abs{\alphabet}\cdot\width\cdot 2^{\width})}}\cdot \oddlength^{2}$
the number of functions $f:\alphabet^{\oddlength}\rightarrow \{0,1\}$ computable by some ODD of length $\oddlength$ in $\automatonlang{\finiteautomaton}$. 
\end{theorem}
\begin{proof}
By Theorem \ref{theorem:CanonizationSecondOrder}, one can construct in time $2^{\nstates(\finiteautomaton)\cdot 2^{\BigOh(\abs{\alphabet}\cdot\width\cdot 2^{\width})}}$ 
a deterministic second-order finite automaton $\canonizationfunction_2(\finiteautomaton)$ (with at most $2^{\nstates(\finiteautomaton)\cdot 2^{\BigOh(\abs{\alphabet}\cdot\width\cdot 2^{\width})}}$ states)
such that $\automatonlang{\canonizationfunction_2(\finiteautomaton)} = \{\canonizationfunction(\aodd)\;:\; \aodd\in \automatonlang{\finiteautomaton}\}$.
This implies that for each language $\lang \in \oddlanghigher{\finiteautomaton}{2}$, there is a unique ODD $\aodd\in \automatonlang{\canonizationfunction_2(\finiteautomaton)}$ such that 
$\oddlang{\aodd}=\lang$. Therefore, counting the number of functions of type $\alphabet^{\width}\rightarrow \{0,1\}$ computable by some ODD in $\automatonlang{\finiteautomaton}$ 
amounts to counting the number of ODDs of length $\oddlength$ accepted by $\canonizationfunction_2(\finiteautomaton)$.
By setting $\finiteautomatonA = \canonizationfunction_2(\finiteautomaton)$ and $\Gammaalphabet = \layeralphabet{\alphabet}{2^{\width}}$ in Proposition \ref{proposition:CountingAcceptedWords}, 
and by using the facts that $|\finiteautomatonA| = 2^{\nstates(\finiteautomaton)\cdot 2^{\BigOh(\abs{\alphabet}\cdot\width\cdot 2^{\width})}}$ and $|\Gammaalphabet| = 2^{O(|\alphabet|\cdot 2^{\width}\cdot \width)}$, 
we have that this counting problem can be solved in time 
$2^{\nstates(\finiteautomaton)\cdot 2^{\BigOh(\abs{\alphabet}\cdot\width\cdot 2^{\width})}} \cdot \oddlength^2$.
\end{proof}

If all ODDs in the language of $\finiteautomaton$ are deterministic and complete then one can adapt the proof of 
Theorem \ref{theorem:CountingFunctionsAutomaton} by using Observation \ref{observation:BetterConstruction} and by
setting $\Gamma = \cdlayeralphabet{\alphabet}{\width}$ in order to obtain a more efficient counting algorithm. 

\begin{observation}
\label{observation:CountingFunctionsAutomatonBetter}
Let $\finiteautomaton$ be a $(\alphabet,\width)$-second order finite automaton such that $\automatonlang{\finiteautomaton}\subseteq \cdlayeralphabet{\alphabet}{\width}$.
For each $\oddlength\in \N$, one can count in time $2^{\nstates(\finiteautomaton)\cdot 2^{\BigOh(\abs{\alphabet}\cdot\width\cdot \log\width)}}\cdot \oddlength^{2}$
the number of functions $f:\alphabet^{\oddlength}\rightarrow \{0,1\}$ computable by some ODD of length $\oddlength$ in $\automatonlang{\finiteautomaton}$. 
\end{observation}

By combining Lemma \ref{lemma:LayerSubsetRegular} with Theorem \ref{theorem:CountingFunctionsAutomaton} and Observation \ref{observation:CountingFunctionsAutomatonBetter}, 
we obtain the following corollary. 

\begin{corollary}
\label{corollary:CountingFunctions}
Let $\alphabet$ be an alphabet, $\width \in \pN$, $\subsetlayers\subseteq \layeralphabet{\alphabet}{\width}$, and $\widehat{\subsetlayers}\subseteq \cdlayeralphabet{\alphabet}{\width}$. 
\begin{enumerate}
\item One can count in time $2^{2^{\BigOh(\abs{\alphabet}\cdot\width\cdot 2^{\width})}}\cdot \oddlength^{2}$ the number of 
functions $f:\alphabet^{\oddlength}\rightarrow \{0,1\}$ computable by some ODD in $\subsetlayers^{\circledast}$. 
\item One can count in time $2^{2^{\BigOh(\abs{\alphabet}\cdot\width\cdot \log\width)}}\cdot \oddlength^{2}$ 
the number of functions $f:\alphabet^{\oddlength}\rightarrow \{0,1\}$ computable by some ODD in $\widehat{\subsetlayers}^{\circledast}$. 
\end{enumerate}
\end{corollary}
\begin{proof}
By Lemma \ref{lemma:LayerSubsetRegular}, one can construct SOFAs  $\finiteautomaton_{\subsetlayers}$ and $\finiteautomaton_{\widehat{\subsetlayers}}$ 
with $(|\subsetlayers|+1)$ and $(|\widehat{\subsetlayers}|+1)$ states respectively such 
that $\automatonlang{\finiteautomaton_{\subsetlayers}} = \subsetlayers^{\circledast}$, and $\automatonlang{\finiteautomaton_{\widehat{\subsetlayers}}} = \widehat{\subsetlayers}^{\circledast}$. 
Since $|\subsetlayers| = 2^{O(|\alphabet|\width^2)}$, it follows from Theorem \ref{theorem:CountingFunctionsAutomaton} that one can count the number of
functions $f:\alphabet^{\oddlength}\rightarrow \{0,1\}$ computable by ODDs in $\subsetlayers^{\circledast}$ in time $2^{2^{\BigOh(\abs{\alphabet}\cdot\width\cdot 2^{\width})}}\cdot \oddlength^{2}$. Analogously, 
since $|\widehat{\subsetlayers}| = 2^{O(|\alphabet|\width \log \width)}$, it follows from Observation \ref{observation:CountingFunctionsAutomatonBetter}
that one can count the number of functions $f:\alphabet^{\oddlength}\rightarrow \{0,1\}$ computable by ODDs in $\widehat{\subsetlayers}^{\circledast}$ 
in time $2^{2^{O(|\alphabet|\cdot \width \cdot \log \width)}}$. 
\end{proof}

\section{Proof of the Canonization as Transduction Theorem}
\label{section:CanonizationTransductionTheorem}

In this section, we prove Theorem \ref{theorem:CanonizationTransductionTheorem}, which states that for each alphabet $\alphabet$, and 
each $\width\in \pN$ the following holds. 

\begin{enumerate}
\item The functional transduction $\cdcanonizationtransduction{\alphabet}{\width} = \set{\tuple{\aodd,\canonizationfunction(\aodd)} \setst \aodd\in \cdodddefiningsetstar{\alphabet}{\width}}$
	is $2^{O(|\alphabet|\cdot \width\cdot \log \width)}$-regular.
\item The functional transduction $\canonizationtransduction{\alphabet}{\width} = \set{\tuple{\aodd,\canonizationfunction(\aodd)} \setst \aodd\in \odddefiningsetstar{\alphabet}{\width}}$
	is $2^{O(|\alphabet|\cdot \width \cdot 2^{\width})}$-regular. 
\end{enumerate}

Although the complete proof of Theorem \ref{theorem:CanonizationTransductionTheorem} is quite technical, it is possible to give an intuitive overview of the main steps in the proof. 
More specifically, we will show that the transduction $\cdcanonizationtransduction{\alphabet}{\width}$ can be cast a composition  

\begin{equation}
\label{equation:TransductionSequence}
\cdcanonizationtransduction{\alphabet}{\width} \defeq \duplicate{\cdodddefiningsetstar{\alphabet}{\width}} \circ \reachabilitytransduction{\alphabet}{2^\width}\circ\mergingtransduction{\alphabet}{2^{\width}}\circ\normalizationtransduction{\alphabet}{2^{\width}}\text{,}
\end{equation}

of regular transductions satisfying the following properties. 

\begin{enumerate}
\item $\duplicate{\cdodddefiningsetstar{\alphabet}{\width}}$ is a functional $2^{O(|\alphabet|\cdot \width\log\width)}$-regular
$(\cdlayeralphabet{\alphabet}{\width},\cdlayeralphabet{\alphabet}{\width})$-transduction that sends each ODD 
$\aodd\in \cdodddefiningsetstar{\alphabet}{\width}$ to itself. This 
transduction is used to limit the domain of $\cdcanonizationtransduction{\alphabet}{\width}$ to deterministic, complete $(\alphabet,\width)$-ODDs.  
\item $\reachabilitytransduction{\alphabet}{\width}$ is a functional $2^{O(|\alphabet|\cdot \width \cdot \log \width)}$-regular
$(\cdlayeralphabet{\alphabet}{\width},\cdlayeralphabet{\alphabet}{\width})$-transduction that sends each
ODD $\aodd\in\cdodddefiningsetstar{\alphabet}{\width}$ to a reachable ODD $\bodd\in\cdodddefiningsetstar{\alphabet}{\width}$ 
with $\oddlang{\aodd}=\oddlang{\bodd}$. This transduction simulates the process of eliminating unreachable states from $\aodd$.
\item $\mergingtransduction{\alphabet}{\width}$ is a functional $2^{O(|\alphabet|\cdot \width \cdot \log \width)}$-regular
$(\cdlayeralphabet{\alphabet}{\width},\cdlayeralphabet{\alphabet}{\width})$-transduction 
that sends each reachable, deterministic, complete ODD $\aodd\in\cdodddefiningsetstar{\alphabet}{\width}$ to a minimized,
deterministic, complete ODD $\bodd\in\odddefiningsetstar{\alphabet}{\width}$ with $\oddlang{\aodd}=\oddlang{\bodd}$. This transduction 
simulates the process of merging equivalent states in a ODD. 
\item $\normalizationtransduction{\alphabet}{\width}$ is a functional $2^{O(|\alphabet|\cdot \width \cdot \log \width)}$-regular
$(\cdlayeralphabet{\alphabet}{\width},\cdlayeralphabet{\alphabet}{\width})$-transduction
that sends each deterministic, complete ODD $\aodd\in\odddefiningsetstar{\alphabet}{\width}$ to its normalized 
version $\bodd\in\odddefiningsetstar{\alphabet}{\width}$. This transduction simulates the process of numbering the states of an ODD according 
to their lexicographical order. This guarantees that the ODD is unique not only up to isomorphism, but also syntactically unique.  
\end{enumerate}

Intuitively, the regular transductions above simulate the steps used in the standard ODD minimization algorithm. 
By using Proposition \ref{proposition:PropertiesTransductions}.(2), we have that the transduction $\canonizationtransduction{\alphabet}{\width}$
is $2^{O(|\alphabet|\cdot \width \cdot 2^{\width})}$-regular. The fact that each of the five transductions above is functional implies that 
$\canonizationtransduction{\alphabet}{\width}$ is also functional. Additionally, it is straightforward to note that 
$\domain(\cdcanonizationtransduction{\alphabet}{\width}) = \cdodddefiningsetstar{\alphabet}{\width}$. Finally, a pair of ODDs $(\aodd,\aodd')$ belongs to 
$\cdcanonizationtransduction{\alphabet}{\width}$ if and only if $\aodd'$ is deterministic, complete, minimized, normalized and $\oddlang{\aodd} = \oddlang{\aodd'}$. 
In other words, if and only if $\aodd'$ is the canonical form $\canonizationfunction(\aodd)$ of Theorem \ref{theorem:MinimizationODDs}. 

Now, the transduction $\canonizationtransduction{\alphabet}{\width}$ can be obtained as the composition 

\begin{equation}
\label{equation:TransductionSequenceDeterministic}
\canonizationtransduction{\alphabet}{\width} \defeq \duplicate{\odddefiningsetstar{\alphabet}{\width}} \circ \determinizationtransduction{\alphabet}{\width}\circ
\cdcanonizationtransduction{\alphabet}{2^{\width}}. 
\end{equation}

Here, $\determinizationtransduction{\alphabet}{\width}$ is a functional  $2$-regular 
$(\layeralphabet{\alphabet}{\width},\cdlayeralphabet{\alphabet}{2^{\width}})$-transduction that sends each 
ODD $\aodd\in\odddefiningsetstar{\alphabet}{\width}$ to a deterministic, complete ODD $\bodd\in\odddefiningsetstar{\alphabet}{2^{\width}}$ with $\oddlang{\aodd}=\oddlang{\bodd}$. 
This transduction simulates the application of the standard power set construction to the states of a ODD, and blows the width of the original ODD 
at most exponentially. Since $\cdcanonizationtransduction{\alphabet}{\width}$ is $2^{O(|\alphabet|\cdot \width \cdot \log \width)}$-regular, we have that
$\cdcanonizationtransduction{\alphabet}{2^{\width}}$ is $2^{O(|\alphabet|\cdot \width \cdot 2^{\width})}$-regular. This implies that $\canonizationtransduction{\alphabet}{\width}$
is also $2^{O(|\alphabet|\cdot \width \cdot 2^{\width})}$-regular. 

Next, in Subsection \ref{subsection:BasicTransductions}, we will define two elementary types of regular transductions: the {\em multimap transductions} and 
the {\em compatibility transductions}. Subsequently we will define $\determinizationtransduction{\alphabet}{\width}$, $\reachabilitytransduction{\alphabet}{\width}$,
$\mergingtransduction{\alphabet}{\width}$ and $\normalizationtransduction{\alphabet}{\width}$ using these elementary transductions. 
The determinization transduction $\determinizationtransduction{\alphabet}{\width}$ will be defined in Subsection \ref{subsection:DeterminizationTransduction} and its properties
analyzed in Lemma \ref{lemma:DeterminizationTransduction}. The reachability transduction $\reachabilitytransduction{\alphabet}{\width}$ will be defined
in Subsection \ref{subsection:ReachabilityTransduction}, and its properties analyzed in Lemma \ref{lemma:ReachabilityTransduction}. The merging transduction 
$\mergingtransduction{\alphabet}{\width}$ will be defined in Subsection \ref{subsection:MergingTransduction}, and its properties analyzed in Lemma \ref{lemma:MergingTransduction}. 
The normalization transduction will be defined Subsection \ref{subsection:NormalizationTransduction} and its properties analyzed in Lemma \ref{lemma:NormalizationTransduction}. 
Finally, in Subsection \ref{subsection:FinalProof} we will combine Observation \ref{observation:CopyRegular} with these four lemmas to conclude the proof of Theorem \ref{theorem:CanonizationTransductionTheorem}. 

\subsection{Basic Transductions}
\label{subsection:BasicTransductions}

Let $\alphabet$ be an alphabet and $\arelation\subseteq \alphabet\times \alphabet$ be a binary 
relation over $\alphabet$. 
For each $\stringlength\in\pN$ and each string $\astring=\asymbol_{1}\cdots\asymbol_{\stringlength}\in\alphabet^{\stringlength}$, we say that $\astring$ is {\em $\arelation$-compatible} if $(\asymbol_{i},\asymbol_{i+1})\in \arelation$ for each $i \in \bbset{\stringlength-1}$. 
We let $$\compTransduction{\arelation} \defeq \set{\tuple{\astring,\astring}\in \alphabet^{+}\cartesianproduct \alphabet^{+} \setst \astring \text{ is $\arelation$-compatible}}$$
be the \emph{$\arelation$-compatibility transduction}, \ie the $(\alphabet,\alphabet)$-transduction that sends each $\arelation$-compatible string $\astring\in \alphabet^+$ to itself. 

Let $\aalphabet$ and $\balphabet$ be two alphabets and $\arelation\subseteq \aalphabet\cartesianproduct\balphabet$ be a relation.
We let 
$$\begin{multlined}[t][0.90\textwidth]
    \multimaptransduction{\arelation} \defeq \set{\tuple{\astring,\bstring} \setst \astring=\asymbol_{1}\cdots\asymbol_{\stringlength}\in\aalphabet^{\stringlength},\, \bstring=\bsymbol_{1}\cdots\bsymbol_{\stringlength}\in\balphabet^{\stringlength},  \tuple{\asymbol_{i},\bsymbol_{i}}\in\arelation \text{ for each } i\in \bbset{\stringlength},\, \stringlength\in\pN}
\end{multlined}$$
be the \emph{$\arelation$-multimap transduction}. 
If $\amap\colon\aalphabet\rightarrow \balphabet$ is a map, then we write $\multimaptransduction{\amap}$ to denote the transduction $\multimaptransduction{\arelation_{\amap}}$, where $\arelation_{\amap} \defeq \set{\tuple{\asymbol,\amap(\asymbol)} \setst \asymbol\in \aalphabet}$.

\begin{proposition}\label{proposition:SizeTransductions}
Let $\alphabet$, $\aalphabet$ and $\balphabet$ be three alphabets, and let $\arelation\subseteq \alphabet\times \alphabet$ and $\brelation\subseteq \aalphabet\times \balphabet$ be binary relations. The following statements hold. 
\begin{enumerate}
	\item \label{SizeTransductionsOne} The transduction $\compTransduction{\arelation}$ is $(\abs{\alphabet}+2)$-regular. 
	\item \label{SizeTransductionsTwo} The transduction $\multimaptransduction{\brelation}$ is $2$-regular. 
\end{enumerate}
\end{proposition}
\begin{proof}
\begin{enumerate}
\item We let $\finiteautomaton_{\multimaptransduction{\brelation}}$ be the finite automaton with state set $\automatonstates(\finiteautomaton_{\multimaptransduction{\brelation}}) = \set{\astate,\bstate}$, initial state set $\automatoninitialstates(\finiteautomaton_{\multimaptransduction{\brelation}}) = \set{\astate}$, final state set $\automatonfinalstates_{\multimaptransduction{\brelation}}=\set{\bstate}$ and transition set $\automatontransitions_{\multimaptransduction{\brelation}} = \set{\tuple{\astate,\tuple{\asymbol,\bsymbol},\bstate}\setst\tuple{\asymbol,\bsymbol}\in\brelation}\cup\set{\tuple{\bstate,\tuple{\asymbol,\bsymbol},\bstate}\setst \tuple{\asymbol,\bsymbol}\in \brelation}$. 
Clearly, $\finiteautomaton$ has exactly two states, namely $\astate$ and $\bstate$. 
Moreover, for each two strings $\astring\in\aalphabet^{+}$ and $\bstring\in\balphabet^{+}$, $\finiteautomaton_{\multimaptransduction{\brelation}}$ accepts the string $\astring\tensorproduct\bstring\in(\aalphabet\cartesianproduct\balphabet)^{+}$ if and only if $\abs{\astring}=\abs{\bstring}$ and $\tuple{\asymbol_{i},\bsymbol_{i}}\in \brelation$ for each $i\in \bbset{\stringlength}$, where $\astring=\asymbol_{1}\cdots\asymbol_{\stringlength}$, $\bstring=\bsymbol_{1}\cdots\bsymbol_{\stringlength}$ and $\stringlength=\abs{\astring}$. \qedhere

\item We let $\finiteautomaton_{\compTransduction{\arelation}}$ be the finite automaton over the alphabet $\alphabet\cartesianproduct\alphabet$, with state set $\automatonstates(\finiteautomaton_{\compTransduction{\arelation}})=\set{\astate,\bstate} \cup \set{\astate_{\asymbol}\setst \asymbol\in \alphabet}$, initial state set $\automatoninitialstates(\finiteautomaton_{\compTransduction{\arelation}}) = \set{\astate}$, final state set $\automatonfinalstates(\finiteautomaton_{\compTransduction{\arelation}}) = \set{\bstate}$ and transition set 
$\automatontransitions(\finiteautomaton_{\compTransduction{\arelation}}) = \set{\tuple{\astate,\tuple{\asymbol,\asymbol},\astate_{\asymbol}} \setst \asymbol \in \alphabet}\cup 
\set{\tuple{\astate_{\asymbol},\tuple{\bsymbol,\bsymbol},\astate_{\bsymbol}} \setst \tuple{\asymbol,\bsymbol} \in \arelation}\cup\set{\tuple{\astate_{\asymbol},\tuple{\bsymbol,\bsymbol},\bstate}\setst\tuple{\asymbol,\bsymbol}\in\arelation}\text{.}$
Clearly, $\finiteautomaton_{\compTransduction{\arelation}}$ has at most $\abs{\alphabet}+2$ states. 
Moreover, it is not hard to check that, for each $\stringlength\in\pN$, $\finiteautomaton_{\compTransduction{\arelation}}$ accepts a string $\astring=\asymbol_{1}\cdots\asymbol_{\stringlength} \in \alphabet^{\stringlength}$ if and only if $\tuple{\asymbol_{i},\asymbol_{i+1}} \in \arelation$ for each $i\in \bset{\stringlength-1}$.  
Therefore, the language of $\finiteautomaton_{\compTransduction{\arelation}}$ is $\automatonlang{\finiteautomaton_{\compTransduction{\arelation}}} = \transductionlang(\compTransduction{\arelation})$.  \vspace{1.0ex}
\end{enumerate}
\end{proof}

The next observation is a direct consequence of Proposition \ref{proposition:PropertiesTransductions}.(3) and Corollary \ref{corollary:AllODDs}. 

\begin{observation}
\label{observation:CopyRegular}
Let $\alphabet$ be an alphabet and $\width\in \pN$. 
\begin{enumerate}
\item $\duplicate{\odddefiningsetstar{\alphabet}{\width}}$ is $2^{O(|\alphabet|\cdot \width \cdot 2^{\width})}$-regular. 
\item $\duplicate{\cdodddefiningsetstar{\alphabet}{\width}}$ is $2^{O(|\alphabet|\cdot \width \cdot \log \width)}$-regular.  
\end{enumerate}
\end{observation}

\subsection{Determinization Transduction}
\label{subsection:DeterminizationTransduction}
In this subsection, we define the determinization transduction $\determinizationtransduction{\alphabet}{\width}$, 
which intuitively simulates the application of the well known power-set construction to the layers of a $(\alphabet,\width)$-ODD.

For each $\width\in \pN$, we let $\pwbijection \colon \pwset{\dbset{\width})} \rightarrow \dbset{2^{\width}}$ be the bijection that sends each subset $\sset \subseteq \dbset{\width}$ to the natural number $\pwbijection(\sset) \defeq \sum_{i \in \sset} 2^{i}$.
In particular, we remark that $\pwbijection(\emptyset) = 0$ and $\pwbijection(\set{i})=2^{i}$ for each $i \in \sset$. 

Let $\alphabet$ be an alphabet, $\width\in\pN$, $\alayer\in\layeralphabet{\alphabet}{\width}$, $\sset \subseteq \layerleftfrontier(\alayer)$ and $\alphabet' \subseteq \alphabet$. 
We let $\layerreachablestates{\alayer}{\sset}{\alphabet'}$ be the set of all right states of $\alayer$ that are reachable from some left state in $\sset$ by reading some symbol in $\alphabet'$. 
More formally, $$\layerreachablestates{\alayer}{\sset}{\alphabet'}\defeq\set{\layerrightstate \in \layerrightfrontier(\layer) \setst\exists\,\layerleftstate\in \sset, \exists\,\asymbol \in \alphabet',\tuple{\layerleftstate,\asymbol,\layerrightstate} \in \layertransitions(\alayer)}\text{.}$$

For each alphabet $\alphabet$ and each $\width\in \pN$, we let $\lpwmappingname[\alphabet,\width]\colon\layeralphabet{\alphabet}{\width}\rightarrow \cdlayeralphabet{\alphabet}{2^{\width}}$
be the map that sends each layer $\alayer \in \layeralphabet{\alphabet}{\width}$ to the deterministic, complete layer 
$\lpwmappingname(\alayer)\in\cdlayeralphabet{\alphabet}{2^{\width}}$ defined as follows: 
  \begin{itemize}
	  	\item $\layerleftfrontier(\lpwmapping{\alayer}) \defeq 
		\begin{cases}
		  \set{\pwbijection(\layerinitialstates(\alayer))} & \text{ if } \layerinitialflag(\alayer) = \true\\
		  \set{\pwbijection(\sset) \setst \sset \subseteq \layerleftfrontier(\alayer)} & \text{ otherwise;}
		\end{cases}$
		\vspace{1.0ex}

		\item $\layerrightfrontier(\lpwmapping{\alayer}) \defeq \set{\pwbijection(\sset) \setst \sset \subseteq \layerrightfrontier(\alayer)}$; 
		\vspace{1.0ex}
		\item $\layertransitions(\lpwmapping{\alayer}) \defeq 
			\begin{cases}
				\{\tuple*{\pwbijection(\layerinitialstates(\alayer)), \asymbol, \pwbijection(\layerreachablestates{\alayer}{\layerinitialstates(\alayer)}{\set{\asymbol}}}, \asymbol \in \alphabet\} & \text{ if } \layerinitialflag(\alayer) = \true \\
				\{\tuple*{\pwbijection(\sset), \asymbol, \pwbijection(\layerreachablestates{\alayer}{\sset}{\set{\asymbol}}} \setst \sset \subseteq \layerleftfrontier(\alayer), \asymbol\in\alphabet\} & \text{ otherwise; }
			\end{cases}$
		\vspace{0.5ex}
		\item $\layerinitialstates(\lpwmapping{\alayer}) \defeq
			\begin{cases}
			  \set{\pwbijection(\layerinitialstates(\alayer))} & \text{if } \layerinitialflag(\alayer) = \true\\
			  \emptyset & \text{otherwise;}
			\end{cases}$
		\vspace{1.0ex}
		\item $\layerfinalstates(\lpwmapping{\alayer}) \defeq \set{\pwbijection(\sset) \setst \sset \subseteq \layerrightfrontier(\alayer), \sset \cap \layerfinalstates(\alayer) \neq \emptyset}$;
		\vspace{1.0ex}
		\item $\layerinitialflag(\lpwmapping{\alayer}) \defeq \layerinitialflag(\alayer)$; 
		\vspace{1.0ex}
		\item $\layerfinalflag(\lpwmapping{\alayer}) \defeq \layerfinalflag(\alayer)$.
	\end{itemize}

	Let $\alphabet$ be an alphabet, $\width\in\pN$, and let $\alayer \in \layeralphabet{\alphabet}{\width}$. 
	Since $\pwbijection$ is a bijection, there exists precisely one right state $\layerrightstate \in \layerrightfrontier(\lpwmapping{\alayer})$, namely $\layerrightstate=\pwbijection(\layerreachablestates{\alayer}{\sset}{\set{\asymbol}})$, such that $\tuple{\pwbijection(\sset), \asymbol, \layerrightstate} \in \layertransitions(\lpwmapping{\alayer})$ for each subset $\sset \subseteq \dbset{\width}$ with $\pwbijection(\sset) \in \layerleftfrontier(\lpwmapping{\alayer})$ and each symbol $\asymbol \in \alphabet$. 
  	Furthermore, note that $\layerinitialflag(\lpwmapping{\alayer}) = \true$ implies $\layerinitialflag(\alayer)=\true$. 
  	Thus, if $\layerinitialflag(\lpwmapping{\alayer}) = \true$, then $\layerinitialstates(\lpwmapping{\alayer}) = \layerleftfrontier(\lpwmapping{\alayer}) = \set{\pwbijection(\layerinitialstates(\alayer))}$.
  	As a result, $\lpwmapping{\layer}$ is indeed a deterministic, complete layer in $\cdlayeralphabet{\alphabet}{2^{\width}}$.

Now, for each alphabet $\alphabet$ and each positive integer $\width\in \pN$, we define the 
$(\layeralphabet{\alphabet}{\width},\cdlayeralphabet{\alphabet}{\width})$-transduction 
$\determinizationtransduction{\alphabet}{\width} \defeq \multimaptransduction{\lpwmappingname[\alphabet,\width]}$.
The next lemma states that $\determinizationtransduction{\alphabet}{\width}$ sends 
each ODD $\aodd\in\odddefiningsetstar{\alphabet}{\width}$ to a deterministic, complete ODD
$\bodd\in\cdodddefiningsetstar{\alphabet}{\width}$ that has the same language as $\aodd$. 

\begin{lemma}[Determinization Transduction]\label{lemma:DeterminizationTransduction}
For each alphabet $\alphabet$ and each positive integer $\width\in\pN$, the following statements hold.
\begin{enumerate}
	\item\label{DeterminizationOne} $\determinizationtransduction{\alphabet}{\width}$ is functional. 
	\item\label{DeterminizationTwo} $\domain(\determinizationtransduction{\alphabet}{\width}) \supseteq \odddefiningsetstar{\alphabet}{\width}$.
	\item\label{DeterminizationThree} For each pair $(\aodd,\bodd)\in \determinizationtransduction{\alphabet}{\width}$, if $\aodd\in\odddefiningsetstar{\alphabet}{\width}$, then $\bodd\in\cdodddefiningsetstar{\alphabet}{2^\width}$ and $\oddlang{\bodd} = \oddlang{\aodd}$.
	\item\label{DeterminizationFour} $\determinizationtransduction{\alphabet}{\width}$ is $2$-regular. 
\end{enumerate}
\end{lemma}
\begin{proof}
First, we note that $\domain(\determinizationtransduction{\alphabet}{\width}) = \layeralphabet{\alphabet}{\width}^{+}$.
This follows from the fact that $\lpwmappingname$ is a \emph{map from the alphabet $\layeralphabet{\alphabet}{\width}$} to the alphabet $\cdlayeralphabet{\alphabet}{2^{\width}}$. 
Thus, for each $\oddlength\in\pN$ and each string $\aodd=\alayer_{1}\cdots\alayer_{\oddlength}\in\layeralphabet{\alphabet}{\width}^{\oddlength}$, there exists  exactly one string $\bodd$ over $\cdlayeralphabet{\alphabet}{2^{\width}}$ such that $\tuple{\aodd,\bodd}\in\determinizationtransduction{\alphabet}{\width}$, namely the string $\bodd=\lpwmapping{\aodd}=\lpwmapping{\alayer_1}\cdots\lpwmapping{\alayer_{\oddlength}}$. 
Consequently, $\domain(\determinizationtransduction{\alphabet}{\width}) \supseteq \odddefiningsetstar{\alphabet}{\width}$. 
Moreover, by the uniqueness of the string $\bodd=\lpwmapping{\aodd}$ with $\tuple{\aodd,\bodd}\in\determinizationtransduction{\alphabet}{\width}$ for each $\aodd\in\layeralphabet{\alphabet}{\width}^{+}$, we obtain that $\determinizationtransduction{\alphabet}{\width}$ is a functional transduction. 

Now, let $\aodd = \alayer_1\cdots\alayer_{\oddlength}\in\odddefiningset{\alphabet}{\width}{\oddlength}$ for some $\oddlength \in \pN$.
Since $\pwbijection$ is a bijection, for each $i \in \bset{\oddlength-1}$, $\layerleftfrontier(\lpwmapping{\alayer_{i+1}}) = \layerrightfrontier(\lpwmapping{\alayer_i})$ if and only if $\layerleftfrontier\tuple{\alayer_{i+1}} = \layerrightfrontier(\alayer_i)$.
Furthermore, $\layerinitialflag(\lpwmapping{\alayer_i})=\layerinitialflag(\alayer_i)$ and $\layerfinalflag(\lpwmapping{\alayer_i})=\layerfinalflag(\alayer_i)$ for each $i \in \bset{\oddlength}$.
Thus, owing to fact that  $\aodd\in\odddefiningset{\alphabet}{\width}{\oddlength}$, $\lpwmapping{\aodd}=\lpwmapping{\alayer_1}\cdots\lpwmapping{\alayer_{\oddlength}}\in\odddefiningset{\alphabet}{2^{\width}}{\oddlength}$.
More specifically, $\lpwmapping{\aodd}$ is a deterministic, complete ODD in $\cdodddefiningset{\alphabet}{2^{\width}}{\oddlength}$.
Indeed, this follows from the fact that $\lpwmapping{\alayer_i}$ is a deterministic, complete $(\alphabet,\width)$-layer for each $i \in \bset{\oddlength}$. 
Thus, it just remains to prove that $\oddlang{\lpwmapping{\aodd}} = \oddlang{\aodd}$.
Let $\astring = \asymbol_1\cdots\asymbol_{\oddlength}$ be a string in $\alphabet^{\oddlength}$. 

First, suppose that $\astring \in \oddlang{\aodd}$. 
Then, there exists an accepting sequence 
$$\sequence{\tuple{\layerleftstate_1, \asymbol_1, \layerrightstate_1}, \ldots, \tuple{\layerleftstate_{\oddlength}, \asymbol_{\oddlength}, \layerrightstate_{\oddlength}}}$$
for $\astring$ in $\aodd$.
Let $\sset_0 = \layerinitialstates\tuple{\alayer_1}$ and, for each $i \in \dbset{\oddlength}$, let $\sset_{i+1} = \layerreachablestates{\alayer_{i+1}}{\sset_i}{\set{\asymbol_{i+1}}}$.
Note that $\sset_{i} \subseteq \layerleftfrontier(\layer_{i+1})$ for each $i \in \dbset{\oddlength}$. 
Furthermore, for each $i \in \bset{\oddlength}$, we have that $\layerrightstate_{i} \in \sset_{i}$, \ie $\layerrightstate_{i} \in \layerreachablestates{\alayer_{i}}{\sset_{i-1}}{\set{\asymbol_{i}}}$, otherwise $\tuple{\layerleftstate_{i}, \asymbol_{i}, \layerrightstate_{i}} \not\in \layertransitions(\alayer_{i})$.
Therefore,
$$\sequence{\tuple{\pwbijection(\sset_0), \asymbol_1, \pwbijection(\sset_1)}, \ldots, \tuple{\pwbijection(\sset_{\oddlength-1}), \asymbol_{\oddlength}, \pwbijection(\sset_{\oddlength})}}$$
is an accepting sequence for $\astring$ in $\lpwmapping{\aodd}$, and we obtain that $\astring \in \oddlang{\lpwmapping{\aodd}}$.

Conversely, suppose that $\astring \in \oddlang{\lpwmapping{\aodd}}$. 
Then, there exists an accepting sequence $$\sequence{\tuple{\pwbijection(\sset_0), \asymbol_1, \pwbijection(\sset_1)}, \ldots,\tuple{\pwbijection(\sset_{\oddlength-1}), \asymbol_{\oddlength}, \pwbijection(\sset_{\oddlength})}}$$ for $\astring$ in $\lpwmapping{\aodd}$, where $\sset_0 = \layerinitialstates(\alayer_1)$ and $\sset_{i+1} = \layerreachablestates{\alayer_{i+1}}{\sset_i}{\set{\asymbol_{i+1}}}$ for each $i \in \dbset{\oddlength}$. 
Thus, let $\layerleftstate_{\oddlength} \in \sset_{\oddlength-1}$ and   $\layerrightstate_{\oddlength} \in \sset_{\oddlength}$ such that $\tuple{\layerleftstate_{\oddlength}, \asymbol_{\oddlength}, \layerrightstate_{\oddlength}} \in \layertransitions(\alayer_{\oddlength})$.
Moreover, for each $i \in \bset{\oddlength-1}$, let $\layerleftstate_{i} \in \sset_{i-1}$ and $\layerrightstate_{i} \in \sset_{i}$ such that $\layerrightstate_{i} = \layerleftstate_{i+1}$ and $\tuple{\layerleftstate_{i}, \asymbol_{i}, \layerrightstate_{i}} \in \layertransitions(\alayer_{i})$.
We note that for each $i \in \dbset{k}$, there exist left states and right states $\layerleftstate_{i+1}$ and $\layerrightstate_{i+1}$ as described above, otherwise $\tuple{\pwbijection(\sset_{i}), \asymbol_{i+1}, \pwbijection(\sset_{i+1})}$ would not be a transition in $\layertransitions(\lpwmapping{\alayer_{i+1}})$. 
Therefore, 
$$\sequence{\tuple{\layerleftstate_1, \asymbol_1, \layerrightstate_1}, \ldots, \tuple{\layerleftstate_{\oddlength}, \asymbol_{\oddlength}, \layerrightstate_{\oddlength}}}$$
is an accepting sequence for $\astring$ in $\aodd$, and $\astring \in \oddlang{\aodd}$. Finally, the fact that $\determinizationtransduction{\alphabet}{\width}$ is $2$-regular follows from the fact that 
$\determinizationtransduction{\alphabet}{\width} \defeq \multimaptransduction{\lpwmappingname[\alphabet,\width]}$ is an instantiation of a multimap transduction and that 
multimap transductions are $2$-regular (Proposition \ref{proposition:SizeTransductions}.(1)). 
\end{proof}

\subsection{Reachability Transduction}
\label{subsection:ReachabilityTransduction}

In this subsection, we define the {\em reachability transduction}, which intuitively simulates the process of eliminating unreachable
states from the frontiers of each layer of an ODD. It is worth noting that unlike the determinization transduction,
that can be defined using a map that acts layerwisely, the reachability transduction will require the use of a compatibility transduction. 
The issue is that reachability of a given state $\astate$ in a given $\layer$ belonging to a given ODD $\aodd$ is a property that 
depends on which layers have been read before $\alayer$. To circumvent this issue, the action of the reachability transduction on
a ODD $\aodd$ can be described in three intuitive steps. First, we use a multimap transduction to expand each layer of the ODD into 
a set of {\em annotated} layers. Each annotation splits states of a layer into two classes: those that are deemed to be useful, and those 
that should be deleted. Subsequently, we use a compatibility transduction
to ensure that only sequences of annotated layers with compatible annotations are considered to be legal. The crucial observation is that 
each ODD $\aodd$ has a unique annotated version where each two adjacent annotated layers are compatible with each other. Finally, 
we apply a mapping that sends each annotated layer to the layer obtained by deleting the states that have been marked for deletion. The
resulting ODD is then the unique ODD obtained from $\aodd$ by eliminating unreachable states.

Let $\alphabet$ be an alphabet, $\width \in \pN$ and $\alayer\in\cdlayeralphabet{\alphabet}{\width}$.
A \emph{reachability annotation} for $\alayer$ is a pair $\tuple{\reachabilityleftname,\reachabilityrightname}$ of functions $\reachabilityleftname\colon\layerleftfrontier(\layer)\rightarrow\set{\false,\true}$ and $\reachabilityrightname\colon\layerrightfrontier(\layer)\rightarrow\set{\false,\true}$ that satisfies the following conditions: \vspace{1.0ex}
\begin{enumerate}
	\item if $\layerinitialflag(\layer)=\true$, then, for each left state $\layerleftstate\in\layerleftfrontier(\layer)$, $\reachabilityleft{\layerleftstate}=\true$ if and only if $\layerleftstate\in\layerinitialstates(\layer)$; 
	\item for each right state $\layerrightstate\in\layerrightfrontier(\layer)$, $\reachabilityright{\layerrightstate}=\true$ if and only if there exists $\layerleftstate\in\layerleftfrontier(\layer)$ and $\asymbol\in\alphabet$ 
	such that $\reachabilityleft{\layerleftstate}=\true$ and $\tuple{\layerleftstate,\asymbol,\layerrightstate}\in\layertransitions(\layer)$. 
\end{enumerate}

Let $\alphabet$ be an alphabet, $\width,\length\in\pN$, and let $\aodd = \alayer_{1}\cdots\alayer_{\oddlength}\in\cdodddefiningset{\alphabet}{\width}{\oddlength}$. 
A {\em reachability annotation} for $\aodd$ is a sequence $\sequence{\tuple{\reachabilityleftname_{1},\reachabilityrightname_{1}},\ldots,\tuple{\reachabilityleftname_{\oddlength},\reachabilityrightname_{\oddlength}}}$ that satisfies the following conditions: 
\begin{enumerate}
	\item for each $i\in \bset{\oddlength}$, $\tuple{\reachabilityleftname_i,\reachabilityrightname_i}$ is a 
		reachability annotation for $\alayer_i$; 
	\item for each $i\in \bset{\oddlength-1}$, $\reachabilityrightname_{i}=\reachabilityleftname_{i+1}$. 
\end{enumerate}

\begin{proposition}\label{proposition:UniqueReachability}
	Let $\alphabet$ be an alphabet and $\width \in \pN$. 
	Every ODD $\aodd \in \cdodddefiningsetstar{\alphabet}{\width}$
	admits a unique reachability annotation. 
\end{proposition}
\begin{proof}
	First, we observe that for each layer $\alayer\in\cdlayeralphabet{\alphabet}{\width}$ and each function $\reachabilityleftname\colon\layerleftfrontier(\alayer)\rightarrow\set{\false,\true}$, there exists exactly one function $\reachabilityrightname\colon\layerrightfrontier(\layer)\rightarrow\set{\false,\true}$ such that $\tuple{\reachabilityleftname,\reachabilityrightname}$ is a 
	reachability annotation for $\alayer$. 

	Let $\oddlength\in\pN$ and $\aodd=\alayer_{1}\cdots\alayer_{\oddlength}\in\cdodddefiningset{\alphabet}{\width}{\oddlength}$, such that $\layerleftfrontier(\alayer_{i+1})=\layerrightfrontier(\alayer_{i})$ for each $i\in\bset{\oddlength-1}$, and $\layerinitialflag(\alayer_{1})=1$ and $\layerinitialflag(\alayer_{i})=0$ for each $i\in\set{2,\ldots,\oddlength}$. 
	Based on the previous observation, we prove by induction on $\oddlength$ that the following statement holds: there exists a unique sequence $\sequence{\tuple{\reachabilityleftname_{1},\reachabilityrightname_{1}},\ldots,\tuple{\reachabilityleftname_{\oddlength},\reachabilityrightname_{\oddlength}}}$ such that $\reachabilityleftname_{i}=\reachabilityrightname_{i+1}$ for each $i\in\bset{\oddlength-1}$, and $\tuple{\reachabilityleftname_{i},\reachabilityrightname_{i}}$ is a reachability annotation for $\alayer_i$ 
	for each $i\in\bset{\oddlength}$. 

	\emph{Base case.} Consider $\oddlength=1$.  
	 Since $\layerinitialflag(\alayer_{\oddlength})=\true$, the function $\reachabilityleftname_{\oddlength}\colon\layerleftfrontier(\alayer_{\oddlength})\rightarrow\set{\false,\true}$ is uniquely determined. 
	 Indeed, by definition, for each left state $\layerleftstate\in\layerleftfrontier(\alayer_{\oddlength})$, $\reachabilityleft[\oddlength]{\layerleftstate}=\true$ if $\layerleftstate\in\layerinitialstates(\alayer_{\oddlength})$, and $\reachabilityleft[\oddlength]{\layerleftstate}=\false$ otherwise. 
	 Thus, there exists a unique sequence $\sequence{\tuple{\reachabilityleftname_{\oddlength},\reachabilityrightname_{\oddlength}}}$ such that $\tuple{\reachabilityleftname_{\oddlength},\reachabilityrightname_{\oddlength}}$ is a reachability annotation for $\alayer_{\oddlength}$. 

	\emph{Inductive step.} Consider $\oddlength>1$.  
 	Let $\bodd=\alayer_{1}\cdots\alayer_{\oddlength-1}$ be the string obtained from $\aodd=\alayer_{1}\cdots\alayer_{\oddlength}$ by removing the layer $\alayer_{\oddlength}$. 
	It follows from the inductive hypothesis that there exists a unique sequence $\sequence{\tuple{\reachabilityleftname_{1},\reachabilityrightname_{1}},\ldots,\tuple{\reachabilityleftname_{\oddlength-1},\reachabilityrightname_{\oddlength-1}}}$ such that $\reachabilityleftname_{i}=\reachabilityrightname_{i+1}$ for each $i\in\bset{\oddlength-2}$, and $\tuple{\reachabilityleftname_{i},\reachabilityrightname_{i}}$ is a reachability annotation for $\alayer_i$ for each $i\in\bset{\oddlength-1}$.  	
	In particular, we note that the function $\reachabilityrightname_{\oddlength-1}$ is uniquely determined. 
	Furthermore, based on the previous observation, for each function $\reachabilityleftname_{\oddlength}\colon\layerleftfrontier(\alayer_{\oddlength})\rightarrow\set{\false,\true}$, 
	there exists a unique function $\reachabilityrightname_{\oddlength}\colon\layerrightfrontier(\alayer_{\oddlength})\rightarrow\set{\false,\true}$ 
	such that $\tuple{\reachabilityleftname_{\oddlength},\reachabilityrightname_{\oddlength}}$ is a reachability annotation for $\alayer_{\oddlength}$.  
	Therefore, since $\reachabilityleftname_{\oddlength}$ must be equal to $\reachabilityrightname_{\oddlength-1}$, there exists a unique sequence $\sequence{\tuple{\reachabilityleftname_{1},\reachabilityrightname_{1}},\ldots,\tuple{\reachabilityleftname_{\oddlength},\reachabilityrightname_{\oddlength}}}$ such that $\reachabilityleftname_{i}=\reachabilityrightname_{i+1}$ for each $i\in\bset{\oddlength-1}$ and $\tuple{\reachabilityleftname_{i},\reachabilityrightname_{i}}$ is a reachability annotation for $\alayer_i$ for each $i\in\bset{\oddlength}$. 
\end{proof}

Let $\alphabet$ be an alphabet and $\width\in \pN$. 
We denote by $\reachabilityalphabet{\alphabet}{\width}$ the set consisting of all triples $\tuple{\alayer,\reachabilityleftname,\reachabilityrightname}$ such that $\alayer$ is a layer in $\cdlayeralphabet{\alphabet}{\width}$ and $\tuple{\reachabilityleftname,\reachabilityrightname}$ is a reachability annotation for $\alayer$. 
Additionally, we denote by $\removingname{\alphabet}{\width}\colon\reachabilityalphabet{\alphabet}{\width}\rightarrow\cdlayeralphabet{\alphabet}{\width}$ the map that sends each triple $\tuple{\alayer,\reachabilityleftname,\reachabilityrightname}\in\reachabilityalphabet{\alphabet}{\width}$ to the layer $\removing{\alphabet}{\width}{\alayer,\reachabilityleftname,\reachabilityrightname}\in\cdlayeralphabet{\alphabet}{\width}$ obtained from $\alayer$ by removing the left states $\layerleftstate\in\layerleftfrontier(\layer)$ with $\reachabilityleft{\layerleftstate}=\false$, 
the right states $\layerrightstate\in\layerrightfrontier(\layer)$ with $\reachabilityright{\layerrightstate}=\false$, and the transitions incident with such left and right states. 
More formally, for each triple $\tuple{\alayer,\reachabilityleftname,\reachabilityrightname}\in\reachabilityalphabet{\alphabet}{\width}$,
we let $\removing{\alphabet}{\width}{\alayer,\reachabilityleftname,\reachabilityrightname} = \blayer$, where $\blayer$ is the layer belonging to $\cdlayeralphabet{\alphabet}{\width}$ defined as follows:
\begin{itemize}
    \item $\layerleftfrontier(\blayer) \defeq \layerleftfrontier(\alayer)\setminus\set{\layerleftstate \setst \reachabilityleft{\layerleftstate}=\false}$;
    \item $\layerrightfrontier(\blayer) \defeq \layerrightfrontier(\alayer)\setminus\set{\layerrightstate \setst \reachabilityright{\layerrightstate}=\false}$; \vspace{1.0ex}
    \item $\layertransitions(\blayer) \defeq \layertransitions(\alayer)\setminus\set{\tuple{\layerleftstate,\asymbol,\layerrightstate} \setst \reachabilityleft{\layerleftstate}=\false}$;\vspace{1.0ex} 
    \item $\layerinitialflag(\blayer) \defeq \layerinitialflag(\alayer)$; $\layerfinalflag(\blayer) \defeq \layerfinalflag(\alayer)$; \vspace{1.0ex}
    \item $\layerinitialstates(\blayer) \defeq \layerinitialstates(\alayer)$; $\layerfinalstates(\blayer) \defeq \layerrightfrontier(\blayer) \cap \layerfinalstates(\alayer)$.
\end{itemize}

We let 
$\oddremovingname{\alphabet}{\width}\colon\cdodddefiningsetstar{\alphabet}{\width}\rightarrow\cdodddefiningsetstar{\alphabet}{\width}$ be the map that for each $\oddlength\in\pN$, sends each 
ODD $\odd = \layer_{1}\cdots\layer_{\oddlength}\in\cdodddefiningset{\alphabet}{\width}{\oddlength}$ to the ODD 
$$\oddremoving{\alphabet}{\width}{\odd} \defeq \removing{\alphabet}{\width}{\layer_{1},\reachabilityleftname_{1},\reachabilityrightname_{1}}\cdots\removing{\alphabet}{\width}{\layer_{\oddlength},\reachabilityleftname_{\oddlength},\reachabilityrightname_{\oddlength}}\in\cdodddefiningset{\alphabet}{\width}{\oddlength}\text{,}$$ where $\sequence{\tuple{\reachabilityleftname_{1},\reachabilityrightname_{1}},\ldots,\tuple{\reachabilityleftname_{\oddlength},\reachabilityrightname_{\oddlength}}}$ denotes the unique reachability annotation for $\odd$ (see Proposition~\ref{proposition:UniqueReachability}).

\begin{proposition}\label{proposition:UniqueReachabilityODD}
	Let $\alphabet$ be an alphabet, $\width\in\pN$ and $\odd\in\cdodddefiningsetstar{\alphabet}{\width}$. 
	Then, $\oddremoving{\alphabet}{\width}{\odd}$ is a reachable ODD in $\cdodddefiningsetstar{\alphabet}{\width}$
 	such that $\oddlang{\oddremoving{\alphabet}{\width}{\odd}} = \oddlang{\odd}$. 
\end{proposition}
\begin{proof}
	Assume that $\odd=\alayer_{1}\cdots\alayer_{\oddlength}$ and $\oddremoving{\alphabet}{\width}{\odd} = \blayer_{1}\cdots\blayer_{\oddlength}$, for some $\oddlength\in\pN$, where $\blayer_{i} = \removing{\alphabet}{\width}{\alayer_{i},\reachabilityleftname_{i},\reachabilityrightname_{i}}$ for each $i \in \bset{\oddlength}$ and $\sequence{\tuple{\reachabilityleftname_{1},\reachabilityrightname_{1}},\ldots,\tuple{\reachabilityleftname_{\oddlength},\reachabilityrightname_{\oddlength}}}$ is the unique reachability annotation of $\odd$. 
	First, we prove that $\oddremoving{\alphabet}{\width}{\odd}$ is reachable. 
	Note that for each $i \in \bset{\oddlength}$ and each $\layerrightstate\in\layerrightfrontier(\alayer_{i})$, 
	\begin{equation*}
		\begin{array}{lcl}
			\layerrightstate\in\layerrightfrontier(\blayer_{i}) & \Leftrightarrow & \exists\, \layerleftstate\in\layerleftfrontier(\alayer_{i}), \text{ with } \reachabilityleft[i]{\layerleftstate} = \true, \text{ and } \exists\, \asymbol\in\alphabet \text{ such that } \tuple{\layerleftstate,\asymbol,\layerrightstate}\in\layertransitions(\alayer_{i}) \\[0.5ex]
			& \Leftrightarrow & \exists\, \layerleftstate\in\layerleftfrontier(\blayer_{i}) \text{ and } \exists\, \asymbol\in\alphabet \text{ such that } \tuple{\layerleftstate,\asymbol,\layerrightstate}\in\layertransitions(\blayer_{i})\text{.}
		\end{array}
	\end{equation*}
	This implies that for each $i \in \bset{\oddlength}$, $\blayer_{i}$ is a reachable layer since $\layerrightfrontier(\blayer_{i}) \subseteq \layerrightfrontier(\alayer_{i})$. 
	Therefore, $\oddremoving{\alphabet}{\width}{\odd}$ is a reachable ODD. 
	Now, we prove that $\oddlang{\oddremoving{\alphabet}{\width}{\odd}} = \oddlang{\odd}$. 
	It is immediate from the definition of $\oddremoving{\alphabet}{\width}{\odd}$ that $\oddlang{\oddremoving{\alphabet}{\width}{\odd}} \subseteq \oddlang{\odd}$. 
	On the other hand, it is not hard to check that for each string $\astring \in \alphabet^{\oddlength}$, every accepting sequence for $\astring$ in $\odd$ is also an accepting sequence for $\astring$  in 
	$\oddremoving{\alphabet}{\width}{\odd}$. 
	Consequently, $\oddlang{\oddremoving{\alphabet}{\width}{\odd}} \supseteq \oddlang{\odd}$. 
	
	To prove that $\oddremovingname{\alphabet}{\width}$ preserves determinism, it is enough to note that
	$\layertransitions(\blayer_{i})\subseteq \layertransitions(\alayer_{i})$ for each $i \in \bset{\oddlength}$. 
	As a result, since $\odd$ is deterministic, so is $\oddremoving{\alphabet}{\width}{\odd}$. 
	Finally, since $\odd$ is complete, by definition, for each $i \in \bset{\oddlength}$ and each $\layerleftstate \in \layerleftfrontier(\alayer_{i}) \cap \layerleftfrontier(\blayer_{i})$, there exists a symbol $\asymbol$ and a right state $\layerrightstate \in \layerrightfrontier(\alayer_{i})$ such that $\tuple{\layerleftstate,\asymbol,\layerrightstate} \in \layertransitions(\alayer_{i})$. 
	This implies that for each $i \in \bset{\oddlength}$ and each $\layerleftstate \in \layerleftfrontier(\blayer_{i})$, there exists a symbol $\asymbol$ and a right state $\layerrightstate \in \layerrightfrontier(\blayer_{i})$ such that $\tuple{\layerleftstate,\asymbol,\layerrightstate} \in \layertransitions(\blayer_{i})$. 
	Therefore, $\oddremoving{\alphabet}{\width}{\odd}$ is also complete. 
\end{proof}

For each alphabet $\alphabet$ and each positive integer $\width\in \pN$, we let $\reachabilityrelation{\alphabet}{\width} \subseteq \cdlayeralphabet{\alphabet}{\width} \times \reachabilityalphabet{\alphabet}{\width}$ 
and $\reachabilitycompatibility{\alphabet}{\width} \subseteq \reachabilityalphabet{\alphabet}{\width}\times \reachabilityalphabet{\alphabet}{\width}$ be the relations defined as follows. 

$$\reachabilityrelation{\alphabet}{\width} \defeq \set{\tuple{\alayer, (\alayer,\reachabilityleftname,\reachabilityrightname)} \setst \tuple{\alayer,\reachabilityleftname,\reachabilityrightname}\in\reachabilityalphabet{\alphabet}{\width}}. 
$$ 
$$
\begin{multlined}[t][0.95\textwidth]
\reachabilitycompatibility{\alphabet}{\width} \defeq \{
	\tuple{(\alayer,\reachabilityleftname,\reachabilityrightname), (\alayer',\reachabilityleftname',\reachabilityrightname')} \setst \tuple{\alayer,\reachabilityleftname,\reachabilityrightname}, \tuple{\alayer',\reachabilityleftname',\reachabilityrightname'} \in \reachabilityalphabet{\alphabet}{\width}, \layerrightfrontier(\alayer) = \layerleftfrontier(\alayer'),\; \reachabilityrightname = \reachabilityleftname'\}\text{.}
\end{multlined}
$$

Now, for each alphabet $\alphabet$ and each positive integer $\width \in \pN$, we define $\reachabilitytransduction{\alphabet}{\width}$ as the $(\cdlayeralphabet{\alphabet}{\width},\cdlayeralphabet{\alphabet}{\width})$-transduction
$$\reachabilitytransduction{\alphabet}{\width} \defeq \multimaptransduction{\reachabilityrelation{\alphabet}{\width}}\circ\compTransduction{\reachabilitycompatibility{\alphabet}{\width}} \circ \multimaptransduction{\removingname{\alphabet}{\width}}.$$

The next lemma states that $\reachabilitytransduction{\alphabet}{\width}$ is a transduction that sends each ODD $\aodd\in\cdodddefiningsetstar{\alphabet}{\width}$ to a \emph{reachable} ODD $\bodd\in\cdodddefiningsetstar{\alphabet}{\width}$ that has the same language as $\aodd$, and that preserves the determinism and completeness properties.

\begin{lemma}[Reachability Transduction]\label{lemma:ReachabilityTransduction}
    For each alphabet $\alphabet$ and each positive integer $\width\in \pN$, the following statements hold.  
    \begin{enumerate}
        \item\label{ReachabilityOne} $\reachabilitytransduction{\alphabet}{\width}$ is functional.
        \item\label{ReachabilityTwo} $\domain(\reachabilitytransduction{\alphabet}{\width}) \supseteq \cdodddefiningsetstar{\alphabet}{\width}$.
        \item\label{ReachabilityThree} For each pair $\tuple{\aodd,\bodd}\in \reachabilitytransduction{\alphabet}{\width}$, 
	$\oddlang{\bodd} = \oddlang{\aodd}$ and $\bodd$ is reachable. 
	\item\label{ReachabilityFour} $\reachabilitytransduction{\alphabet}{\width}$ is $2^{O(|\alphabet|\cdot \width\cdot \log \width)}$-regular.
    \end{enumerate}
\end{lemma}
\begin{proof}
	We note that $\reachabilitytransduction{\alphabet}{\width}$ consists of all pairs $\tuple{\aodd,\bodd}$ of non-empty strings over the alphabet $\cdlayeralphabet{\alphabet}{\width}$ satisfying the conditions that $\abs{\aodd}=\abs{\bodd}$ and that, if $\aodd=\alayer_{1}\cdots\alayer_{\oddlength}$ and $\bodd=\blayer_{1}\cdots\blayer_{\oddlength}$ for some $\oddlength\in\pN$, then there exists a reachability annotation $\tuple{\reachabilityleftname_{i},\reachabilityrightname_{i}}$ for the layer $\alayer_{i}$ such that $\blayer_{i}=\removing{\alphabet}{\width}{\alayer_{i},\reachabilityleftname_{i},\reachabilityrightname_{i}}$ for each $i\in\bset{\oddlength}$, and $\layerrightfrontier(\alayer_{j})=\layerleftfrontier(\alayer_{j+1})$ and $\reachabilityrightname_{j}=\reachabilityleftname_{j+1}$ for each $j\in\bset{\oddlength-1}$. 
	Additionally, based on Proposition~\ref{proposition:UniqueReachability}, each $(\alphabet,\width)$-ODD admits a
	unique reachability annotation. 
	As a result, we obtain that $\domain(\reachabilitytransduction{\alphabet}{\width}) \supseteq \cdodddefiningsetstar{\alphabet}{\width}$.
	Moreover, $\bodd = \oddremoving{\alphabet}{\width}{\aodd}$; thus, by the uniqueness of $\oddremoving{\alphabet}{\width}{\aodd}$, the transduction 
	$\reachabilitytransduction{\alphabet}{\width}$ is functional.
	Finally, it follows from Proposition~\ref{proposition:UniqueReachabilityODD} that for each pair $\tuple{\aodd,\bodd}\in \reachabilitytransduction{\alphabet}{\width}$, 
	$\bodd = \oddremoving{\alphabet}{\width}{\aodd}$ is a reachable ODD in $\cdodddefiningsetstar{\alphabet}{\width}$ that has the same language as $\aodd$.
	
	The fact that $\reachabilitytransduction{\alphabet}{\width}$ is $2^{O(|\alphabet|\cdot \width \cdot \log \width)}$-regular follows from 
	Proposition \ref{proposition:PropertiesTransductions}.(2) together with the fact that the multimap transductions 
	$\multimaptransduction{\reachabilityrelation{\alphabet}{\width}}$ and $\multimaptransduction{\removingname{\alphabet}{\width}}$ are $2$-regular 
	(Proposition \ref{proposition:SizeTransductions}.(1)), and that the transduction $\compTransduction{\reachabilitycompatibility{\alphabet}{\width}}$ is 
	$2^{O(|\alphabet|\cdot \width \cdot \log \width)}$-regular
	(Proposition \ref{proposition:SizeTransductions}.(2)), given that 
	$\reachabilitycompatibility{\alphabet}{\width}\subseteq \reachabilityalphabet{\alphabet}{\width}\times \reachabilityalphabet{\alphabet}{\width}$ and that 
	$|\reachabilityalphabet{\alphabet}{\width}| = 2^{O(|\alphabet|\cdot \width\cdot \log\width)}$.
\end{proof}

\subsection{Merging Transduction}
\label{subsection:MergingTransduction}

In this subsection, we define the {\em merging transduction}, which intuitively simulates the process 
of merging equivalent states in the frontiers of each layer of an ODD $\aodd$. As in the case of the reachability 
transduction, the merging transduction will be defined as the composition of three elementary transductions. 
First, we use a multimap transduction to expand each layer of the ODD into 
a set of {\em annotated} layers. Each annotation partitions each frontier of 
the layer into cells containing states that are deemed to be equivalent. Subsequently, we use a compatibility transduction
to ensure that only sequences of annotated layers with compatible annotations are considered to be legal.
As in the case of the reachability transduction, it is possible to show that each ODD $\aodd$ 
has a unique annotated version where each two adjacent annotated layers are compatible with each other. 
Finally, we apply a mapping that sends each annotated layer to the layer obtained by merging all states 
in each cell of each partition to the smallest state in the cell. The result is a minimized ODD with same
language as $\aodd$.

Let $\alphabet$ be an alphabet, $\width \in \pN$, $\alayer\in\cdlayeralphabet{\alphabet}{\width}$ and $\mergingright$ be a partition of $\layerrightfrontier(\alayer)$. 
Two (not necessarily distinct) left states $\layerleftstate,\layerleftstate'\in \layerleftfrontier(\alayer)$ are said to be \emph{$\mergingright$-equivalent} if,
for each symbol $\asymbol\in\alphabet$, there exists a right state $\layerrightstate\in\rightfrontier(\alayer)$ such that $\tuple{\layerleftstate,\asymbol,\layerrightstate}$
is a transition in $\layertransitions(\alayer)$ if and only if there exists a right state $\layerrightstate'\in\layerrightfrontier(\alayer)$ 
such that $\tuple{\layerleftstate',\asymbol,\layerrightstate'}$ is a transition in $\layertransitions(\alayer)$, 
and $\layerrightstate$ and $\layerrightstate'$ belong to the same cell of $\mergingright$. 
We remark that each left state $\layerleftstate$ is trivially $\mergingright$-equivalent to itself.

A \emph{merging annotation} for $\alayer$ is a pair $\tuple{\mergingleft,\mergingright}$,
where $\mergingleft$ is a partition of $\leftfrontier(\alayer)$ and $\mergingright$ is a partition of $\rightfrontier(\alayer)$, that satisfies the following two conditions: 
\begin{enumerate}
	\item if $\finalflag(\alayer) = \true$, then $\mergingright = \set{\layerrightfrontier(\alayer)\setminus\layerfinalstates(\alayer),\layerfinalstates(\alayer)}$
	whenever $\layerrightfrontier(\alayer)\setminus\layerfinalstates(\alayer)\neq\emptyset$ and $\layerfinalstates(\alayer) \neq \emptyset$, and $\mergingright = \set{\layerrightfrontier(\alayer)}$ whenever $\layerrightfrontier(\alayer)\setminus\layerfinalstates(\alayer)=\emptyset$ or $\layerfinalstates(\alayer) = \emptyset$; 
	\item for each two left states $\layerleftstate,\layerleftstate'\in\layerleftfrontier(\alayer)$, $\layerleftstate$ and $\layerleftstate'$ belong to the same cell of $\mergingleft$ if and only if $\layerleftstate$ and $\layerleftstate'$ are $\mergingright$-equivalent. 
\end{enumerate}

Let $\alphabet$ be an alphabet, $\width,\length\in\pN$, and let $\aodd = \alayer_{1}\cdots\alayer_{\oddlength}\in\cdodddefiningset{\alphabet}{\width}{\oddlength}$. 
A {\em merging annotation} for $\aodd$ is a sequence $\sequence{\tuple{\mergingleft_{1},\mergingright_{1}}\cdots\tuple{\mergingleft_{\oddlength},\mergingright_{\oddlength}}}$ that satisfies the following conditions: 
\begin{enumerate}
	\item for each $i\in \bset{\oddlength}$, $\tuple{\mergingleft_i,\mergingright_i}$ is a merging annotation for $\alayer_i$; 
	\item for each $i\in \bset{\oddlength-1}$, $\mergingright_{i}=\mergingleft_{i+1}$. 
\end{enumerate}

\begin{proposition}\label{proposition:UniqueMerging}
Let $\alphabet$ be an alphabet and $\width \in \pN$. 
Every deterministic, complete $(\alphabet,\width)$-ODD admits a unique merging annotation. 
\end{proposition}
\begin{proof}
First, we claim that for each layer $\alayer\in\cdlayeralphabet{\alphabet}{\width}$ and each partition $\mergingright$ of 
$\layerrightfrontier(\alayer)$, there exists a unique partition $\mergingleft$ of $\layerleftfrontier(\alayer)$ such that $\tuple{\mergingleft,\mergingright}$ is a merging annotation for $\alayer$. 
Indeed, any two left states $\layerleftstate,\layerleftstate'\in\layerleftfrontier(\alayer)$ belong to the same cell of $\mergingleft$ if and only if they are $\mergingright$-equivalent. 
Thus, the partition $\mergingleft$ is uniquely defined as the set of all maximal subsets $\sset\subseteq\layerleftfrontier(\alayer)$ of pairwise $\mergingright$-equivalent left states.  

Let $\oddlength\in\pN$ and $\aodd=\alayer_{1}\cdots\alayer_{\oddlength}\in\cdodddefiningset{\alphabet}{\width}{\oddlength}$, be 
such that $\layerleftfrontier(\alayer_{i+1})=\layerrightfrontier(\alayer_{i})$ for each $i\in\bset{\oddlength-1}$, 
$\layerfinalflag(\alayer_{i})=0$ for each $i\in\bset{\oddlength-1}$ and  $\layerfinalflag(\alayer_{\oddlength})=1$. 
Based on the previous claim, we prove by induction on $j$ that the following statement holds for each $j\in \{0,\dots,\oddlength-1\}$: there exists
a unique sequence $\sequence{\tuple{\mergingleft_{\oddlength-j},\mergingright_{\oddlength-j}}\cdots\tuple{\mergingleft_{\oddlength},\mergingright_{\oddlength}}}$ 
such that $\tuple{\mergingleft_{i},\mergingright_{i}}$ is a merging annotation for $\alayer_{i}$ for each $i\in \{j,\dots,\oddlength\}$, 
and  $\mergingright_{i} = \mergingleft_{i+1}$ for each $i\in \{\oddlength-j,\dots,\oddlength-1\}$. In particular, this implies that 
the ODD $\aodd$ admits a unique merging annotation $\sequence{\tuple{\mergingleft_{1},\mergingright_{1}}\cdots\tuple{\mergingleft_{\oddlength},\mergingright_{\oddlength}}}$. 

\emph{Base case.} Consider $j=0$. Then $k-j=k$. Since $\layerfinalflag(\alayer_{\oddlength})=\true$, the partition $\mergingright_{\oddlength}$ is uniquely determined. 
 Indeed, $\mergingright_{\oddlength} = \set{\layerfinalstates(\alayer_{\oddlength}),\layerrightfrontier(\alayer_{\oddlength})\setminus\layerfinalstates(\alayer_{\oddlength})}$ if both 
$\layerrightfrontier(\alayer_{\oddlength})\setminus\layerfinalstates(\alayer_{\oddlength})\neq\emptyset$ and $\layerfinalstates(\alayer_{\oddlength}) \neq \emptyset$, and 
$\mergingright_{\oddlength} = \set{\layerrightfrontier(\alayer_{\oddlength})}$ otherwise. Thus, there exists a unique sequence $\sequence{\tuple{\mergingleft_{\oddlength},\mergingright_{\oddlength}}}$ 
such that $\tuple{\mergingleft_{\oddlength},\mergingright_{\oddlength}}$ is a merging annotation for $\alayer_{\oddlength}$. 
 
\emph{Inductive step.} Consider $j\in \{1,\dots,\oddlength-1\}$. 
We show that there is a unique sequence 
$$\sequence{\tuple{\mergingleft_{\oddlength-j},\mergingright_{\oddlength-j}}\cdots\tuple{\mergingleft_{\oddlength},\mergingright_{\oddlength}}}$$
such that $\tuple{\mergingleft_{i},\mergingright_{i}}$  is a merging annotation for $\alayer_{i}$ for each $i\in \set{\oddlength-j,\ldots,\oddlength}$, 
and $\mergingright_{i} = \mergingleft_{i+1}$ for each $i\in \set{\oddlength-j,\ldots,\oddlength-1}$. 
It follows from the inductive hypothesis that there exists a unique sequence 
$\sequence{\tuple{\mergingleft_{\oddlength-(j-1)},\mergingright_{\oddlength-(j-1)}}\cdots\tuple{\mergingleft_{\oddlength},\mergingright_{\oddlength}}}$
such that $\tuple{\mergingleft_{i},\mergingright_{i}}$  is a merging annotation for $\alayer_{i}$ for each $i\in \set{\oddlength-(j-1),\ldots,\oddlength}$, 
and $\mergingright_{i} = \mergingleft_{i+1}$ for each $i\in \set{\oddlength-(j-1),\ldots,\oddlength-1}$. 
Now, let $(\mergingleft_{\oddlength-j},\mergingright_{\oddlength-j})$ be the merging annotation of $\alayer_{\oddlength-j}$ 
with the property that $\mergingright_{\oddlength-j}=\mergingleft_{\oddlength-(j-1)}$. Such a merging annotation exists 
(since $\layerrightfrontier(\alayer_{\oddlength-j}) = \layerleftfrontier(\alayer_{\oddlength-(j-1)})$) and is unique 
since $\mergingleft_{\oddlength-j}$ is uniquely determined by $\mergingright_{\oddlength-j}$. This concludes the proof of the 
inductive step, and therefore of the proposition.  
\end{proof}

Let $\alphabet$ be an alphabet, $\width,\oddlength\in \pN$ and $\odd=\layer_{1}\cdots\layer_{\oddlength}\in\cdodddefiningset{\alphabet}{\width}{\oddlength}$.
For each $i\in\bset{\oddlength}$, we say that a string $\astring=\asymbol_{1}\ldots\asymbol_{\oddlength}$ is \emph{accepted by $\odd$ from a left state $\layerleftstate\in\layerleftfrontier(\layer_{i})$} if there exists a sequence $\sequence{\tuple{\layerleftstate_{i},\asymbol_{i},\layerrightstate_{i}},\ldots,\tuple{\layerleftstate_{\oddlength},\asymbol_{\oddlength},\layerrightstate_{\oddlength}}}$ of transitions such that $\layerleftstate_{i}=\layerleftstate$, $\layerrightstate_{\oddlength}\in\layerfinalstates(\layer_{\oddlength})$ and, for each $j\in\set{i,\ldots,\oddlength}$, $\tuple{\layerleftstate_{j},\asymbol_{j},\layerrightstate_{j}}\in\layertransitions(\layer_{j})$. 
For each  $i\in\bset{\oddlength}$ and each left state $\layerleftstate\in\layerleftfrontier(\layer_{i})$, we let $$\lstateoddlang{\odd}{i}{\layerleftstate}\defeq \set{\astring\in\alphabet^{\oddlength-i+1}\setst \astring \text{ is accepted by } \odd \text{ from } \layerleftstate}\text{.}$$

\begin{proposition}
\label{prop:mergingLanguagesFromLeftState}
Let $\alphabet$ be an alphabet, $\width,\oddlength\in \pN$, $\odd=\layer_{1}\cdots\layer_{\oddlength}$ be a deterministic,
complete ODD in $\cdodddefiningset{\alphabet}{\width}{\oddlength}$, and let $\sequence{\tuple{\mergingleft_{1},\mergingright_{1}}\cdots\tuple{\mergingleft_{\oddlength},\mergingright_{\oddlength}}}$
be the unique merging annotation for $\odd$. For each  $i\in\bset{\oddlength}$ and each two left states $\layerleftstate,\layerleftstate'\in\layerleftfrontier(\layer_{i})$, $\layerleftstate$ 
and $\layerleftstate'$ belong to the same cell of $\mergingleft_{i}$ if and only if $\lstateoddlang{\odd}{i}{\layerleftstate}=\lstateoddlang{\odd}{i}{\layerleftstate'}$. 
\end{proposition}
\begin{proof}
	The proof is by induction on $\oddlength-i$. 
	\emph{Base case.} Consider $\oddlength-i=0$. Then $i=\oddlength$.  
	By definition, two left states $\layerleftstate,\layerleftstate'\in\layerleftfrontier(\layer_{\oddlength})$ belong to the same cell of $\mergingleft_{\oddlength}$ if and only if $\layerleftstate$ and $\layerleftstate'$ are $\mergingright_{\oddlength}$-equivalent. 
	In other words, $\layerleftstate$ and $\layerleftstate'$ belong to the same cell of $\mergingleft_{\oddlength}$ if and only if, for each symbol $\asymbol\in\alphabet$, there exists a final state $\layerrightstate\in\layerfinalstates(\layer_{\oddlength})$ such that $\tuple{\layerleftstate,\asymbol,\layerrightstate}\in\layertransitions(\layer_{\oddlength})$ if and only if there exists a final state $\layerrightstate'\in\layerfinalstates(\layer_{\oddlength})$ (possibly $\layerrightstate'=\layerrightstate$) such that $\tuple{\layerleftstate',\asymbol,\layerrightstate'}\in\layertransitions(\layer_{\oddlength})$. 
	Consequently, $\layerleftstate$ and $\layerleftstate'$ belong to the same cell of $\mergingleft_{\oddlength}$ if and only if $\lstateoddlang{\odd}{i}{\layerleftstate}=\lstateoddlang{\odd}{i}{\layerleftstate'}$. 

	\emph{Inductive step.} Consider $\oddlength -i > 0$. 
	Since $\layerrightfrontier(\layer_{i})=\layerleftfrontier(\layer_{i+1})$ and $\mergingright_{i}=\mergingleft_{i+1}$, it follows from the inductive hypothesis
	that any two right states $\layerrightstate,\layerrightstate'\in\layerrightfrontier(\layer_{i})$ belong to the same cell of $\mergingright_{i}$ if and only if 
	$\lstateoddlang{\odd}{i+1}{\layerrightstate}=\lstateoddlang{\odd}{i+1}{\layerrightstate'}$. Moreover, note that for each left state
	$\layerleftstate\in\layerleftfrontier(\layer_{i})$, 
	\begin{equation}\label{eq:layerODDLang}
		\lstateoddlang{\odd}{i}{\layerleftstate}=\bigcup_{\layerrightstate\in\layerrightfrontier(\layer_{i})}\set{\asymbol\bstring\in\alphabet^{\oddlength-i+1} \setst \tuple{\layerleftstate,\asymbol,\layerrightstate}\in\layertransitions(\layer_{i}), \bstring\in\lstateoddlang{\odd}{i+1}{\layerrightstate}}\text{.}
	\end{equation}
	Let $\layerleftstate,\layerleftstate'\in\layerleftfrontier(\layer_{i})$. We will prove that $\layerleftstate,\layerleftstate'$ belong to the same cell of $\mergingleft_i$ if and only if 
	 $\lstateoddlang{\odd}{i}{\layerleftstate}=\lstateoddlang{\odd}{i}{\layerleftstate'}$. The proof is split in two parts. 

	First, suppose that $\layerleftstate$ and $\layerleftstate'$ belong to the same cell of $\mergingleft_{i}$. 
	Then $\layerleftstate$ and $\layerleftstate'$ are $\mergingright_i$ equivalent. In other words, for each symbol $\asymbol\in\alphabet$, there exists 
	$\layerrightstate\in\rightfrontier(\layer_{i})$ such that $\tuple{\layerleftstate,\asymbol,\layerrightstate}\in\layertransitions(\layer_{i})$ if and only if there exists 
	$\layerrightstate'\in\layerrightfrontier(\layer_{i})$ such that $\tuple{\layerleftstate',\asymbol,\layerrightstate'}\in\layertransitions(\layer_{i})$ and $\layerrightstate_i$ and 
	$\layerrightstate_i'$ belong to the same cell of $\mergingright_i$. Using the induction hypothesis, we have that for each pair $\layerrightstate,\layerrightstate'$ belonging
	to the same cell of $\mergingright_i$, $\lstateoddlang{\odd}{i+1}{\layerrightstate}=\lstateoddlang{\odd}{i+1}{\layerrightstate'}$. Therefore, 
	using \eqref{eq:layerODDLang}, that $\lstateoddlang{\odd}{i}{\layerleftstate}=\lstateoddlang{\odd}{i}{\layerleftstate'}$. 

	Now, in order to prove the converse, suppose for contradiction that $\layerleftstate$ and $\layerleftstate'$ do not belong to the same cell of $\mergingleft_i$ and
	 that $\lstateoddlang{\odd}{i}{\layerleftstate}=\lstateoddlang{\odd}{i}{\layerleftstate'}$. 
	Since $\layer_{i}$ is a deterministic, complete layer, for each symbol $\asymbol\in\alphabet$, there exists exactly one right state $\layerrightstate\in\layerrightfrontier(\layer_{i})$
	such that $\tuple{\layerleftstate,\asymbol,\layerrightstate}\in\layertransitions(\layer_{i})$.
	Similarly, for each symbol $\asymbol\in\alphabet$, there exists exactly one right state 
	$\layerrightstate'\in\layerrightfrontier(\layer_{i})$ such that $\tuple{\layerleftstate',\asymbol,\layerrightstate'}\in\layertransitions(\layer_{i})$. 
	Consequently, for some symbol $\asymbol\in\alphabet$, the right states $\layerrightstate$ and $\layerrightstate'$ associated with $\asymbol$, and $\layerleftstate$ and $\layerleftstate'$,
	respectively, belong to distinct cells of $\mergingright_{i}$. 
	Then, it follows from the induction hypothesis that $\lstateoddlang{\odd}{i+1}{\layerrightstate}\neq\lstateoddlang{\odd}{i+1}{\layerrightstate'}$. 
	Assume without loss of generality that $\lstateoddlang{\odd}{i+1}{\layerrightstate}\setminus\lstateoddlang{\odd}{i+1}{\layerrightstate'}\neq\emptyset$, 
	and let $\bstring\in\lstateoddlang{\odd}{i+1}{\layerrightstate}\setminus\lstateoddlang{\odd}{i+1}{\layerrightstate'}$. 
	Based on~\eqref{eq:layerODDLang}, we have that $\asymbol\bstring\in\lstateoddlang{\odd}{i}{\layerleftstate}$ but $\asymbol\bstring\not\in\lstateoddlang{\odd}{i}{\layerleftstate'}$.
	This implies that $\lstateoddlang{\odd}{i}{\layerleftstate}\neq\lstateoddlang{\odd}{i}{\layerleftstate'}$, contradicting our initial supposition. 
\end{proof}

Let $\alphabet$ be an alphabet and $\width\in \pN$. 
We denote by $\mergingalphabet{\alphabet}{\width}$ the set consisting of all triples $\tuple{\alayer,\mergingleft,\mergingright}$ such that $\alayer$ is a deterministic, complete layer in $\cdlayeralphabet{\alphabet}{\width}$,  and $\tuple{\mergingleft,\mergingright}$ is a merging annotation for $\alayer$. 
Additionally, we denote by $\collapsing{\alphabet}{\width}\colon\mergingalphabet{\alphabet}{\width}\rightarrow\cdlayeralphabet{\alphabet}{\width}$
the map that sends each triple $\tuple{\alayer,\mergingleft,\mergingright}\in\mergingalphabet{\alphabet}{\width}$ to the layer $\collapsing{\alphabet}{\width}(\alayer,\mergingleft,\mergingright)\in\cdlayeralphabet{\alphabet}{\width}$ obtained from $\alayer$ by identifying, for each $\sset\in \mergingleft\cup \mergingright$, all states belonging to $\sset$ with the smallest state that belongs to~$\sset$. 
More formally, for each triple $\tuple{\alayer,\mergingleft,\mergingright}\in\mergingalphabet{\alphabet}{\width}$,
we let $\collapsing{\alphabet}{\width}(\alayer,\mergingleft,\mergingright) = \blayer$, where $\blayer$ is the deterministic, complete layer belonging to  $\cdlayeralphabet{\alphabet}{\width}$ defined as follows:
\begin{itemize}
    \item $\layerleftfrontier(\blayer) \defeq \bigcup_{\sset \in \mergingleft}\set{\min\sset}$; $\layerrightfrontier(\blayer) \defeq \bigcup_{\sset' \in \mergingright}\set{\min\sset'}$; \vspace{1.0ex}
    \item $\layertransitions(\blayer) \defeq \bigcup_{\sset\in\mergingleft,\sset'\in\mergingright}\set{\tuple{\min\sset,\asymbol,\min\sset'} \setst \exists\,\layerleftstate\in\sset, \exists\,\layerrightstate\in\sset', \tuple{\layerleftstate,\asymbol,\layerrightstate}\in\layertransitions(\alayer)}$; \vspace{1.0ex} 
    \item $\layerinitialflag(\blayer) \defeq \layerinitialflag(\alayer)$; $\layerfinalflag(\blayer) \defeq \layerfinalflag(\alayer)$; \vspace{1.0ex}
    \item $\layerinitialstates(\blayer) \defeq \layerinitialstates(\alayer)$; $\layerfinalstates(\blayer) \defeq \layerrightfrontier(\blayer)  \cap \layerfinalstates(\alayer)$.
\end{itemize} 

Let $\oddcollapsingname{\alphabet}{\width}\colon\cdodddefiningsetstar{\alphabet}{\width}\rightarrow\cdodddefiningsetstar{\alphabet}{\width}$ be the map that for each $\oddlength\in\pN$, sends each deterministic, complete ODD $\odd = \layer_{1}\cdots\layer_{\oddlength}\in\cdodddefiningset{\alphabet}{\width}{\oddlength}$ to the deterministic, complete ODD $$\oddcollapsing{\alphabet}{\width}{\odd} \defeq \collapsing{\alphabet}{\width}(\layer_{1},\mergingleft_{1},\mergingright_{1})\cdots\collapsing{\alphabet}{\width}(\layer_{\oddlength},\mergingleft_{\oddlength},\mergingright_{\oddlength})\in\cdodddefiningset{\alphabet}{\width}{\oddlength}\text{,}$$ where $\sequence{\tuple{\mergingleft_{1},\mergingright_{1}},\ldots,\tuple{\mergingleft_{\oddlength},\mergingright_{\oddlength}}}$ denotes the unique merging annotation for $\odd$ (see Proposition~\ref{proposition:UniqueMerging}).

Let $\alphabet$ be an alphabet, $\width,\oddlength\in \pN$ and $\odd\in\cdodddefiningset{\alphabet}{\width}{\oddlength}$. 
We recall that since $\odd$ is a deterministic, complete ODD, we have that for each string $\astring=\asymbol_{1}\cdots\asymbol_{\oddlength}\in\alphabet^{\oddlength}$, there is a unique valid sequence $\sequence{\tuple{\layerleftstate_1,\asymbol_{1},\layerrightstate_1},\ldots,\tuple{\layerleftstate_{\oddlength},\asymbol_{\oddlength},\layerrightstate_{\oddlength}}}$ for $\astring$ in $\odd$. 
Thus, for each string $\astring\in\alphabet^{\oddlength}$ and each $i\in\bset{\oddlength}$, we let $\slayerrightstate{\odd}{\astring}{i}\defeq\layerrightstate_{i}$ denote the unique right state $\layerrightstate_{i}\in\layerrightfrontier(\layer_{i})$ that belongs to the valid sequence for $\astring$ in $\odd$. 
Moreover, we let $$\layerequivclass{\odd}{\astring}{i}\defeq\set{\astring'\in\alphabet^{\oddlength}\setst \slayerrightstate{\odd}{\astring'}{i}=\slayerrightstate{\odd}{\astring}{i}}\text{}$$ denote the equivalence class of $\astring$ with respect to $\odd$ and $i$.

\begin{proposition}\label{prop:EquivClassConcatenation}
	Let $\alphabet$ be an alphabet, $\width,\oddlength\in\pN$, $\odd=\layer_{1}\cdots\layer_{\oddlength}\in\cdodddefiningset{\alphabet}{\width}{\oddlength}$, and let $\layerrightstate$ be a right state in $\layerrightfrontier(\layer_{i})$ such that $\layerrightstate=\slayerrightstate{\odd}{\astring}{i}$ for some string $\astring\in\alphabet^{\oddlength}$ and some $i\in\bset{\oddlength-1}$. 
	For each string $\astring'=\asymbol_{1}'\cdots\asymbol_{\oddlength}'\in\layerequivclass{\odd}{\astring}{i}$ and each string $\bstring\in\alphabet^{\oddlength-i}$, we have that $\bstring\in\lstateoddlang{\odd}{i}{\layerrightstate}$ if and only if $\asymbol_{1}'\cdots\asymbol_{i}'\bstring\in\oddlang{\odd}$. 
\end{proposition}
\begin{proof}
	Let $\astring'=\asymbol_{1}'\cdots\asymbol_{\oddlength}'\in\layerequivclass{\odd}{\astring}{i}$ and $\bstring=\bsymbol_{i+1}\cdots\bsymbol_{\oddlength}\in\alphabet^{\oddlength-i}$. 
	Also, let 
	$$\sequence{\tuple{\layerleftstate_1,\asymbol_{1}',\layerrightstate_1},\ldots,\tuple{\layerleftstate_{\oddlength},\asymbol_{\oddlength}',\layerrightstate_{\oddlength}}}$$ 
	be the unique valid sequence for $\astring'$ in $\odd$. 
	We note that $\layerleftstate_1\in\layerinitialstates(\layer_{1})$ and $\layerleftstate_{i+1}=\layerrightstate$.  
	Suppose that $\bstring\in\lstateoddlang{\odd}{i}{\layerleftstate_{i+1}}$. 
	By definition, there is a sequence
$$\sequence{\tuple{\layerleftstate_{i+1}',\bsymbol_{i+1},\layerrightstate_{i+1}'},\ldots,\tuple{\layerleftstate_{\oddlength}',\bsymbol_{\oddlength},\layerrightstate_{\oddlength}'}}$$
of transitions such that $\layerleftstate_{i+1}'=\layerleftstate_{i+1}$, $\layerrightstate_{\oddlength}'\in\layerfinalstates(\layer_{\oddlength})$ and, for each $j\in\set{i+1,\ldots,\oddlength}$, $\tuple{\layerleftstate_{j}',\asymbol_{j}',\layerrightstate_{j}'}\in\layertransitions(\layer_{j})$. 
	Thus, $\sequence{\tuple{\layerleftstate_1,\asymbol_{1}',\layerrightstate_1},\ldots,\tuple{\layerleftstate_{i},\asymbol_{i}',\layerrightstate_{i}},\tuple{\layerleftstate_{i+1}',\bsymbol_{i+1},\layerrightstate_{i+1}'},\ldots,\tuple{\layerleftstate_{\oddlength}',\bsymbol_{\oddlength},\layerrightstate_{\oddlength}'}}$ is an accepting sequence for the string $\asymbol_{1}'\cdots\asymbol_{i}'\bstring$ in $\odd$, and therefore $\asymbol_{1}'\cdots\asymbol_{i}'\bstring\in\oddlang{\odd}$. 

	Conversely, suppose that $\asymbol_{1}'\cdots\asymbol_{i}'\bstring\in\oddlang{\odd}$.
	Then, there exists a unique accepting sequence $$\sequence{\tuple{\layerleftstate'_1,\asymbol_{1}',\layerrightstate'_1},\ldots,\tuple{\layerleftstate'_i,\asymbol_{i}',\layerrightstate'_i},\tuple{\layerleftstate'_{i+1},\bsymbol_{i+1},\layerrightstate'_{i+1}},\ldots,\tuple{\layerleftstate'_{\oddlength},\bsymbol_{\oddlength},\layerrightstate'_{\oddlength}}}$$ for $\asymbol_{1}'\cdots\asymbol_{i}'\bstring$ in $\odd$. 
	By the uniqueness of this sequence, we have that $\layerleftstate'_{1}=\layerleftstate_{1}$ and $\layerrightstate'_{j}=\layerrightstate_{j}$ for each $j \in \bset{i}$. 
	In particular, $\layerleftstate'_{i+1} = \layerrightstate'_{i} = \layerrightstate$. 
	Therefore, $\bstring\in\lstateoddlang{\odd}{i}{\layerrightstate}$. 
\end{proof}

\begin{proposition}\label{prop:EquivClass}
	Let $\alphabet$ be an alphabet, $\width,\oddlength\in\pN$, and let $\aodd$ and $\bodd$ be two deterministic, complete ODDs in $\cdodddefiningset{\alphabet}{\width}{\oddlength}$. If $\oddlang{\oddcollapsing{\alphabet}{\width}{\aodd}}=\oddlang{\oddcollapsing{\alphabet}{\width}{\bodd}}$, then $\layerequivclass{\oddcollapsing{\alphabet}{\width}{\aodd}}{\astring}{i}=\layerequivclass{\oddcollapsing{\alphabet}{\width}{\bodd}}{\astring}{i}$ for each $\astring\in\alphabet^{\oddlength}$ and each $i\in\bset{\oddlength}$. 
\end{proposition}
\begin{proof}
	For the sake of contradiction, suppose that $\oddlang{\oddcollapsing{\alphabet}{\width}{\aodd}}=\oddlang{\oddcollapsing{\alphabet}{\width}{\bodd}}$ but, for some string $\astring=\asymbol_{1}\cdots\asymbol_{\oddlength}\in\alphabet^{\oddlength}$ and some $i\in\bset{\oddlength}$, $\layerequivclass{\oddcollapsing{\alphabet}{\width}{\aodd}}{\astring}{i}\neq\layerequivclass{\oddcollapsing{\alphabet}{\width}{\bodd}}{\astring}{i}$. 

	Assume without loss of generality that $\layerequivclass{\oddcollapsing{\alphabet}{\width}{\aodd}}{\astring}{i}\setminus\layerequivclass{\oddcollapsing{\alphabet}{\width}{\bodd}}{\astring}{i}\neq\emptyset$.
	Then, let $\astring'=\asymbol_{1}'\cdots\asymbol_{\oddlength}'\in\layerequivclass{\oddcollapsing{\alphabet}{\width}{\aodd}}{\astring}{i}\setminus\layerequivclass{\oddcollapsing{\alphabet}{\width}{\bodd}}{\astring}{i}$.
	Consider $\layerleftstate_{i+1}'=\slayerrightstate{\oddcollapsing{\alphabet}{\width}{\bodd}}{\astring}{i}$ and $\layerleftstate_{i+1}''=\slayerrightstate{\oddcollapsing{\alphabet}{\width}{\bodd}}{\astring'}{i}$. 
	We note that $i < \oddlength$, otherwise $\oddlang{\oddcollapsing{\alphabet}{\width}{\aodd}}$ would be different from $\oddlang{\oddcollapsing{\alphabet}{\width}{\bodd}}$. 
	Moreover, since $\layerleftstate_{i+1}'\neq\layerleftstate_{i+1}''$, we obtain by Proposition~\ref{prop:mergingLanguagesFromLeftState} that $$\lstateoddlang{\oddcollapsing{\alphabet}{\width}{\bodd}}{i+1}{\layerleftstate_{i+1}'}\neq\lstateoddlang{\oddcollapsing{\alphabet}{\width}{\bodd}}{i+1}{\layerleftstate_{i+1}''}\text{.}$$ 
	Assume without loss of generality $\lstateoddlang{\oddcollapsing{\alphabet}{\width}{\bodd}}{i+1}{\layerleftstate_{i+1}'}\setminus\lstateoddlang{\oddcollapsing{\alphabet}{\width}{\bodd}}{i+1}{\layerleftstate_{i+1}''}\neq\emptyset$. 
	Let $\bstring\in\lstateoddlang{\oddcollapsing{\alphabet}{\width}{\bodd}}{i+1}{\layerleftstate_{i+1}'}\setminus\lstateoddlang{\oddcollapsing{\alphabet}{\width}{\bodd}}{i+1}{\layerleftstate_{i+1}''}$. 
	Since $\oddcollapsing{\alphabet}{\width}{\bodd}$ is deterministic, there exists a unique valid sequence for the string $\asymbol_{1}'\cdots\asymbol_{i}'\bstring$ in $\oddcollapsing{\alphabet}{\width}{\bodd}$, and by definition this sequence must contain the left state $\layerleftstate_{i+1}''$. 
	Consequently, it follows from Proposition~\ref{prop:EquivClassConcatenation} and from the fact that $\bstring$ is not accepted by $\oddcollapsing{\alphabet}{\width}{\bodd}$ from $\layerleftstate_{i+1}''$ that 
	\begin{equation}\label{eq:MergEquiv1}
		\asymbol_{1}'\cdots\asymbol_{i}'\bstring\not\in\oddlang{\oddcollapsing{\alphabet}{\width}{\bodd}}\text{.} 
	\end{equation} 
	On the other hand, $\bstring$ is accepted by $\oddcollapsing{\alphabet}{\width}{\bodd}$ from $\layerleftstate_{i+1}'$. 
	As a result, we obtain by Proposition~\ref{prop:EquivClassConcatenation} that $\asymbol_{1}\cdots\asymbol_{i}\bstring\in\oddlang{\oddcollapsing{\alphabet}{\width}{\bodd}}$. 
	In addition, we have that $\asymbol_{1}\cdots\asymbol_{i}\bstring\in\oddlang{\oddcollapsing{\alphabet}{\width}{\aodd}}$ since $\oddlang{\oddcollapsing{\alphabet}{\width}{\aodd}}=\oddlang{\oddcollapsing{\alphabet}{\width}{\bodd}}$. 
	This further implies that $\bstring\in\lstateoddlang{\oddcollapsing{\alphabet}{\width}{\aodd}}{i+1}{\layerleftstate_{i+1}}$, where $\layerleftstate_{i+1}$ denotes $\slayerrightstate{\oddcollapsing{\alphabet}{\width}{\aodd}}{\astring}{i}$. 
	However, since $\astring'\in\layerequivclass{\oddcollapsing{\alphabet}{\width}{\aodd}}{\astring}{i}$, it follows from Proposition~\ref{prop:EquivClassConcatenation} that 
	\begin{equation}\label{eq:MergEquiv2}
		\asymbol_{1}'\cdots\asymbol_{i}'\bstring\in\oddlang{\oddcollapsing{\alphabet}{\width}{\aodd}}\text{,} 
	\end{equation} 
	which, along with~\eqref{eq:MergEquiv1}, implies that $\oddlang{\oddcollapsing{\alphabet}{\width}{\aodd}}\neq\oddlang{\oddcollapsing{\alphabet}{\width}{\bodd}}$.
\end{proof}

\begin{proposition}\label{proposition:UniqueMergingODD}
	Let $\alphabet$ be an alphabet, $\width\in\pN$ and $\odd\in\cdodddefiningsetstar{\alphabet}{\width}$. 
	If $\odd$ is reachable, then $\oddcollapsing{\alphabet}{\width}{\odd}$ is a minimized ODD such that $\oddlang{\oddcollapsing{\alphabet}{\width}{\odd}} = \oddlang{\odd}$. 
\end{proposition}
\begin{proof}
	Assume that $\odd = \alayer_{1}\cdots\alayer_{\oddlength}$, for some $\oddlength \in \pN$, and let  
	$\sequence{\tuple{\mergingleft_{1},\mergingright_{1}}\cdots\tuple{\mergingleft_{\oddlength},\mergingright_{\oddlength}}}$ be the unique merging annotation for $\odd$. 
	First, we prove that $\oddlang{\oddcollapsing{\alphabet}{\width}{\odd}} = \oddlang{\odd}$. 
	Let $\astring = \asymbol_{1}\cdots\asymbol_{\oddlength} \in \alphabet^{\oddlength}$. 
	Suppose that $\astring \in \oddlang{\odd}$.
	Then, there exists an accepting sequence $\sequence{\tuple{\layerleftstate_{1},\asymbol_{1},\layerrightstate_{1}},\ldots,\tuple{\layerleftstate_{\oddlength},\asymbol_{\oddlength},\layerrightstate_{\oddlength}}}$ for $\astring$ in $\odd$. 
	For each $i \in \bset{\oddlength}$, let $\sset_{i}$ be the unique cell of $\mergingright_{i}$ that contains $\layerrightstate_{i}$. 
	Then, we have that  $$\sequence{\tuple{\layerleftstate_{1},\asymbol_{1},\min \sset_{1}},\tuple{\min \sset_{1},\asymbol_{\oddlength},\min \sset_{2}},\ldots,\tuple{\min \sset_{\oddlength-1},\asymbol_{\oddlength},\min \sset_{\oddlength}}}$$ is an accepting sequence for $\astring$ in $\oddcollapsing{\alphabet}{\width}{\odd}$.  
	As a result, we obtain that $\oddlang{\oddcollapsing{\alphabet}{\width}{\odd}} \subseteq \oddlang{\odd}$. 
	Now, suppose that $\astring \in \oddlang{\oddcollapsing{\alphabet}{\width}{\odd}}$. 
	Then, there exists an accepting sequence $\sequence{\tuple{\layerleftstate'_{1},\asymbol_{1},\layerrightstate'_{1}},\ldots,\tuple{\layerleftstate'_{\oddlength},\asymbol_{\oddlength},\layerrightstate'_{\oddlength}}}$ for $\astring$ in $\oddcollapsing{\alphabet}{\width}{\odd}$. 
	We note that for each $i \in \bset{\oddlength}$, there exists a right state $\layerrightstate_{i}\in\layerrightfrontier(\alayer_{i})$ such that $\layerrightstate_{i}$ and $\layerrightstate'_{i}$ belong to a same cell of $\mergingright_{i}$ and $\tuple{\layerleftstate_{i},\asymbol_{i},\layerrightstate_{i}}\in\layertransitions(\alayer_{i})$, where $\layerleftstate_{1} = \layerleftstate'_{1}$ and $\layerleftstate_{j}=\layerrightstate_{j-1}$ for each $j \in \set{2,\ldots,\oddlength}$. 
	Thus, there exists an accepting sequence $\sequence{\tuple{\layerleftstate_{1},\asymbol_{1},\layerrightstate_{1}},\ldots,\tuple{\layerleftstate_{\oddlength},\asymbol_{\oddlength},\layerrightstate_{\oddlength}}}$ for $\astring$ in $\odd$. 
	Therefore, $\oddlang{\oddcollapsing{\alphabet}{\width}{\odd}} \supseteq \oddlang{\odd}$. 

	Now, we prove that $\oddcollapsing{\alphabet}{\width}{\odd}$ is minimized if $\odd$ is reachable. 
	Thus, assume that $\odd$ is reachable. 
	This implies that $\oddcollapsing{\alphabet}{\width}{\odd}$ is also reachable and thus, for each $i \in \bset{\oddlength}$ and each $\layerrightstate \in \layerrightfrontier(\collapsing{\alphabet}{\width}(\alayer_{i},\mergingleft_{i},\mergingright_{i}))$, $\layerrightstate = \slayerrightstate{\oddcollapsing{\alphabet}{\width}{\odd}}{\astring}{i}$ for some $\astring \in \alphabet^{\oddlength}$. 
	Then, for each $i \in \bset{\oddlength}$, let $\width_{i} = \abs{\layerrightfrontier(\collapsing{\alphabet}{\width}(\alayer_{i},\mergingleft_{i},\mergingright_{i}))}$ and let $\astring_{1}^{i},\ldots,\astring_{\width_{i}}^{i}$ be strings such that $\layerequivclass{\oddcollapsing{\alphabet}{\width}{\odd}}{\astring_{j}^{i}}{i} \neq \layerequivclass{\oddcollapsing{\alphabet}{\width}{\odd}}{\astring_{j'}^{i}}{i}$ for each $j \in \bset{\width_{i}}$ and each $j' \in \bset{\width_{i}}$ with $j \neq j'$. 
	Also, let $\bodd=\blayer_{1}\cdots\blayer_{\oddlength}\in\cdodddefiningset{\alphabet}{\width}{\oddlength}$ be a minimized ODD such that $\oddlang{\bodd} = \oddlang{\odd}$. 
	We note that $\bodd$ is reachable.
	Thus, for each $i \in \bset{\oddlength}$ and each $\layerrightstate' \in \layerrightfrontier(\blayer_i)$, $\layerrightstate' = \slayerrightstate{\bodd}{\astring'}{i}$ for some $\astring' \in \alphabet^{\oddlength}$. 
	Moreover, we have that $\oddcollapsing{\alphabet}{\width}{\bodd}=\bodd$, otherwise $\bodd$ would not be minimized. 
	Then, for each $i \in \bset{\oddlength}$, we let $\isomorphism_{i} \colon \layerrightfrontier(\collapsing{\alphabet}{\width}(\alayer_{i},\mergingleft_{i},\mergingright_{i})) \rightarrow \layerrightfrontier(\blayer_{i})$ be the mapping such that for each $j \in \bset{\width_{i}}$, $\isomorphism_{i}(\slayerrightstate{\oddcollapsing{\alphabet}{\width}{\odd}}{\astring_{j}^{i}}{i}) = \slayerrightstate{\bodd}{\astring_{j}^{i}}{i}$. 
	It follows from Proposition~\ref{prop:EquivClass} that $\isomorphism_{i}$ is a bijection. 
	Consequently, we obtain that $\sequence{\isomorphism_{0},\ldots,\isomorphism_{\oddlength}}$ is a isomorphism between $\oddcollapsing{\alphabet}{\width}{\odd}$ and $\bodd$, where $\isomorphism_{0} \colon \layerleftfrontier(\collapsing{\alphabet}{\width}(\alayer_{1},\mergingleft_{1},\mergingright_{1})) \rightarrow \layerleftfrontier(\blayer_{1})$ is the trivial bijection that sends the unique left state in $\layerleftfrontier(\collapsing{\alphabet}{\width}(\alayer_{1},\mergingleft_{1},\mergingright_{1}))$ to the unique left state in $\layerleftfrontier(\blayer_{1})$. 
	Therefore, $\oddcollapsing{\alphabet}{\width}{\odd}$ is minimized. 
\end{proof}

For each alphabet $\alphabet$ and each positive integer $\width\in \pN$, we let
$\mergingrelation{\alphabet}{\width}\subseteq  \cdlayeralphabet{\alphabet}{\width} \times \mergingalphabet{\alphabet}{\width}$ 
and $\mergingcompatible{\alphabet}{\width}\subseteq  \mergingalphabet{\alphabet}{\width}\times \mergingalphabet{\alphabet}{\width}$ be the 
following relations. 
$$\mergingrelation{\alphabet}{\width} \defeq \set{\tuple{\alayer, \tuple{\alayer,\mergingleft,\mergingright}} \setst \tuple{\alayer,\mergingleft,\mergingright}\in\mergingalphabet{\alphabet}{\width}} \text{ and }
$$ 
$$
\begin{multlined}[t][0.95\textwidth]
\mergingcompatible{\alphabet}{\width} \defeq \{
\tuple{\tuple{\alayer,\mergingleft,\mergingright},\tuple{\alayer',\mergingleft',\mergingright'}} \setst \tuple{\alayer,\mergingleft,\mergingright}, \tuple{\alayer',\mergingleft',\mergingright'} \in \mergingalphabet{\alphabet}{\width}, \layerrightfrontier(\alayer) = \layerleftfrontier(\alayer'),\; \mergingright = \mergingleft'\}\text{.}
\end{multlined}
$$

For each alphabet $\alphabet$ and each positive integer $\width \in \pN$, we define the 
$(\cdlayeralphabet{\alphabet}{\width},\cdlayeralphabet{\alphabet}{\width})$-transduction
$\mergingtransduction{\alphabet}{\width}$ as 
$$\mergingtransduction{\alphabet}{\width} \defeq \multimaptransduction{\mergingrelation{\alphabet}{\width}}\circ\compTransduction{\mergingcompatible{\alphabet}{\width}}\circ \multimaptransduction{\collapsing{\alphabet}{\width}}.$$

The next lemma states that $\mergingtransduction{\alphabet}{\width}$ is a transduction that sends each deterministic, complete ODD $\aodd\in\cdodddefiningsetstar{\alphabet}{\width}$ to a 
\emph{minimized} deterministic, complete ODD $\bodd\in\cdodddefiningsetstar{\alphabet}{\width}$ that has the same language as $\aodd$.

\begin{lemma}[Merging Transduction]\label{lemma:MergingTransduction}
    For each alphabet $\alphabet$ and each positive integer $\width\in \pN$, the following statements hold.  
    \begin{enumerate}
	\item\label{MergingOne} $\mergingtransduction{\alphabet}{\width}$ is functional.
	\item\label{MergingTwo} $\domain(\mergingtransduction{\alphabet}{\width}) \supseteq \cdodddefiningsetstar{\alphabet}{\width}$.
	\item\label{MergingThree} For each pair $(\aodd,\bodd)\in \mergingtransduction{\alphabet}{\width}$, if $\aodd$ is a reachable ODD, 
		then $\bodd\in\cdodddefiningsetstar{\alphabet}{\width}$, $\oddlang{\bodd} = \oddlang{\aodd}$ and $\bodd$ is minimized.
	\item\label{MergingFour} $\mergingtransduction{\alphabet}{\width}$ is $2^{O(|\alphabet|\cdot \width\log \width)}$-regular. 
    \end{enumerate}
\end{lemma}
\begin{proof}
	We note that $\mergingtransduction{\alphabet}{\width}$ consists of all pairs $\tuple{\aodd,\bodd}$ of non-empty strings over the alphabet $\cdlayeralphabet{\alphabet}{\width}$ satisfying the conditions that $\abs{\aodd}=\abs{\bodd}$ and that, if $\aodd=\alayer_{1}\cdots\alayer_{\oddlength}$ and $\bodd=\blayer_{1}\cdots\blayer_{\oddlength}$ for some $\oddlength\in\pN$, then there exists a merging annotation $\tuple{\mergingleft_{i},\mergingright_{i}}$ for the layer $\alayer_{i}$ such that  $\blayer_{i}=\collapsing{\alphabet}{\width}(\alayer_{i},\mergingleft_{i},\mergingright_{i})$ for each $i\in\bset{\oddlength}$, and $\layerrightfrontier(\alayer_{j})=\layerleftfrontier(\alayer_{j+1})$ and $\mergingright_{j}=\mergingleft_{j+1}$ for each $j\in\bset{\oddlength-1}$. 
	Additionally, based on Proposition~\ref{proposition:UniqueMerging}, each $(\alphabet,\width)$-ODD admits a unique merging annotation. 
	As a result, we obtain that $\domain(\mergingtransduction{\alphabet}{\width}) \supseteq \cdodddefiningsetstar{\alphabet}{\width}$. 
	Moreover, if $\tuple{\aodd,\bodd}\in \mergingtransduction{\alphabet}{\width}$, then $\bodd = \oddcollapsing{\alphabet}{\width}{\aodd}$; 
thus, by the uniqueness of $\oddcollapsing{\alphabet}{\width}{\aodd}$, the transduction $\mergingtransduction{\alphabet}{\width}$ is functional. 
	Finally, it follows from Proposition~\ref{proposition:UniqueMergingODD} that for each pair
	$\tuple{\aodd,\bodd}\in \mergingtransduction{\alphabet}{\width}$ such that $\aodd$ is a reachable ODD in 
	$\cdodddefiningsetstar{\alphabet}{\width}$, we have that $\bodd = \oddcollapsing{\alphabet}{\width}{\aodd}$ is 
	a minimized ODD in $\cdodddefiningsetstar{\alphabet}{\width}$ that has the same language as $\aodd$.

	The fact that $\mergingtransduction{\alphabet}{\width}$ is $2^{O(|\alphabet|\cdot \width\cdot \log \width)}$-regular 
	follows from Proposition \ref{proposition:PropertiesTransductions}.(2) together with the fact that the multimap transductions 
	$\multimaptransduction{\mergingrelation{\alphabet}{\width}}$ and $\multimaptransduction{\collapsing{\alphabet}{\width}}$ are 
	$2$-regular (Proposition \ref{proposition:SizeTransductions}.(1)), and that the transduction 
	$\compTransduction{\mergingcompatible{\alphabet}{\width}}$ is $2^{O(|\alphabet|\cdot \width\cdot \log\width)}$-regular (Proposition \ref{proposition:SizeTransductions}.(2)),
	given that $\mergingcompatible{\alphabet}{\width} \subseteq \mergingalphabet{\alphabet}{\width}\times \mergingalphabet{\alphabet}{\width}$, and that 
	$|\mergingalphabet{\alphabet}{\width}|=2^{O(|\alphabet|\cdot \width\cdot \log \width)}$. 
\end{proof}

\subsection{Normalization Transduction.}
\label{subsection:NormalizationTransduction}

In this subsection, we define the {\em normalization transduction}, which intuitively simulates the process 
of numbering the states in each frontier of each layer of an ODD $\aodd$ according to their lexicographical order. 
This transduction can be defined as the composition of three elementary transductions. 
First, we use a multimap transduction to expand each layer of the ODD into 
a set of {\em annotated} layers. Each annotation relabels the left and right frontier 
vertices of the layer in such a way that the layer itself is normalized. 
Subsequently, we use a compatibility transduction that defines two consecutive annotated 
layers to be compatible if and only if the relabeling of the right-frontier of the first 
is equal to the relabeling of the left frontier of the second. It is possible to show that 
each reachable ODD $\aodd$ gives rise to a unique sequence of annotated layers where each two consecutive layers are compatible. 
Finally, we apply a mapping that sends each annotated layer to the layer obtained sending the numbers
in the frontiers to their relabeled versions. The resulting ODD is isomorphic to the original one, and therefore
besides preserving the language, it also preserves reachability and minimality.

Let $\alphabet$ be an alphabet, $\width \in \pN$, and let $\alayer\in \cdlayeralphabet{\alphabet}{\width}$.
For each two bijections $\isomorphism\colon\layerleftfrontier(\alayer)\rightarrow\dbset{\abs{\layerleftfrontier(\alayer)}}$ and $\isomorphism'\colon \layerrightfrontier(\alayer)\rightarrow\dbset{\abs{\layerrightfrontier(\alayer)}}$, we denote by $\permlayer{\alayer}{\isomorphism}{\isomorphism'}$ the $(\alphabet,\width)$-layer obtained from $\alayer$ by applying the bijection $\isomorphism$ to the left frontier of $\alayer$ and by applying the bijection $\isomorphism'$ to the right frontier of $\alayer$. 
More formally, $\permlayer{\alayer}{\isomorphism}{\isomorphism'}=\blayer$ is the $(\alphabet,\width)$-layer defined as follows:
\begin{itemize}
    \item $\layerleftfrontier(\blayer) \defeq \set{\isomorphism(\layerleftstate) \setst \layerleftstate \in \layerleftfrontier(\alayer)}$; $\layerrightfrontier(\blayer) \defeq \set{\isomorphism'(\layerrightstate) \setst \layerrightstate \in \layerrightfrontier(\alayer)}$; 
    \vspace{1.0ex}
    \item $\layerinitialstates(\blayer) \defeq \set{\isomorphism(\layerleftstate) \setst \layerleftstate \in \layerinitialstates(\alayer)}$;  $\layerfinalstates(\blayer) \defeq   \set{\isomorphism'(\layerrightstate) \setst \layerrightstate \in \layerfinalstates(\alayer)}$; 
    \vspace{1.0ex}
    \item $\layertransitions(\blayer) \defeq \set{\tuple{\isomorphism(\layerleftstate), \asymbol, \isomorphism'(\layerrightstate)} \setst \tuple{\layerleftstate, \asymbol, \layerrightstate} \in \layertransitions(\alayer)}$;
    \vspace{1.0ex}
    \item $\layerinitialflag(\blayer) \defeq \layerinitialflag(\alayer)$; $\layerfinalflag(\blayer) \defeq \layerfinalflag(\alayer)$. 
\end{itemize}

We note that since $\alayer\in \cdlayeralphabet{\alphabet}{\width}$, $\permlayer{\alayer}{\isomorphism}{\isomorphism'}$ also belongs to $\cdlayeralphabet{\alphabet}{\width}$. 

Let $\alphabet$ be an alphabet, $\width \in \pN$. 
A \emph{normalizing isomorphism} for a reachable layer $\layer\in \cdlayeralphabet{\alphabet}{\width}$ 
is a pair $\tuple{\isomorphism,\isomorphism'}$ of bijections $\isomorphism\colon\layerleftfrontier(\alayer)\rightarrow\dbset{\abs{\layerleftfrontier(\alayer)}}$
and $\isomorphism'\colon \layerrightfrontier(\alayer)\rightarrow\dbset{\abs{\layerrightfrontier(\alayer)}}$ such that the layer $\permlayer{\alayer}{\isomorphism}{\isomorphism'}$ is normalized. 
Let $\oddlength\in\pN$ and $\odd = \layer_{1}\cdots\layer_{\oddlength}$ be a reachable ODD in $\cdodddefiningset{\alphabet}{\width}{\oddlength}$.
A \emph{normalizing isomorphism} for $\odd$ is a sequence $\isomorphismsequence = \sequence{\isomorphism_{0},\isomorphism_{1},\ldots,\isomorphism_{\oddlength}}$
such that for each $i\in\bset{\oddlength}$, $\tuple{\isomorphism_{i-1},\isomorphism_{i}}$ is a normalizing isomorphism for $\layer_{i}$.

\begin{proposition}\label{proposition:UniqueIsomorphism}
Let $\alphabet$ be an alphabet and $\width$. 
Every reachable ODD in $\cdodddefiningsetstar{\alphabet}{\width}$ admits a unique normalizing isomorphism. 
\end{proposition}
\begin{proof}
First, we claim that for each reachable, layer $\alayer\in\cdlayeralphabet{\alphabet}{\width}$ and each bijection $\isomorphism\colon\layerleftfrontier(\alayer)\rightarrow\dbset{\abs{\layerleftfrontier(\alayer)}}$, there exists a unique bijection $\isomorphism'\colon\layerrightfrontier(\alayer)\rightarrow\dbset{\abs{\layerrightfrontier(\alayer)}}$ such that $\permlayer{\alayer}{\isomorphism}{\isomorphism
'}$ is normalized. 
Indeed, consider $\blayer=\permlayer{\alayer}{\permutation}{\permutation''}$, where $\permutation''\colon\layerrightfrontier(\alayer)\rightarrow\dbset{\abs{\layerrightfrontier(\alayer)}}$ denotes the identity function. 
Then, let $\permutation'\colon\layerrightfrontier(\alayer)\rightarrow\dbset{\abs{\layerrightfrontier(\alayer)}}$ be a bijection such that for each two right states  $\layerrightstate,\layerrightstate'\in\layerrightfrontier(\alayer)$, we have that $\permutation'(\layerrightstate)\leq\permutation'(\layerrightstate')$ if and only if $\porderfrontier{\blayer}{\layerrightstate} \leq \porderfrontier{\blayer}{\layerrightstate'}$. 
One can readily verify that $\permlayer{\alayer}{\isomorphism}{\isomorphism
'}$ is normalized. 
Furthermore, since $\blayer$ is deterministic, $\porderfrontierName{\blayer}$ is an injection from $\layerrightfrontier(\blayer)$ to $\layerleftfrontier(\blayer)\cartesianproduct\alphabet$, \ie, for each two distinct right states $\layerrightstate,\layerrightstate'\in\layerrightfrontier(\blayer)$, we have that either $\porderfrontier{\blayer}{\layerrightstate}<\porderfrontier{\blayer}{\layerrightstate'}$ or $\porderfrontier{\blayer}{\layerrightstate'}<\porderfrontier{\blayer}{\layerrightstate}$. 
In other words, $\porderfrontierName{\blayer}$ describes a total order on $\layerrightfrontier(\blayer)$. 
Therefore, $\isomorphism'$ is the unique bijection from $\layerrightfrontier(\alayer)$ to $\dbset{\abs{\layerrightfrontier(\alayer)}}$ such that $\permlayer{\alayer}{\isomorphism}{\isomorphism
'}$ is normalized. 

Let $\oddlength\in\pN$ and $\aodd=\alayer_{1}\cdots\alayer_{\oddlength}\in\cdodddefiningset{\alphabet}{\width}{\oddlength}$ be an ODD such that 
$\alayer_{i}$ is a reachable layer for each $i\in\bset{\oddlength}$, $\layerleftfrontier(\alayer_{i+1})=\layerrightfrontier(\alayer_{i})$ for each 
$i\in\bset{\oddlength-1}$, $\layerinitialflag(\alayer_{1})=1$ and $\layerinitialflag(\alayer_{i})=0$ for each $i\in\set{2,\ldots,\oddlength}$. 
Based on the previous claim, we prove by induction on $\oddlength$ that the following statement holds: there exists a unique sequence $\sequence{\isomorphism_{0},\isomorphism_{1},\ldots,\isomorphism_{\oddlength}}$ such that (1) $\isomorphism_{0}\colon\layerleftfrontier(\alayer_{1})\rightarrow\dbset{\abs{\layerleftfrontier(\alayer_{1})}}$ is a bijection, (2) $\isomorphism_{i}\colon\layerrightfrontier(\alayer_{i})\rightarrow\dbset{\abs{\layerrightfrontier(\alayer_{i})}}$ is a bijection for each $i\in \bset{\oddlength}$, and (3) $\permlayer{\alayer_{i}}{\isomorphism_{i-1}}{\isomorphism_{i}}$ is a normalized layer for each $i\in \bset{\oddlength}$.

\emph{Base case.} Consider $\oddlength=1$. 
Since $\alayer_{1}$ is deterministic, $\abs{\layerleftfrontier(\alayer_{1})} = 1$. 
Thus, the bijection $\isomorphism_{0}\colon\layerleftfrontier(\alayer_{1})\rightarrow\dbset{\abs{\layerleftfrontier(\alayer_{1})}}$ is trivially uniquely determined. 
As a result, there exists a unique sequence $\sequence{\isomorphism_{0},\isomorphism_{1}}$ satisfying the required conditions (1)--(3). 

\emph{Inductive step.} Consider $\oddlength>1$. 
 Let $\bodd=\alayer_{1}\cdots\alayer_{\oddlength-1}$ be the string obtained from $\aodd=\alayer_{1}\cdots\alayer_{\oddlength}$ by removing the layer $\alayer_{\oddlength}$. 
 It follows from the inductive hypothesis that there exists a unique sequence $\sequence{\isomorphism_{0},\isomorphism_{1},\ldots,\isomorphism_{\oddlength-1}}$ such that $\isomorphism_{0}\colon\layerleftfrontier(\alayer_{1})\rightarrow\dbset{\abs{\layerleftfrontier(\alayer_{1})}}$ is a bijection, $\isomorphism_{i}\colon\layerrightfrontier(\alayer_{i})\rightarrow\dbset{\abs{\layerrightfrontier(\alayer_{i})}}$ is a bijection for each $i\in \bset{\oddlength-1}$, and $\permlayer{\alayer_{i}}{\isomorphism_{i-1}}{\isomorphism_{i}}$ is a normalized layer for each $i\in \bset{\oddlength-1}$. 
 In particular, we note that the bijection $\isomorphism_{\oddlength-1}$ is uniquely determined. 
 Furthermore, based on the previous claim, there exists a unique bijection $\isomorphism_{\oddlength}\colon\layerrightfrontier(\alayer_{\oddlength})\rightarrow\dbset{\abs{\layerrightfrontier(\alayer_{\oddlength})}}$ such that $\permlayer{\alayer_{\oddlength}}{\isomorphism_{\oddlength-1}}{\isomorphism_{\oddlength}}$ is normalized. 
 Therefore, there exists a unique sequence $\sequence{\isomorphism_{0},\isomorphism_{1},\ldots,\isomorphism_{\oddlength}}$ satisfying the required conditions (1)--(3). 
\end{proof}

\begin{proposition}\label{prop:UniqueNormalizedODD}
	Let $\alphabet$ be an alphabet, $\width,\oddlength\in\pN$ and $\aodd\in\odddefiningset{\alphabet}{\width}{\oddlength}$.
	If $\aodd$ is a reachable, deterministic ODD and $\isomorphismsequence = \sequence{\isomorphism_{0},\isomorphism_{1},\ldots,\isomorphism_{\oddlength}}$ is the unique normalizing isomorphism for $\aodd$, then $\bodd = \permlayer{\alayer_{1}}{\isomorphism_{0}}{\isomorphism_{1}}\cdots\permlayer{\alayer_{\oddlength}}{\isomorphism_{\oddlength-1}}{\isomorphism_{\oddlength}}$ is a normalized ODD such that $\oddlang{\bodd} = \oddlang{\aodd}$. 
\end{proposition}
\begin{proof}
	It immediately follows from the definition of normalizing isomorphism that $\bodd$ is normalized. 
	Finally, we note that $\isomorphismsequence$ is an isomorphism from $\aodd$ to $\bodd$. 
	Therefore, by Proposition~\ref{proposition:IsomorphicEquivalenceLanguages}, $\oddlang{\bodd} = \oddlang{\aodd}$. 
\end{proof}

For each finite set $\xset$, we denote by $\symmetricgroup{\xset} \defeq \set{\isomorphism\colon\xset\rightarrow\dbset{\abs{\xset}} \setst \isomorphism \text{ is a bijection}}$ the set
of all bijections from $\xset$ to $\dbset{\abs{\xset}}$. 
For each alphabet $\alphabet$ and each $\width \in \pN$ we define the following set. 
$$
\layerssymmericgroup{\alphabet}{\width} = \set{\tuple{\isomorphism,\alayer,\isomorphism'} \setst \alayer\in \cdlayeralphabet{\alphabet}{\width}, 
\isomorphism\in\symmetricgroup{\layerleftfrontier(\alayer)},\; 
\isomorphism'\in \symmetricgroup{\layerrightfrontier(\alayer)}, \permlayer{\alayer}{\isomorphism}{\isomorphism'} \text{ is normalized}}. 
$$
We let $\relabelingmap{\alphabet}{\width}:\layerssymmericgroup{\alphabet}{\width} \rightarrow \cdlayeralphabet{\alphabet}{\width}$ be the map that 
sends each triple $\tuple{\isomorphism,\alayer,\isomorphism'} \in \layerssymmericgroup{\alpahbet}{\width}$ to the layer $\permlayer{\alayer}{\isomorphism}{\isomorphism'}$. 
Moreover, we let $\normalizingrelation{\alphabet}{\width}\subseteq \cdlayeralphabet{\alphabet}{\width}\times \layerssymmericgroup{\alphabet}{\width}$
and $\normalizingcompatibility{\alphabet}{\width}\subseteq \layerssymmericgroup{\alphabet}{\width}\times  \layerssymmericgroup{\alphabet}{\width}$ 
be the following relations. 
$$\normalizingrelation{\alphabet}{\width} \defeq  \{\tuple{\alayer,\tuple{\isomorphism,\alayer,\isomorphism'}} \setst
\tuple{\isomorphism,\alayer,\isomorphism'}\in \layerssymmericgroup{\alphabet}{\width}\}.
$$ 
$$
\normalizingcompatibility{\alphabet}{\width} \defeq \{ \tuple{\tuple{\isomorphism,\alayer,\isomorphism'},\tuple{\isomorphism',\blayer,\isomorphism''}}\in 
\layerssymmericgroup{\alphabet}{\width}\times  \layerssymmericgroup{\alphabet}{\width} \setst
\layerrightfrontier(\alayer) = \layerleftfrontier(\blayer),  
\}\text{.}
$$

Finally, for each alphabet $\alphabet$, and each positive integer $\width \in \pN$, 
we let $\normalizationtransduction{\alphabet}{\width}$ be the 
$(\cdlayeralphabet{\alphabet}{\width},\cdlayeralphabet{\alphabet}{\width})$-transduction
$$\normalizationtransduction{\alphabet}{\width} \defeq
\multimaptransduction{\normalizingrelation{\alphabet}{\width}} \circ \compTransduction{\normalizingcompatibility{\alphabet}{\width}}\circ \multimaptransduction{\relabelingmap{\alphabet}{\width}}.$$

The next lemma states that $\normalizationtransductionname$ is a 
transduction that sends each reachable, deterministic, complete ODD $\aodd\in\cdodddefiningsetstar{\alphabet}{\width}$ to as \emph{normalized}, deterministic, complete ODD $\bodd\in\cdodddefiningsetstar{\alphabet}{\width}$ that has the same language as $\aodd$.

\begin{lemma}[Normalization Transduction] \label{lemma:NormalizationTransduction}
For each alphabet $\alphabet$ and each positive integer $\width\in\pN$, the following statements hold. 
\begin{enumerate}
\item \label{NormalizationOne} $\normalizationtransduction{\alphabet}{\width}$ is functional.
\item \label{NormalizationTwo} $\domain(\normalizationtransduction{\alphabet}{\width}) \supseteq \set{\odd \in \cdodddefiningsetstar{\alphabet}{\width} \setst \odd \text{ is reachable}}$. 
\item \label{NormalizationThree} For each pair $(\aodd,\bodd)\in \normalizationtransduction{\alphabet}{\width}$, if $\aodd$ is reachable 
then $\bodd\in\cdodddefiningsetstar{\alphabet}{\width}$, $\oddlang{\bodd} = \oddlang{\aodd}$ and $\bodd$ is normalized.  
\end{enumerate}
\end{lemma}
\begin{proof}
We note that $\normalizationtransduction{\alphabet}{\width}$ consists of all pairs $\tuple{\aodd,\bodd}$ of non-empty strings over the alphabet $\cdlayeralphabet{\alphabet}{\width}$ satisfying the conditions 
that $\abs{\aodd}=\abs{\bodd}$ and that, if $\aodd=\alayer_{1}\cdots\alayer_{\oddlength}$ and 
$\bodd=\blayer_{1}\cdots\blayer_{\oddlength}$ for some $\oddlength\in\pN$, then there exists a sequence 
of permutations $\isomorphismsequence = \sequence{\isomorphism_{0},\isomorphism_{1},\ldots,\isomorphism_{\oddlength}}$ such that 
for each $i\in [\oddlength]$, $(\isomorphism_{i-1},\alayer_i,\isomorphism_{i})\in \layerssymmericgroup{\alphabet}{\width}$ 
and $\alayer_{i}' = \permlayer{\alayer_{i}}{\isomorphism_{i-1}}{\isomorphism_i}$.  
Additionally, based on Proposition~\ref{proposition:UniqueIsomorphism}, each reachable, deterministic $(\alphabet,\width)$-ODD admits a unique normalizing isomorphism. 
As a result, we obtain that $\domain(\normalizationtransduction{\alphabet}{\width}) \supseteq \set{\odd \in \cdodddefiningsetstar{\alphabet}{\width} \setst \odd \text{ is reachable}}$. 
Moreover, if $\tuple{\aodd,\bodd}\in \duplicate{\cdodddefiningsetstar{\alphabet}{\width}}\circ\normalizationtransduction{\alphabet}{\width}$, then $\bodd = \permlayer{\alayer_{1}}{\isomorphism_{0}}{\isomorphism_{1}}\cdots\permlayer{\alayer_{\oddlength}}{\isomorphism_{\oddlength-1}}{\isomorphism_{\oddlength}}$, where $\isomorphismsequence = \sequence{\isomorphism_{0},\isomorphism_{1},\ldots,\isomorphism_{\oddlength}}$ denotes the unique normalizing isomorphism of $\aodd$; thus, by the uniqueness of $\isomorphismsequence$, the transduction $\duplicate{\cdodddefiningsetstar{\alphabet}{\width}}\circ\normalizationtransduction{\alphabet}{\width}$ is functional. 
Finally, it follows from Proposition~\ref{prop:UniqueNormalizedODD} that for each pair $\tuple{\aodd,\bodd}\in \normalizationtransduction{\alphabet}{\width}$ such that $\aodd$ is a reachable ODD in $\cdodddefiningsetstar{\alphabet}{\width}$, we have that $\bodd$ is a normalized ODD in $\cdodddefiningsetstar{\alphabet}{\width}$ that has the same language as $\aodd$. 

The fact that $\normalizationtransduction{\alphabet}{\width}$ is $2^{O(|\alphabet|\cdot \width\cdot \log \width)}$-regular
follows from Proposition \ref{proposition:PropertiesTransductions}.(2) together with the fact that the multimap transductions 
$\multimaptransduction{\normalizingrelation{\alphabet}{\width}}$ and $\multimaptransduction{\relabelingmap{\alphabet}{\width}}$ are
$2$-regular (Proposition \ref{proposition:SizeTransductions}.(1)), and that the transduction 
$\compTransduction{\normalizingcompatibility{\alphabet}{\width}}$ is $2^{O(|\alphabet|\cdot \width\cdot \log\width)}$-regular (Proposition \ref{proposition:SizeTransductions}.(2)),
given that $\normalizingcompatibility{\alphabet}{\width} \subseteq \layerssymmericgroup{\alphabet}{\width}\times \layerssymmericgroup{\alphabet}{\width}$, and that 
$|\layerssymmericgroup{\alphabet}{\width}|=2^{O(|\alphabet|\cdot \width\cdot \log \width)}$. 
\end{proof}

\subsection{Putting All Steps Together}\label{sec:proofMainTheorem}
\label{subsection:FinalProof}

In this subsection we combine Observation \ref{observation:CopyRegular} with Lemma \ref{lemma:DeterminizationTransduction}, Lemma \ref{lemma:ReachabilityTransduction}, 
Lemma \ref{lemma:MergingTransduction} and Lemma \ref{lemma:NormalizationTransduction} to prove our Canonization as Transduction Theorem (Theorem \ref{theorem:CanonizationTransductionTheorem}).
Consider the transduction 
$$\cdcanonizationtransduction{\alphabet}{\width}
\defeq \duplicate{\cdodddefiningsetstar{\alphabet}{\width}} \circ \reachabilitytransduction{\alphabet}{\width}\circ\mergingtransduction{\alphabet}{\width}\circ\normalizationtransduction{\alphabet}{\width}.$$  
Since each of the four transductions in the composition is at most $2^{O(|\alphabet|\cdot \width \cdot \log \width)}$-regular, we have that 
$\cdcanonizationtransduction{\alphabet}{\width}$ is $2^{O(|\alphabet|\cdot \width\cdot \log \width)}$-regular. Since each these four
transductions is functional, the transduction $\cdcanonizationtransduction{\alphabet}{\width}$
is functional. Since $\domain(\duplicate{\cdodddefiningsetstar{\alphabet}{\width}}) = \cdodddefiningsetstar{\alphabet}{\width}$ and the image of each of the three first 
transductions is contained in the domain of the next transduction (from left to right), we have that $\domain(\cdcanonizationtransduction{\alphabet}{\width}) = \cdodddefiningsetstar{\alphabet}{\width}$.  
Now, let $(\aodd,\aodd')$ be a pair of ODDs in $\cdcanonizationtransduction{\alphabet}{\width}$. Then there exist ODDs $\aodd_1$ and $\aodd_2$ such that
$(\aodd,\aodd_1)\in \reachabilitytransduction{\alphabet}{\width}$, $(\aodd_1,\aodd_2)\in \mergingtransduction{\alphabet}{\width}$,
and $(\aodd_2,\aodd')\in \normalizationtransduction{\alphabet}{\width}$. Since each of these transductions is language preserving, we have that 
$\oddlang{\aodd}=\oddlang{\aodd_1}=\oddlang{\aodd_2}=\oddlang{\aodd'}$. 
Since $\aodd\in \cdodddefiningsetstar{\alphabet}{\width}$, we have that $\aodd$ is by definition deterministic and complete. 
By Lemma \ref{lemma:ReachabilityTransduction}, $\aodd_1$ is deterministic, complete and reachable. 
By Lemma \ref{lemma:MergingTransduction}, $\aodd_2$ is deterministic, complete and minimized. Finally, by Lemma \ref{lemma:NormalizationTransduction}, 
$\aodd'$ is deterministic, complete, minimized and normalized. Since for each ODD $\aodd$, there is a unique deterministic, complete, minimized and normalized
ODD $\canonizationfunction(\aodd)$ with the same language as $\aodd$, we have that $\aodd'=\canonizationfunction(\aodd)$. This shows that 
$\cdcanonizationtransduction{\alphabet}{\width} = \{(\aodd,\canonizationfunction(\aodd))\;:\; \aodd\in \odddefiningsetstar{\alphabet}{\width}\}$.

Now, consider the transduction $$\canonizationtransduction{\alphabet}{\width} \defeq 
\duplicate{\odddefiningsetstar{\alphabet}{\width}} \circ \determinizationtransduction{\alphabet}{\width}\circ \cdcanonizationtransduction{\alphabet}{2^{\width}}.$$
Since $\cdcanonizationtransduction{\alphabet}{\width}$ is $2^{O(|\alphabet|\cdot \width \cdot \log \width)}$-regular, we have that
$\cdcanonizationtransduction{\alphabet}{2^{\width}}$ is $2^{O(|\alphabet|\cdot \width \cdot 2^{\width})}$-regular. 
This implies that $\canonizationtransduction{\alphabet}{\width}$ is also $2^{O(|\alphabet|\cdot \width \cdot 2^{\width})}$-regular.  
Since $\odddefiningsetstar{\alphabet}{\width}= \domain(\determinizationtransduction{\alphabet}{\width})$ and
$\image(\determinizationtransduction{\alphabet}{\width})$ is included in $\domain(\cdcanonizationtransduction{\alphabet}{2^{\width}})$,
we have that $\domain(\canonizationtransduction{\alphabet}{\width}) = \odddefiningsetstar{\alphabet}{\width}$. Now, let $(\aodd,\aodd')$ be 
a pair of ODDs in $\canonizationtransduction{\alphabet}{\width}$. Then there is an ODD $\aodd_1\in \cdodddefiningsetstar{\alphabet}{2^{\width}}$
such that $(\aodd,\aodd_1)\in \determinizationtransduction{\alphabet}{\width}$ and $(\aodd_1,\aodd')\in \cdcanonizationtransduction{\alphabet}{2^\width}$. 
By Lemma \ref{lemma:DeterminizationTransduction}, we have that $\aodd_1$ is complete, deterministic and $\oddlang{\aodd}=\oddlang{\aodd_1}$. 
Additionally, $\aodd' = \canonizationfunction(\aodd_1)$. Since $\oddlang{\aodd} = \oddlang{\aodd_1} = \oddlang{\canonizationfunction(\aodd_1)} = \oddlang{\aodd'}$, we have that 
$\aodd' = \canonizationfunction(\aodd_1) = \canonizationfunction(\aodd)$.
This shows that $\canonizationtransduction{\alphabet}{\width} = \{(\aodd,\canonizationfunction(\aodd))\;:\; \aodd\in \odddefiningsetstar{\alphabet}{\width}\}$. $\square$

\section{Conclusion}
\label{section:Conclusion}
In this work, we have introduced the notion of second-order finite automata, 
a formalism that combines traditional finite automata with ODDs of bounded
width in order to represent possibly infinite {\em classes of languages}.
Our main result (Theorem \ref{theorem:CanonizationSecondOrder})
is a {\em canonical form of canonical forms theorem}. It states for each second-order
finite automaton $\finiteautomaton$, one can construct a canonical form $\canonizationfunction_2(\finiteautomaton)$
whose language $\automatonlang{\canonizationfunction_2(\finiteautomaton)} = \{\canonizationfunction(\aodd)\;:\;\aodd\in \automatonlang{\finiteautomaton}\}$ is precisely the set of canonical forms of ODDs in $\automatonlang{\finiteautomaton}$. 
Here, the canonical form $\canonizationfunction(\aodd)$ of an ODD $\aodd$ is the usual
deterministic, complete, normalized ODD with {\em minimum number of states} having the 
same language as $\aodd$. In this sense, the ODDs in $\automatonlang{\canonizationfunction_2(\finiteautomaton)}$
carry useful complexity theoretic information about the languages they represent in the class 
$\automatonlanghigher{\finiteautomaton}{2} = \automatonlanghigher{\canonizationfunction_2(\finiteautomaton)}{2}$. 

Our canonization result immediately implies that the collection of regular-decisional classes of 
languages is closed under union, intersection, set difference, and a suitable notion of bounded-width
complementation. This result also implies that inclusion and non-emptiness of intersection 
for regular-decisional classes of languages are decidable. Furthermore, non-emptiness of intersection
for the second languages of second-order finite automata $\finiteautomaton_1$ and $\finiteautomaton_2$
can be solved in fixed parameter tractable time when the parameter is the maximum width of 
an ODD accepted by $\finiteautomaton_1$ or $\finiteautomaton_2$. 

We also provided two algorithmic applications of second-order automata to the theory 
of ODDs. First, we have shown that several width/size minimization problems for ODDs can 
be solved in fixed-parameter tractable time when parameterized by the width of 
the input ODD. This implies corresponding FPT algorithms for width/size
minimization of ordered binary decision diagrams (OBDDs) with a fixed ordering. 
Previous to our work, only exponential algorithms were known. 
Finally, we have shown that second-order finite automata can be used to
count the exact number of distinct functions computable by $(\alphabet,\width)$-ODDs of a
given width $\width$ and a given length $\oddlength$ in time 
$2^{O(|\alphabet|\cdot \width \cdot 2^{\width})}\cdot \oddlength^{O(1)}$,
and in time $2^{O(|\alphabet|\cdot \width \cdot \log \width)}\cdot \oddlength^{O(1)}$ if 
only deterministic, complete ODDs are considered. It is worth noting that the 
naive process of enumerating functions while eliminating repetitions takes time 
(and space) exponential in both $\width$ and in $\oddlength$. 

\paragraph{Regular Canonizing Relations.}
Most results in this work are obtained as a consequence of Theorem \ref{theorem:CanonizationTransductionTheorem},
which states that the relation $\canonizationtransduction{\alphabet}{\width} = \{(\aodd,\canonizationfunction(\aodd))\;:\; 
\aodd\in \odddefiningsetstar{\alphabet}{\width}\}$ is a regular relation. It is worth noting that aside from 
complexity theoretic considerations, Theorem \ref{theorem:CanonizationSecondOrder} and Theorem \ref{theorem:ClosureProperties}
have identical proofs if we replace $\canonizationtransduction{\alphabet}{\width}$ with any {\em regular} canonizing relation 
$R(\alphabet,\width)$ for $\odddefiningsetstar{\alphabet}{\width}$ in the sense we will define below. Nevertheless, 
when taking complexity considerations into account, and also when considering our applications in Section \ref{section:Applications}, 
the fact that the transductions $\canonizationtransduction{\alphabet}{\width}$ and 
$\cdcanonizationtransduction{\alphabet}{\width}$ are $2^{O(|\alphabet|\cdot \width \cdot 2^{\width})}$-regular and 
$2^{O(|\alphabet|\cdot \width \cdot \log \width)}$-regular respectively play an important role. Additionally, some 
of our results use explicitly the fact the canonical form $\canonizationfunction(\aodd)$ has minimum number
of states among all deterministic, complete ODDs with the same language as $\aodd$. 

Say that a relation $R(\alphabet,\width) \subseteq \odddefiningsetstar{\alphabet}{\width}\times \odddefiningsetstar{\alphabet}{\width}$
is canonizing for $\odddefiningsetstar{\alphabet}{\width}$ if the following three conditions are verified. 
\begin{enumerate}
\item $R(\alphabet,\width)$ is functional and the domain of $R(\alphabet,\width)$ is equal to $\odddefiningsetstar{\alphabet}{\width}$.
\item For each $(\aodd,\aodd')\in R(\alphabet,\width)$, $\oddlang{\aodd} = \oddlang{\aodd'}$. 
\item $(\aodd,\aodd')\in R(\alphabet,\width)$ implies that $(\aodd'',\aodd')\in R(\alphabet,\width)$ for each $\aodd''\in \odddefiningsetstar{\alphabet}{\width}$ with $\oddlang{\aodd}=\oddlang{\aodd''}$. 
\end{enumerate}

The notion of a relation $\widehat{R}(\alphabet,\width)$ that is canonizing for $\cdodddefiningsetstar{\alphabet}{\width}$ can be defined analogously. An interesting question is whether there are canonizing relations with significantly better complexity 
than the ones of $\canonizationtransduction{\alphabet}{\width}$ and $\cdcanonizationtransduction{\alphabet}{\width}$. 
More specifically, is there some canonizing relation for $\odddefiningsetstar{\alphabet}{\width}$ that is
$\alpha$-regular for $\alpha = 2^{f(|\alphabet|)\cdot 2^{o(\width)}}$ where $f$ is a function depending only on the size of the alphabet? 
Similarly, is there some canonizing relation for $\cdodddefiningsetstar{\alphabet}{\width}$ that is $\alpha$-regular for some
$\alpha = 2^{o(f(|\alphabet|)\cdot \width \cdot \log \width)}$? In view of Observation \ref{observation:FPTIntersection},
a canonizing relation of complexity $\alpha = 2^{O(f(|\alphabet|)\cdot \width)}$ would imply that emptiness of intersection 
regular-decisional classes of languages can be realized in polynomial time even when $\width$ is 
logarithmic in the size of the input second-order finite automata representing these classes of languages. 

\paragraph{Connections with the Theory of Automatic Structures.}
Finite automata operating with ODDs and tuples of ODDs were 
first considered in \cite{deOliveiraOliveira2020symbolic} as a formalism to
provide a uniform representation of classes of finite relational 
structures of bounded ODD-width. The technical results from \cite{deOliveiraOliveira2019width}
rely on two observations. First, that the relation 
$\mathcal{R}_{\in}(\alphabet,\width) = \{(\aodd,s)\;:\; \aodd\in \odddefiningsetstar{\alphabet}{\width},\; s\in \oddlang{\aodd}\}$ is regular (Proposition 6.3  of \cite{deOliveiraOliveira2019width}). 
Second, that the relation $\mathcal{R}_{\subseteq}(\alphabet,\width) = 
\{(\aodd,\aodd')\;:\; \aodd,\aodd'\in\odddefiningsetstar{\alphabet}{\width}, \oddlang{\aodd}\subseteq \oddlang{\aodd'}\}$ is 
regular (Proposition 6.6 of \cite{deOliveiraOliveira2019width}). Similar observations
have been used in \cite{Kuske21} to study second-order finite automata using the framework
of the theory of automatic structures \cite{blumensath2000automatic,baranyi2011automata,khoussainov1994automatic}.
In particular, some of our decidability and closure results have been rederived in \cite{Kuske21} using this framework,
and some new applications of second-order finite automata to partial-order theory have been obtained. 

Jain, Luo and Stephan have introduced the notion of automatic indexed classes of languages as a tool to address some
problems in computational learning theory \cite{jain2012learnability}. An indexed class of languages 
$\{L_{\alpha}\;:\; \alpha\in I\}$ is said to be {\em automatic} if the relation 
$E=\{(\alpha,x)\;:\; x\in L_{\alpha},\; \alpha\in I\}$ is automatic. The fact that $\mathcal{R}_{\in}(\alphabet,\width)$
is regular immediately implies that any regular-decisional class of languages corresponds to an 
automatic class of languages. Indeed, given a second-order finite automaton $\finiteautomaton$, 
the second language of $\finiteautomaton$, $\automatonlanghigher{\finiteautomaton}{2} = \{\oddlang{\aodd}\;:\;\aodd \in \automatonlang{\finiteautomaton}\}$, is an automatic class of languages where each $\aodd\in \automatonlang{\finiteautomaton}$ 
is regarded as an index, and $\oddlang{\aodd}$ is regarded as the language indexed by $\aodd$. 
Henning Fernau conjectured that if $\{L_{\alpha}\;:\; \alpha\in I\}$ is an automatic class of languages 
where each index $\alpha\in I$ is a finite string and all strings in $L_{\alpha}$ have the same length, 
then this class is regular-decisional (i.e. is equal to the second language of some second-order finite automaton).
This conjecture has recently been confirmed by Kuske in \cite{Kuske21}. Similar connections can be established 
with the framework of uniform classes of automatic structures \cite{abu2017advice}, which are defined 
with basis on the notion of automatic structures with advice. In this context, an ODD $\aodd$ may be regarded
as an advice string, while the language $\oddlang{\aodd}$ may be regarded as the set of strings associated 
with the advice $\aodd$. This point of view is particularly relevant when ODDs are used to represent relations, 
as done for instance in \cite{deOliveiraOliveira2019width}. 

In view of the connections discussed above, our framework provides
a suitable parameterization for problems arising in the realm of the theory of automatic classes of languages 
\cite{jain2012learnability} and in the realm of the theory of uniformly automatic classes of structures \cite{abu2017advice}.
The intuition is that the size of the representation for the whole class of languages/structures 
(i.e, the size of the second-order finite automaton given at the input) is completely dissociated from the 
complexity of the languages/structures being represented in the class (i.e. the ODD-width $\width$ necessary to 
represent languages/structures in the class). Since the concepts in \cite{jain2012learnability,abu2017advice} have 
applications in the fields of learning theory \cite{holzl2018learning,jain2012automatic,case2014automatic,howar2018active}
and algebra \cite{phdthesisAutomatic,abu2017advice,colcombet2007transforming,kartzow2013structures}, an interesting line 
of research would be the investigation of potential applications of our fixed-parameter tractable algorithms to 
problems in these fields.

{\small
\paragraph{\bf Acknowledgements.} We thank Henning Fernau and Dietrich Kuske for interesting discussions at CSR 2020.
Alexsander A. de Melo acknowledges support 
from the Brazilian agencies CNPq/GD 140399/2017-8 and CAPES/PDSE 88881.187636/2018-01. 
Mateus de O. Oliveira acknowledges support from the Trond Mohn Foundation and from the Research Council of Norway (Grant Nr. 288761).
}

\bibliographystyle{abbrv}
\bibliography{secondOrderFAs}

\end{document}